\documentclass[12pt,letterpaper]{article}
\usepackage{amsmath}
\usepackage{graphicx}
\usepackage{enumerate}
\usepackage{natbib}
\usepackage{url}
\usepackage{booktabs} % for "middlerule" in tables
\usepackage{amssymb,amsfonts,amsthm}
\usepackage{graphicx,mathtools,bm,amsmath}%,xcolor}
\usepackage{hhline} % for tables
\usepackage{algpseudocode}
\usepackage{algorithm}

\newtheorem{proposition}{Proposition}
\newtheorem{lemma}{Lemma}
\theoremstyle{definition}

\DeclareMathOperator{\median}{\mbox{median}}
\DeclareMathOperator{\MD}{\mbox{MD}}
\DeclareMathOperator{\RD}{\mbox{RD}}
\newcommand{\Cov}{\mbox{Cov}}
\newcommand{\obj}{\mbox{Obj}}
\DeclareMathOperator*{\argmin}{\mbox{argmin}}
\DeclareMathOperator*{\argmax}{\mbox{argmax}}
\DeclareMathOperator*{\tr}{trace}
\DeclareMathOperator*{\diag}{diag}
\newcommand{\tEM}{{\mbox{\tiny EM}}}

\newcommand{\penalt}{q} % penalty
\newcommand{\bpenalt}{\boldsymbol{\penalt}}
\newcommand{\eig}{\lambda} % eigenvalue
\newcommand{\lM}{\eig_{\mbox{\tiny max}}}
\newcommand{\lm}{\eig_{\mbox{\tiny min}}}
\newcommand{\dwi}{d^{(\bw_i)}}
\newcommand{\dw}{d^{(\bw)}}
\newcommand{\Dmax}{D}
\newcommand{\g}{g} % in objective function
\newcommand{\tg}{\tilde{\g}}
\newcommand{\G}{G} % objective function

\newcommand{\bzero}{\boldsymbol{0}}
\newcommand{\bone}{\boldsymbol{1}}
\newcommand{\bu}{\boldsymbol{u}}
\newcommand{\bv}{\boldsymbol{v}}
\newcommand{\bw}{\boldsymbol{w}}
\newcommand{\bx}{\boldsymbol{x}}
\newcommand{\by}{\boldsymbol{y}}
\newcommand{\bz}{\boldsymbol{z}}

\newcommand{\bA}{\boldsymbol{A}}
\newcommand{\bB}{\boldsymbol{B}}
\newcommand{\bC}{\boldsymbol{C}}
\newcommand{\bI}{\boldsymbol{I}}

\newcommand{\bW}{\boldsymbol{W}}
\newcommand{\bX}{\boldsymbol{X}}

\newcommand{\bmu}{\boldsymbol{\mu}}
\newcommand{\bSigma}{\boldsymbol{\Sigma}}

\newcommand{\hx}{\widehat{x}}
\newcommand{\hmu}{\widehat{\mu}}
\newcommand{\hSigma}{\widehat{\Sigma}}

\newcommand{\tw}{\widetilde{w}}
\newcommand{\tL}{\widetilde{L}}

\newcommand{\bhx}{\boldsymbol{\widehat{x}}}
\newcommand{\bhmu}{\boldsymbol{\widehat{\mu}}}
\newcommand{\bhSigma}{\boldsymbol{\widehat{\Sigma}}}

\newcommand{\Btw}{\boldsymbol{\widetilde{w}}}
\newcommand{\btz}{\boldsymbol{\widetilde{z}}}
\newcommand{\BtW}{\boldsymbol{\widetilde{W}}}
\newcommand{\btX}{\boldsymbol{\widetilde{X}}}

\newcommand{\blind}{0}

\begin{document}

\def\spacingset#1{\renewcommand{\baselinestretch}%
{#1}\small\normalsize} \spacingset{1}

%%%%%%%%%%%%%%%%%%%%%%%%%%%%%%%%%%%%%%%%%%%%%%%%%%%%%%%%%%%

\if0\blind
{
  \title{{\bf The Cellwise Minimum Covariance
	       Determinant Estimator}\thanks{To
          appear, {\it Journal of the American
          Statistical Association}.}}
  \author{Jakob Raymaekers\\
    {\normalsize Department of Quantitative Economics,
	  Maastricht University, The Netherlands}\\
    {\normalsize \phantom{and}} \\
    Peter J. Rousseeuw\thanks{Corresponding author, 
    \texttt{peter@rousseeuw.net}\,.}\\
    {\normalsize Section of Statistics and Data Science, 
	  University of Leuven, Belgium}\\ \\}
	\date{November 15, 2023} 
  \maketitle
} \fi

\if1\blind
{
  \phantom{abc}
  \vspace{3.5cm}
	%\bigskip
  %\bigskip
  %\bigskip
  \begin{center}
    {\LARGE\bf The Cellwise Minimum Covariance}
		
	\vspace{0.5cm}
		
		{\LARGE\bf Determinant Estimator}
	\vskip1.5cm
	{\large \today}	
	% {\large July 27, 2022}
\end{center}
	% \bigskip
	\vspace{2cm}
} \fi

\begin{abstract}
The usual Minimum Covariance Determinant 
(MCD) estimator of a covariance matrix is 
robust against casewise outliers. 
These are cases (that is, rows of the data 
matrix) that behave differently from the 
majority of cases, raising suspicion that 
they might belong to a different population.
On the other hand, cellwise outliers are
individual cells in the data matrix.
When a row contains one or more outlying
cells, the other cells in the same row  
still contain useful information that we
wish to preserve.
We propose a cellwise robust version of 
the MCD method, called cellMCD. 
Its main building blocks are observed 
likelihood and a penalty term on the 
number of flagged cellwise outliers.
It possesses good breakdown properties.
We construct a fast algorithm for cellMCD 
based on concentration steps (C-steps)
that always lower the objective.
The method performs well in simulations
with cellwise outliers, and has high 
finite-sample efficiency on clean data.
It is illustrated on real data with
visualizations of the results.
\end{abstract}

\vspace{0.5cm} 

\noindent%

\vfill
\newpage
\spacingset{1.45} % DON'T change the spacing!

%%%%%%%%%%%%%%%%%%%%%%%%%%%%%%%%%%%%%%%%%%%%%
\section{Motivation}
\label{sec:intro}

Any practicing statistician or data scientist
knows that real data sets often contain outliers.
One definition of outliers says that they are
cases that do not obey the fit suggested by the 
majority of the data, which raises suspicion 
that they may have been generated by a different 
mechanism.
Since cases typically correspond to rows of the 
data matrix, they are sometimes called rowwise 
outliers.
They may be the result of gross errors, but
they can also be nuggets of valuable information.
In either case, it is important to find them.
In computer science this is called anomaly
detection, and in some areas it is known as
exception mining. 
In statistics several approaches were tried, 
such as testing for outliers and the 
computation of outlier diagnostics.
In our experience the approach working best is
that of robust statistics, which aims to fit 
the majority of the data first, and then flags 
outliers by their large deviation from that fit. 

In this paper we focus on single-class
multivariate numerical data without a response 
variable (although the results are relevant
for classification and regression too).
The goal is to robustly estimate the central 
location of the point cloud as well as its 
covariance matrix, and at the same time  
flag the outliers that may be present.
The underlying model is that the data come
from a multivariate Gaussian distribution,
after which some data has been replaced by
outliers that can be anywhere.

The Minimum Covariance Determinant (MCD) 
estimator introduced by
\cite{R1984LMS,R1985multivariate}
is highly robust to casewise outliers.
Its definition is quite intuitive.
Take an integer $h$ that is at least half
the sample size $n$.
We then look for the subset containing $h$
cases such that the determinant of its 
usual covariance matrix is as small as
possible. 
The resulting robust location estimate is
then the mean of that subset, and the
robust covariance matrix is its covariance
matrix multiplied by a consistency factor.
One can show that the estimates are not
overly affected when there are fewer than 
$n-h$ outlying cases.
The MCD became computationally feasible
with the algorithm of \cite{fastMCD1999},
followed by even faster algorithms by
\cite{hubert2012DetMCD} and
\cite{DeKetelaere:RT-DetMCD}.
\cite{copt:missingsMCD} computed the
MCD for incomplete data.
The MCD has also been generalized to high
dimensions \citep{MRCD2020}, and to
non-elliptical distributions using
kernels \citep{KMRCD2021}.
For a survey on the MCD and its 
applications see \cite{hubertWIRE-MCD2}.
The MCD is available in the procedure 
ROBUSTREG in SAS, in SAS/IML, in Matlab's
PLS Toolbox, 
and in the R packages {\it robustbase} 
\citep{robustbase} and {\it rrcov} 
\citep{todorov-rrcov} on CRAN. 
In Python one can use \texttt{MinCovDet}
in \texttt{scikit-learn}
\citep{pedregosa2011scikit}.

In recent times a different outlier
paradigm has gained prominence, that of
{\it cellwise outliers}, first published
by \cite{alqallaf2009}. 
It assumes that the data were generated
from a certain distributional model, after
which some individual cells (entries) were 
replaced by other values.
The difference between the casewise and the
cellwise paradigm is illustrated in
Figure~\ref{fig:bycasebycell}.
In the left panel the outlying cases are 
shown as black rows.
In the panel on the right the cellwise
outliers correspond to fewer black squares
in total, but together they contaminate 
over half of the cases, so the existing 
methods for casewise outliers may fail.

\begin{figure}[!ht]
\centering
\vspace{0.3cm}
\includegraphics[width=11cm]{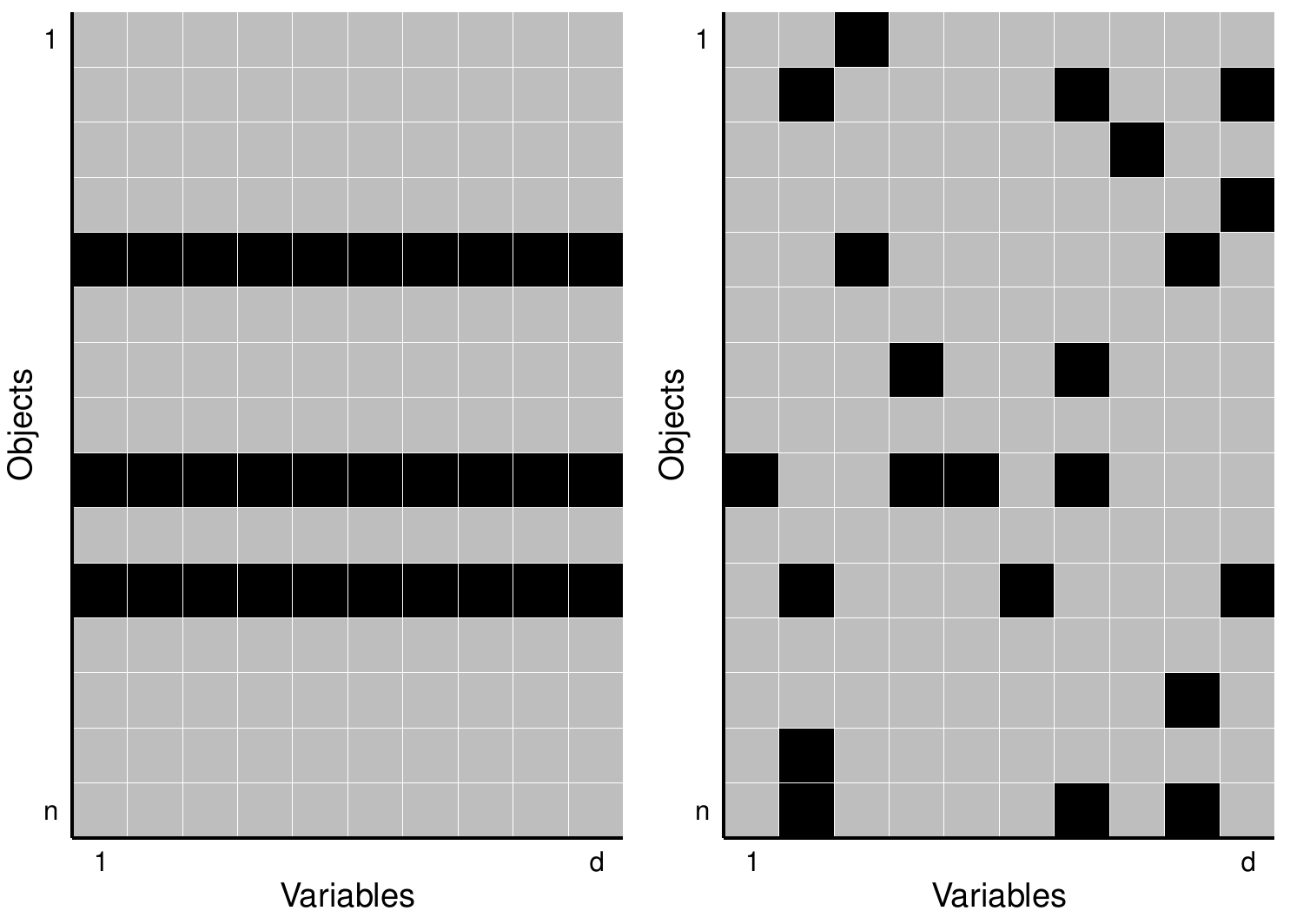}
\caption{The casewise (left) and cellwise 
         (right) outlier paradigms. 
				 (Black means outlying.)}
\label{fig:bycasebycell}
\end{figure}

In reality we do not know in advance 
{\it which} cells in the right panel of
Figure~\ref{fig:bycasebycell} are outlying
(black), unlike the simpler problem of 
incomplete data where we do know which
cells are missing.
When the variables have substantial 
correlations, the cellwise outliers need not 
be marginally outlying, and then it can be 
quite hard to detect them.
\cite{VanAelst2011} proposed one of the first 
detection methods.
\cite{DDC2018} predict the values of all 
cells and flag the cells that differ much 
from their prediction.

There has been some work on 
estimating the underlying covariance matrix 
in the presence of cellwise outliers. 
One approach is to compute robust 
covariances between each pair of variables, 
and to assemble them in a matrix.
To estimate these pairwise covariances, 
\cite{ollerer2015} and \cite{croux2016} use
rank correlations.
\cite{tarr2016} instead use the 
robust pairwise correlation estimator of 
\cite{GK1972} in combination with the 
robust scale estimator $Q_n$ of 
\cite{Qn1993}. 
As the resulting matrix is not necessarily 
positive semidefinite (PSD), they then compute 
the nearest PSD matrix by the method of 
\cite{Higham2002}. 
\cite{raymaekers2021fast} obtain a PSD  
covariance matrix by transforming 
(`wrapping') the original data variables.

Many cellwise robust methods were developed 
for settings such as principal 
components \citep{hubert2019macropca},
discriminant 
analysis \citep{aerts2017cellwise},
clustering \citep{garcia2021cluster},
graphical models \citep{katayama2018robust},
low-rank approximation 
\citep{maronna2008robust}, regression 
\citep{ollerer2016shooting, filzmoser2020cellwise},
and variable selection \citep{su2021robust}.
Also, isolated outliers in functional
data \citep{Hubert:MFOD} can be seen as 
cellwise outliers.

In the next section we introduce the cellwise 
MCD estimator. 
It is the first method with a
single objective that combines detection and
estimation, unlike some existing methods
which do detection and estimation
separately. Because of this cellMCD has 
provable cellwise breakdown properties, 
see section~\ref{sec:theory}.
There we also derive its consistency.
Section~\ref{sec:algo} describes its 
algorithm, and proves that it converges.
It is faster than the earlier methods.
Some illustrations on real data are shown 
in section~\ref{sec:examples}. 
The performance of the method is studied by 
simulation in section~\ref{sec:simul},
indicating that it is very robust 
against adversarial contamination. 
Section~\ref{sec:concl} concludes with a
discussion.

%%%%%%%%%%%%%%%%%%%%%%%%%%%%%%%%%%%%%%%%%%%%%
\section{A cellwise MCD}
\label{sec:cellMCD}

We first note that the {\it casewise} MCD can 
be reformulated in terms of likelihood.
The likelihood of a $d$-variate
Gaussian distribution is
\begin{equation} \label{eq:fcase}
f(\bx,\bmu,\bSigma)=
   \frac{1}{(2\pi)^{d/2}|\bSigma|^{1/2}}
	 e^{\displaystyle
	 -\MD^2(\bx,\bmu,\bSigma)/2}
\end{equation}
where $\bmu$ is a column vector, $\bSigma$ 
is a positive definite matrix, and 
the Mahalanobis distance is
$\MD(\bx,\bmu,\bSigma) = \sqrt{
(\bx-\bmu)^\top \bSigma^{-1}(\bx-\bmu)}$.
For a sample $\bx_1,\ldots,\bx_n$ we
put $L(\bx_i,\bmu,\bSigma) :=
-2\ln(f(\bx_i,\bmu,\bSigma))$ so the 
maximum likelihood estimator (MLE) of
$(\bmu, \bSigma)$ minimizes
\begin{equation} \label{eq:MLEcase}
  \sum_{i=1}^{n} L(\bx_i,\bmu,\bSigma) =
  \sum_{i=1}^{n} \big( \ln |\bSigma| +
	 d\ln(2\pi) +	
	\MD^2(\bx_i,\bmu,\bSigma)\,\big)\;.
\end{equation}
Let us now look for a subset 
$H \subset \{1,...,n\}$ with $h$ elements
which minimizes \eqref{eq:MLEcase} where the
sum is only over $i$ in $H$. 
We can also write this with weights $w_i$ that
are 0 or 1 in the objective
$\sum_{i=1}^{n} w_i L(\bx_i,\bmu,\bSigma)$, 
so we minimize
\begin{equation} \label{eq:MCDcase}
\begin{aligned}
   \sum_{i=1}^{n} w_i \big(
	 \ln |\bSigma| + d\ln(2\pi) +
   \MD^2(\bx_i,\bmu,\bSigma)\,\big)\\
	\mbox{under the constraint that }
	\sum_{i=1}^n{w_i} = h\;.
\end{aligned}
\end{equation}
For the minimizing set of weights $w_i$ we
know from maximum likelihood that $\bhmu$ is 
the mean of the $\bx_i$ in $H$, so it is the
weighted mean of all $\bx_i$\,, and similarly
\begin{equation} \label{eq:wcov}
  \bhSigma = \frac{1}{h} \sum_{i=1}^n 
  w_i(\bx_i - \bhmu)(\bx_i - \bhmu)^\top\;.
\end{equation}	
But then the third term of \eqref{eq:MCDcase}
becomes
$$ \sum_{i=1}^{n} {w_i(\bx_i-\bhmu)^\top 
	  \bhSigma^{-1}(\bx_i-\bhmu)} =
  \sum_{i=1}^{n} {\tr\Big(w_i(\bx_i-\bhmu)
	  (\bx_i-\bhmu)^\top \bhSigma^{-1}\Big)} = $$
$$ \tr\Big(\sum_{i=1}^{n} w_i(\bx_i-\bhmu)
    (\bx_i-\bhmu)^\top \bhSigma^{-1}\Big) = 
	 \tr\big(h \bhSigma \bhSigma^{-1}\big) = hd$$
which is constant, and so is the second term. 
Therefore minimizing \eqref{eq:MCDcase} is 
equivalent to minimizing the determinant of
\eqref{eq:wcov}, which is the definition 
of the casewise MCD.

In the context of incomplete data, 
\citet{dempster1977} and others defined the 
{\it observed likelihood}.
Let us denote the missingness pattern of the 
$n\times d$ data matrix $\bX$ by the 
$n \times d$ matrix $\bW$ with entries 
$w_{ij}$ that are 0 for missing $x_{ij}$ 
and 1 otherwise. 
Its rows $\bw_i$ take the place of the 
scalar weights $w_i$ in \eqref{eq:MCDcase}.
For the Gaussian model the observed likelihood of
the $i$th observation \citep{LR2020} is given by:
\begin{equation}
  \displaystyle f(\bx_i^{(\bw_i)},\bmu^{(\bw_i)},
  \bSigma^{(\bw_i)}) :=
	\frac{1}{(2\pi)^{\dwi/2}
	|\bSigma^{(\bw_i)}|^{1/2}} e^{\displaystyle 
	-\MD^2(\bx_i,\bw_i,\bmu,\bSigma)/2}
\end{equation}
in which
\begin{equation}
  \MD(\bx_i,\bw_i,\bmu,\bSigma) :=
  \sqrt{(\bx_i^{(\bw_i)}-\bmu^{(\bw_i)})^\top
	(\bSigma^{(\bw_i)})^{-1}
	(\bx_i^{(\bw_i)}-\bmu^{(\bw_i)})}
\end{equation}
is called the {\it partial Mahalanobis distance} 
by Danilov et al. (2012).
Here $\bx_i^{(\bw_i)}$ is the vector with only 
the entries for which $w_{ij}=1$, and similarly 
for $\bmu^{(\bw_i)}$.
The matrix $\bSigma^{(\bw_i)}$ is the submatrix 
of $\bSigma$ containing only the rows and columns 
of the variables $j$ with $w_{ij}=1$. 
Finally, $\dwi$ is the dimension of 
$\bx_i^{(\bw_i)}$, i.e. the number of 
non-missing entries of $\bx_i$\,. 
By convention, a case $\bx_i$ consisting 
exclusively of NA's has $\dwi=0$, 
$\MD(\bx_i,\bw_i,\bmu,\bSigma)=0$ and 
$|\bSigma^{(\bw_i)}|=1$.
Putting $L(\bx_i,\bw_i,\bmu,\bSigma) :=
-2\ln(f(\bx_i,\bw_i,\bmu,\bSigma))$ we see that
maximizing the observed likelihood of the 
entire data set comes down to minimizing
\begin{equation} \label{eq:cellMLE}
  \sum_{i=1}^{n} L(\bx_i,\bw_i,\bmu,\bSigma) =
  \sum_{i=1}^{n} \big( \ln |\bSigma^{(\bw_i)}|
	+ \dwi\ln(2\pi) +
	\MD^2(\bx_i,\bw_i,\bmu,\bSigma)\,\big)\;\;.
\end{equation}
This maximum likelihood estimate of 
$(\bmu,\bSigma)$ is typically computed 
by the EM algorithm \citep{dempster1977}. 

When constructing a cellwise MCD, the matrix 
$\bW$ now describes which cells are flagged: a flagged 
cell $x_{ij}$ gets $w_{ij}=0$. The notations 
$\bx_i^{(\bw_i)}$, $\dwi$, $\bmu^{(\bw_i)}$, and 
$\bSigma^{(\bw_i)}$ are interpreted analogously. 
The matrix $\bW$ is not given in advance, but will 
be obtained through the estimation procedure.
Now $h$ can no longer apply to the number of 
unflagged {\it cases}. Instead, we apply it to
the number of unflagged {\it cells} per 
column. We could minimize
\begin{equation} \label{eq:tempMCD}
\begin{aligned}
  \sum_{i=1}^{n} \Big( 
	\ln |\bSigma^{(\bw_i)}| + 
	\dwi\ln(2\pi) + 
	\MD^2(\bx_i,\bw_i,\bmu,\bSigma)\, \Big)\\
 \mbox{ under the constraints }
 \eig_d(\bSigma) \geqslant a
 \mbox{ and }
||\bW_{.j}||_0 \geqslant h
 \mbox{ for all } j=1,\ldots,d
\end{aligned}
\end{equation}
over $(\bmu,\bSigma,\bW)$.
The first constraint says that the smallest
eigenvalue of $\bSigma$ is at least as 
large as a number $a > 0$, where the 
eigenvalues of $\bSigma$ are denoted as 
$\eig_1(\bSigma) \geqslant
\ldots \geqslant \eig_d(\bSigma)$. 
This ensures that $\bSigma$ is nonsingular,
which is required to compute Mahalanobis 
distances.
In the second constraint, 
$||\bW_{.j}||_0$ is the number of nonzero 
entries in the $j$-th column of $\bW$.
Note that we should not choose $h$ too low.
Whereas for the casewise MCD we can take $h$
as low as $0.5n$, that would be ill-advised 
here because it could happen that two 
variables $j$ and $k$ do not overlap in 
the sense that $w_{ij}w_{ik}=0$ for all $i$, 
making it impossible to estimate their 
covariance. We will impose that 
$h \geqslant 0.75n$ throughout. 

However, minimizing \eqref{eq:tempMCD} 
typically treats too many cells as outlying.
This is because a value of $h$ that is suitable 
for one variable may be too low for another, 
and we do not know ahead of time which variables 
have many outlying cells and which have few or
none.
To avoid flagging too many cells, we add a 
penalty counting the number of flagged cells 
in each column. The objective function of the 
{\bf cellwise MCD} (cellMCD) then becomes 
\begin{equation} \label{eq:cellMCD}
\begin{aligned}
  \sum_{i=1}^{n} {\Big(
	\ln |\bSigma^{(\bw_i)}| + 
	\dwi\ln(2\pi) +
	\MD^2(\bx_i,\bw_i,\bmu,\bSigma)\,\Big)} + 
	\sum_{j=1}^d \penalt_j
	||\bone_d - \bW_{.j}||_0\\
	\mbox{ under the constraints }
 \eig_d(\bSigma) \geqslant a
 \mbox{ and }
 ||\bW_{.j}||_0 \geqslant h
 \mbox{ for all } j=1,\ldots,d\,.
\end{aligned}
\end{equation}
The notation $||\bone_d - \bW_{.j}||_0$ 
stands for the number of nonzero elements in 
this vector, so the number of zero weights
in column $j$ of $\bW$, i.e. the number of 
flagged cells in column $j$ of $\bX$.
The constants $\penalt_j$ for 
$j=1,\ldots,d$ are computed  
(in Section~\ref{sec:algo}) 
from the desired percentage of flagged 
cells in the absence of contamination.
At the same time we keep the robustness
constraint that $||\bW_{.j}||_0 \geqslant h$.
Combining a penalty term with a $||.||_0$ 
constraint is not new, see the work of 
\cite{she2021progq} on casewise 
robust regression. 
In our context, the constraint 
$||\bW_{.j}||_0 \geqslant h$ will ensure
the robustness of the estimator (through
Proposition~\ref{prop:bdv} below), 
whereas the penalty term $\sum_j
\penalt_j ||\bone_d - \bW_{.j}||_0$
discourages flagging too many cells, 
which improves the estimation accuracy
at clean data as seen in simulations.

The cellMCD method is the first cellwise
robust technique that combines the fitting 
of the parameters and the flagging of 
outlying cells $(\bW)$ in one objective
function.
The constraint $||\bW_{.j}||_0 \geqslant h$
for $j=1,\ldots,d$ says that we require 
at least $h$ unflagged cells in each column.
In order to avoid a singular covariance
matrix, we obviously need $h > d$.
Combining these inequalities we obtain
$n > 4d/3$\,. 
But the curse of dimensionality implies 
that many spurious structures can be found 
in increasing dimensions, so we want a 
more comfortable ratio of cases per 
dimension. 
For the casewise MCD the rule of thumb 
is $n/d \geqslant 5$ \citep{RvZ1990},
and we will require that here too.

The cellMCD method defined by 
\eqref{eq:cellMCD} is equivariant for 
permuting the cases, for shifting the 
data, and for multiplying the variables by 
nonzero constants.
But unlike the casewise MCD it is not
equivariant under general nonsingular linear
transformations, or even orthogonal
transformations.
This is because cells are intimately tied to 
the coordinate system, and an orthogonal
transformation changes the cells. 
This is an important difference between 
the casewise and cellwise approaches. 
For instance, consider the standard 
multivariate Gaussian model in dimension 
$d=4$ with the suspicious point $(10,0,0,0)$.
By an orthogonal transformation of the data,
this point can be moved to 
$(\sqrt{50},\sqrt{50},0,0)$ or to
$(5,5,5,5)$. 
The casewise MCD is equivariant to such
transformations and will still flag the 
same case.
But in the cellwise paradigm $(10,0,0,0)$ 
has one outlying cell, 
$(\sqrt{50},\sqrt{50},0,0)$ has two, and
$(5,5,5,5)$ has four, so cellMCD will 
react differently, as it should.

%%%%%%%%%%%%%%%%%%%%%%%%%%%%%%%%%%%%%%%%%%%%
\section{Theoretical properties}
\label{sec:theory}

\cite{alqallaf2009} define the cellwise 
breakdown value of a location estimator.
Here we will focus on finite-sample breakdown
values in the sense of \cite{donoho1983}
and \cite{lopuhaa1991breakdown}.
The {\it finite-sample cellwise breakdown
value} of an estimator $\bhmu$ at a dataset
$\bX$ is given by the smallest fraction of
cells per column that need to be replaced to 
carry the estimate outside all bounds.
Formally, let $\bX$ be a dataset of size $n$, 
and denote by $\bX^m$ any corrupted sample 
obtained by replacing at most $m$ cells in 
each column of $\bX$ by arbitrary values. 
Then the finite-sample cellwise 
breakdown value of a location estimator 
$\bhmu$ at $\bX$ is given by 
\begin{equation} \label{eq:bdvloc}
  \varepsilon^*_n(\bhmu, \bX)=
  \min \left\{\frac{m}{n}:\;
	\sup_{\bX^m}{\left|\left|\bhmu(\bX^m) - 
	\bhmu(\bX)\right|\right|} = 
	\infty\right\}.
\end{equation}
Analogously to the casewise setting, we can 
also define the {\it cellwise explosion 
breakdown value} of a covariance estimator 
$\bhSigma$ as
\begin{equation} \label{eq:explosion}
  \varepsilon^+_n(\bhSigma, \bX)=
  \min \left\{\frac{m}{n}:\;
	\sup_{\bX^m}\eig_1(\bhSigma) = 
	\infty\right\}.
\end{equation}
Moreover, we define the {\it cellwise 
implosion breakdown value} of $\bhSigma$ as
\begin{equation} \label{eq:implosion}
  \varepsilon^-_n(\bhSigma, \bX)=
	\min \left\{\frac{m}{n}:\;
	\inf_{\bX^m}\eig_d(\bhSigma) = 0\right\}.
\end{equation}

The definitions of the corresponding 
casewise breakdown values are very similar,
the only difference being that the 
corrupted samples, let us call them 
$\btX^m$, are obtained by replacing at 
most $m$ rows of $\bX$ by arbitrary rows.
If we denote the casewise breakdown values
by $\delta^*_n$, $\delta^+_n$ and 
$\delta^-_n$ we can formulate the following 
simple but useful result:
\begin{proposition} \label{prop:cellbdv}
For all estimators $\bhmu$ and $\bhSigma$
at any dataset $\bX$ it holds that\linebreak
 $\varepsilon^*_n(\bhmu, \bX) 
  \leqslant \delta^*_n(\bhmu, \bX)$, 
 $\varepsilon^+_n(\bhSigma, \bX)
  \leqslant \delta^+_n(\bhSigma, \bX)$, 
 and $\varepsilon^-_n(\bhSigma, \bX)
  \leqslant \delta^-_n(\bhSigma, \bX)$.
\end{proposition}

The proof consists of realizing that 
the casewise contaminated samples $\btX^m$
can be seen as cellwise contaminated
samples $\bX^m$. It is thus generally
true that the cellwise breakdown value is 
less than or equal to the casewise
breakdown value. Therefore, all upper
bounds on casewise breakdown values in the 
literature also hold for cellwise
breakdown values.

When proving breakdown values one often
assumes that the original
data set $\bX$ is in {\it general position},
meaning that no more than $d$ points lie in
any $d-1$ dimensional affine subspace.
In particular, no three points lie on a 
line, no 4 points lie on a plane, and so on.
When the data are drawn from a continuous
distribution, it is in general position
with probability 1. 
Real data have a limited precision, so 
they are not always in general position.

The inequalities in 
Proposition~\ref{prop:cellbdv} can be
strict. For instance,
the casewise implosion breakdown value of
the classical covariance matrix {\bf Cov}
at a dataset in general position is very 
high, in fact it is $(n-d)/n$ which goes 
to 1 for increasing sample size $n$.
This is because whenever $d+1$ of the
original data points are kept, {\bf Cov}
remains nonsingular.
In stark contrast, its
{\it cellwise} implosion breakdown
value is quite low:
\begin{equation} \label{eq:cov}
  \varepsilon^-_n(\mbox{\bf Cov}, \bX) =
	\left\lceil \frac{n-d}{d} 
	\right\rceil/n\, \leqslant \, 
	\frac{1}{d}\;.
\end{equation}
To see why, let us pick $d$ points of
$\bX$ which lie on a hyperplane that is
not parallel to any coordinate axis.
In the remaining $n-d$ rows we can then
replace a single cell such that all of
the resulting points lie on the same 
hyperplane, so {\bf Cov} becomes singular.
We can do this by replacing no more than
$\lceil (n-d)/d \rceil$ cells in each 
variable, which is a fraction
$\lceil (n-d)/d \rceil/n$ of its $n$ cells.

\cite{challenges} recently derived 
a similar upper bound for all affine 
equivariant estimators $\bhSigma$.
In order to obtain a higher cellwise
breakdown value we are thus forced to
leave the realm of affine equivariance.
In fact, the constraint
$\eig_d(\bhSigma) \geqslant a > 0$ in the 
definition \eqref{eq:cellMCD} of cellMCD
is not affine invariant, but it keeps
$\bhSigma$ from imploding. Therefore 
the cellwise implosion breakdown value 
of cellMCD is 1. 

We also want to know the breakdown value
of its location estimate $\bhmu$ and the
explosion breakdown value of $\bhSigma$.
These naturally depend on the choice of $h$.

\begin{proposition} \label{prop:bdv}
If the dataset $\bX$ is in general position
and $h \geqslant\lfloor\frac{n}{2}\rfloor+1$, 
the cellMCD estimators $\bhmu$ and 
$\bhSigma$ satisfy the properties
\begin{itemize}
\item[(a)] $\varepsilon^-_n(\bhSigma, \bX) 
            = 1$
\item[(b)] $\varepsilon^+_n(\bhSigma, \bX)
            \geqslant (n-h+1)/n$ 
\item[(c)] $\varepsilon^*_n(\bhmu, \bX) 
            \geqslant (n-h+1)/n$ 
\item[(d)] The lower bound $(n-h+1)/n$ 
           is sharp.
\end{itemize}
\end{proposition}

Proposition~\ref{prop:bdv} shows that cellMCD
is highly robust. Its proof is in Section A.1 
of the Supplementary Material. By
Proposition~\ref{prop:cellbdv}, it follows 
that these lower bounds also hold for the
casewise breakdown values. This also implies 
that the method works on a mix of cellwise
and casewise outliers as well.
We do not actually recommend to choose $h$ as
low as the proposition allows: as explained 
before this could lead to some poorly defined 
covariances and numerical instability. 
We stick with our earlier recommendation of 
$h \geqslant 0.75n$, and in fact $h = 0.75n$ 
is the default in our implementation.

Let us now turn to the asymptotic behavior 
of cellMCD. At the 
uncontaminated model distribution and for
large $n$ only a small fraction of cells 
is actually discarded, due to our choice of the 
constants $q_j$ in the penalty term. 
In that situation the large-sample behavior of
cellMCD is therefore the same as without the 
columnwise constraint on $\bW$.
The cellMCD objective can then be written as
\begin{equation}\label{eq:obj1}
\G(\mu, \Sigma, F) \coloneqq \int \g_{\mu, \Sigma}(x) F(\mathrm{d}x)    
\end{equation}
where
\begin{equation}\label{eq:obj2}
    \g_{\mu,\Sigma}(x) \coloneqq \min_{w\in \{0,1\}^d}{\left\{\ln \left|\Sigma^{(w)}\right|+\dw\ln(2\pi) + \text{MD}^2(x,w,\mu, \Sigma)+\bpenalt\,(\bone - w)^{\top} \right\}}
\end{equation}
in which $\bpenalt = (q_1,\ldots,q_d)$ and $w = (w_1,\ldots,w_d)$.
The cellMCD estimate is then 
$$\argmin_{\left(\mu, \Sigma\right) \in \Theta} \G(\mu, \Sigma, F_n)$$
with $F_n$ the empirical distribution and $\Theta$ the parameter space
of $(\mu, \Sigma)$, which incorporates the condition
$\eig_d(\Sigma) \geqslant a$.
Denote the set of minimizers as $\Theta^*$.
In section A.2 of the Supplementary Material the following
Wald-type consistency result is shown, using work 
of~\cite{VdVaart2000}:
\begin{proposition}\label{prop:consistency}
Let $(\hat{\mu}_n, \hat{\Sigma}_n)$ be a sequence of 
estimators which nearly minimize $\G(\cdot, \cdot, F_n)$
in the sense that 
$\G(\hat{\mu}_n, \hat{\Sigma}_n, F_n) \leqslant 
\G(\mu^*, \Sigma^*, F_n)+o_P(1)$ for some 
$(\mu^*, \Sigma^*) \in \Theta^*$. Then it holds for all 
$\varepsilon > 0$ that
$$P(\,\Dmax((\hat{\mu}_n, \hat{\Sigma}_n),\Theta^*) 
    \geqslant \varepsilon \,)\rightarrow 0\,, $$
where $\Dmax((\hat{\mu}_n,\hat{\Sigma}_n),
(\mu^*,\Sigma^*)) := \max(||\hat{\mu}_n-\mu^*||_2, 
 ||\hat{\Sigma}_n - \hSigma^*||_F)$
combines the Euclidean and Frobenius norms.
\end{proposition}

The population minimizer for $\Sigma$ is not quite the 
underlying parameter, since a small fraction of cells is always
given weight zero due to the penalty term in the objective. 
But for the location $\mu$ we can prove that the unique
minimizer is indeed the underlying parameter vector, so
the cellMCD functional for location is Fisher consistent:
\begin{proposition}\label{prop:Fisherconsistency}
Let $F$ be a strictly unimodal elliptical distribution  
with center $\mu$ and a density function.  
For any $\Sigma$, we then have the unique 
$\argmin_{m \in \mathbb{R}^d} \,
  \G(m, \Sigma, F) = \mu\;.$
\end{proposition}

%%%%%%%%%%%%%%%%%%%%%%%%%%%%%%%%%%%%%%%%%%%%
\section{Algorithm}
\label{sec:algo}

In the algorithm we will need the following
result about decomposing the Mahalanobis
distance and the likelihood.

\begin{proposition}\label{prop:split}
Let us split the $d$-variate case $\bx$
into two nonempty blocks, and split $\bmu$ and 
the $d \times d$ positive definite matrix
$\bSigma$ accordingly, like 
\begin{equation*}
\begin{array}{lll}
\bx=\begin{bmatrix}
\bx_1\\
\bx_2
\end{bmatrix} \;\;\;\; &
\bmu = \begin{bmatrix}
\bmu_1\\
\bmu_2
\end{bmatrix} \;\;\;\; &
\bSigma = \begin{bmatrix}
\bSigma_{11} & \bSigma_{12}\\
\bSigma_{21} & \bSigma_{22}
\end{bmatrix} \;.
\end{array}
\end{equation*}
Then $\MD^2(\bx,\bmu,\bSigma) =
(\bx-\bmu)^\top \bSigma^{-1}(\bx-\bmu)$
and $L(\bx,\bmu,\bSigma) =
-2\ln(f(\bx,\bmu,\bSigma))$ satisfy
\begin{equation}\label{eq:addmd}
  \MD^2(\bx,\bmu,\bSigma) =
	\MD^2(\bx_1,\bhx_1,\bC_1) +
	\MD^2(\bx_2,\bmu_2,\bSigma_{22})
\end{equation}
\begin{equation}\label{eq:addlikelihood}
  L(\bx,\bmu,\bSigma) =
	L(\bx_1,\bhx_1,\bC_1) + 
	L(\bx_2,\bmu_2,\bSigma_{22})
\end{equation}
for\; $\bhx_1= \bmu_1 + 
  \bSigma_{12}\bSigma_{22}^{-1}
  (\bx_2 - \bmu_2)$\; 
and\; $\bC_1=\bSigma_{11}-\bSigma_{12}
 \bSigma_{22}^{-1}\bSigma_{21}$\,.
\end{proposition}

The proof can be found in section A.3 in the
Supplementary Material. 
The proposition can be interpreted as follows. 
Take a case $\bx_i$ with some but not all
cells missing, and for simplicity assume 
that its missing components come first. 
Then put $\bx_1=\bx_i^{(\bone-\bw_i)}$ and 
$\bx_2$ the remainder.
If $(\bmu,\bSigma)$ are the true underlying 
parameters, $\bhx_1$ is the conditional 
expectation $E[\bX_1|\bX_2=\bx_2]$ and $\bC_1$ 
is the conditional covariance matrix 
$\Cov[\bX_1|\bX_2=\bx_2]$.
The additivity in \eqref{eq:addmd} and 
\eqref{eq:addlikelihood} justifies the use of the 
partial Mahalanobis distances and the observed 
likelihood in our setting. Moreover, the fact 
that the difference of two `nested' $\MD^2$ is 
again an $\MD^2$ and hence non-negative implies 
that the $\MD^2$ is monotone for nested sets of 
variables. In particular, if $\bx$ is observed 
fully we can write 
\begin{equation}\label{eq:summd}
\begin{array}{lll} 
	& \MD^2(\bx,\bmu,\bSigma)\\
	&= \displaystyle 
	\frac{r^2(x_1|x_2,\ldots,x_d)}
	{s^2(X_1|x_2,\ldots,x_d)}+
	\frac{r^2(x_2|x_3,\ldots,x_d)}
	{s^2(X_2|x_3,\ldots,x_d)}+\cdots
  \displaystyle + \frac{r^2(x_{d-1}|x_d)}
	{s^2(X_{d-1}|x_d)} + 
	\frac{(x_d-\mu_d)^2}{\Sigma_{dd}}
\end{array}
\end{equation}
where each time $s^2$ is the matrix $\bC_1$ 
(which is a scalar here) and the residuals 
are\linebreak
$r(x_1|x_2,\ldots,x_d)=
x_1-\hx_1(x_2,\ldots,x_d)$ and so on.
Note that \eqref{eq:summd} holds for any order 
of the $d$ variables. However, in each order the 
relative contribution of variable $j$ to 
the total $\MD^2(\bx,\bmu,\bSigma)$ may be 
different.
For the likelihood we obtain similarly
\begin{equation} \label{eq:sumlikelihood}
\begin{array}{lll}
  L(\bx,\bmu,\bSigma) &=&
	L(x_{1},\mu_1,C_{1|2,\ldots,d})+
	L(x_{2},\mu_2,C_{2|3,\ldots,d})+
	\cdots \\&&
  +\, L(x_{d-1},\mu_{d-1},C_{d-1|d})+
	L(x_{d},\mu_{d},\Sigma_{dd})
\end{array}
\end{equation}
in which the terms do not need to be positive.

If we set $\penalt_j = 0$ in the 
objective function~\eqref{eq:cellMCD} of  
cellMCD and use casewise weights, i.e. 
casewise constant $w_{ij}$\,, we recover 
the objective function~\eqref{eq:MCDcase} of 
the original casewise MCD. The latter is not 
convex in $\bmu$ and $\bSigma$, so neither
is~\eqref{eq:cellMCD}.
The crucial ingredient in the algorithm 
for the casewise MCD is the concentration 
step (C-step) of \cite{fastMCD1999}.
After each C-step the new objective value is 
less than or equal to the old objective value, 
so iterating C-steps always converges
to a stationary point.
We will now construct a C-step for cellMCD
with the same properties. 
Let us denote the current solution of 
cellMCD by $\bhmu^{(k)}$, $\bhSigma^{(k)}$, 
and $\bW^{(k)}$.
Then the new C-step proceeds as follows.

\noindent {\bf Part (a) of the C-step.} 
In this part we update the matrix $\bW$ in
\eqref{eq:cellMCD} while keeping $\bhmu^{(k)}$ 
and $\bhSigma^{(k)}$ unchanged.
We start the new pattern $\BtW$ 
as $\BtW = \bW^{(k)}$, and then we
modify $\BtW$ column by column, by
cycling over the variables $j=1,\ldots,d$.
The fact that this job can be done by column
is advantageous for maintaining the constraint.
Assume we are working on column $j$ of
$\BtW$, possibly after having modified
other columns of $\BtW$ already.
The current pattern of variable $j$ is
$\BtW_{\cdot j}$ and we want to obtain
a new pattern for column $j$ to reduce the 
objective while leaving the other columns
of $\BtW$ unchanged. 
Note that we can write the objective 
\eqref{eq:cellMCD} as
$\sum_{i=1}^n{\tL(\bx_i,\bw_i,\bmu,
\bSigma,\bpenalt)}$ where
\begin{equation*}
  \tL(\bx_i,\bw_i,\bmu,\bSigma,\bpenalt)
	= \ln |\bSigma^{(\bw_i)}| + 
	\dwi\ln(2\pi) +
	\MD^2(\bx_i,\bw_i,\bmu,\bSigma)\, + 
	\sum_{j=1}^d \penalt_j |1-w_{ij}|
\end{equation*}
with $\bpenalt = (\penalt_1,\ldots,\penalt_d)$.
For each $i=1,\ldots,n$ we compute the
difference in the total objective 
\eqref{eq:cellMCD} between putting
$\tw_{ij} = 1$ and putting $\tw_{ij} = 0$,
which is
\begin{align}\label{eq:deltaij}
  \Delta_{ij} &= \tL(\bx_i,\tw_{ij}=1,
	   \bhmu^{(k)},\bhSigma^{(k)},\bpenalt) -
	   \tL(\bx_i,\tw_{ij}=0,\bhmu^{(k)},
	   \bhSigma^{(k)},\bpenalt) \nonumber \\ 
	&= \ln |\bSigma^{(\tw_{ij}=1)}| -
	   \ln |\bSigma^{(\tw_{ij}=0)}| +
	   \ln(2\pi) +
     \MD^2(x_{ij},\hx_{ij},C_{ij})
		 -\penalt_j \nonumber \\		
	&= \ln(C_{ij}) + \ln(2\pi) +
	   (x_{ij} - \hx_{ij})^2/C_{ij}
		 - \penalt_j
\end{align}
where the second and third equalities use
Proposition \ref{prop:split} in which 
$\hx_{ij}$ and $C_{ij}$ are now scalars.
Note that $\hx_{ij} = \hmu_j^{(k)} + 
 \bhSigma_{j,o}^{(k)} 
 (\bhSigma_{o,o}^{(k)})^{-1}
 (\bhx_{i,o}- \bhmu_{o}^{(k)})$ is the 
conditional expectation of the cell 
$X_{ij}$ conditional on the observed 
(subscript `o') cells in row $i$, i.e. 
those with $\tw_{i\cdot}=1$, taking into 
account any earlier modifications to $\BtW$.
Analogously, $C_{ij} = \bhSigma_{j,j}^{(k)}- 
\bhSigma_{j,o}^{(k)}
(\bhSigma_{o,o}^{(k)})^{-1} 
\bhSigma_{o,j}^{(k)}$ is the conditional
variance of $X_{ij}$\,.  
We now need to minimize 
$\sum_{i=1}^n \tL(\bx_i,\tw_{ij},
\bhmu^{(k)},\bhSigma^{(k)},\bpenalt)$
subject to the constraint
$\sum_{i=1}^n \tw_{ij} \geqslant h$.
If $\Delta_{ij} \leqslant 0$ holds for $h$ 
or more $i$, then the minimum is attained
by setting those $\tw_{ij}$ to 1 and the 
others to 0. If not, it is attained by
setting $\tw_{ij}$ to 1 for the $i$
with the $h$ smallest $\Delta_{ij}$  
and to 0 otherwise.
After cycling through all columns of $\BtW$
we set $\bW^{(k+1)} = \BtW$.

\noindent {\bf Part (b) of the C-step.} Keeping 
the new pattern $\bW^{(k+1)}$ fixed we now want
to update $\bhmu$ and $\bhSigma$. 
As $\bW^{(k+1)}$ is fixed the penalty term 
in \eqref{eq:cellMCD} does
not enter the minimization, so we are in the
situation of the objective \eqref{eq:cellMLE}
for incomplete data, where the EM algorithm can
be used. We first carry out one E-step 
which computes conditional means and products 
for the data entries with $\bW^{(k+1)}_{ij}=0$,
for all rows. 
Next, we carry out an M-step, followed by
imposing the constraint $\eig_d \geqslant a$
by truncating the eigenvalues of $\bhSigma$
from below at $a$.
The C-step ends by reporting $\bW^{(k+1)}$, 
$\bhmu^{(k+1)}$ and $\bhSigma^{(k+1)}$.

\begin{proposition}\label{prop:Cstep}
(i) Each C-step turns a triplet
$(\bhmu^{(k)}, \bhSigma^{(k)}, 
\bW^{(k)})$ satisfying the constraints
in \eqref{eq:cellMCD} into a new triplet
$(\bhmu^{(k+1)}, \bhSigma^{(k+1)}, 
\bW^{(k+1)})$ which satisfies the same
constraints and whose 
objective~\eqref{eq:cellMCD} is less than
or equal to before. (ii) Iterating C-steps
always converges.
\end{proposition}

For the proof see section A.3 in the 
Supplementary Material, which also 
contains the pseudocode of the algorithm.
Many variations of the \mbox{C-step} are 
possible, such as cycling through the 
columns of $\BtW$ in a different order. 
We could also cycle through the columns of 
$\BtW$ more than once in part (a), and/or 
run more than one EM-step in part (b).
But experiments in section A.6
of the Supplementary Material show that 
these changes have a negligible and 
non-systematic effect on estimation 
accuracy, so we stay with the
current version which is the fastest.

Note that cellMCD can still be used when the 
data contains missing cells, indicated by 
$u_{ij}$ which are 0 for missing cells and 1
elsewhere. 
In that situation we first have to remove 
variables with more than $n-h$ missing
values. 
In the C-step it then suffices to force 
$w_{ij}=0$ whenever $u_{ij}=0$.

In order to start our C-steps we need an
initial estimator.
In our experiments we found that the DDCW
estimator of \cite{cellHandler} gives good
results and is very fast. It is a 
combination of the DetectDeviatingCells 
(DDC) method of \cite{DDC2018} and the
fast correlation method in 
\citep{raymaekers2021fast}.
DDCW is described in section A.4 of the 
Supplementary Material. Instead of 
starting from a single initial estimate, 
%$(\bhmu^0,\bhSigma^0)$, 
one could also start from several initial 
estimates. Iterating C-steps from each (with 
the same $\penalt_j$ and $a>0$) until 
convergence, one can then keep the solution 
with the lowest objective~\eqref{eq:cellMCD}.

The only remaining question is how to select
the constants $\penalt_j$ but this is quite 
simple, we do not need cross-validation or 
an information criterion.
In \eqref{eq:deltaij} the term
$(x_{ij} - \hx_{ij})^2/C_{ij}$ is the
square of the residual $x_{ij} - \hx_{ij}$
standardized robustly. 
For inlying cells this should be below a
cutoff, for which we take the chi-squared 
quantile $\chi^2_{1,p}$ with one degree of 
freedom and probability $p$. 
The term $\ln(C_{ij})$ is approximated by
using the conditional variance of variable 
$j$ in the initial estimate $\bhSigma_0$\,, 
given by $C_j := 1/(\bhSigma_0^{-1})_{jj}$\,.
So we set each $\penalt_j$ equal to
\begin{equation} \label{eq:lambdaj}
  \penalt_j = \chi^2_{1,p} + \ln(2\pi) +
  \ln(C_j)\;.
\end{equation}

The effect of this choice is that a cell 
$x_{ij}$ is flagged iff it lies outside a robust 
tolerance interval around its predicted value
$\hx_{ij}$ with coverage probability $p$.
Therefore we only have to choose a single 
cutoff probability $p$ to generate all 
$\penalt_j$ automatically. 
From simulations and examples we found
that $p=0.99$ was a good choice overall,
so it is set as the default. 
Section A.5 provides more information 
on the $\penalt_j$ and the choice of $p$. 

The algorithm has been implemented as the 
R function \texttt{cellMCD()}. 
It starts by checking the data
for non-numerical variables, cases with
too many NA's and so on.
Next, it robustly standardizes the variables,
and then computes the initial estimator
followed by C-steps until convergence.
The constraint $\eig_d(\bhSigma) \geqslant a$
is applied to the standardized data, with 
default $a = 10^{-4}$.
The function also reports the number 
of flagged cells in each variable.
All the plots in the next section were made 
by the companion function 
\texttt{plot\_cellMCD()}.
Both functions have been included in the 
R package {\it cellWise} on CRAN.

%%%%%%%%%%%%%%%%%%%%%%%%%%%%%%%%%%%%%%%%%%%%%%
\section{Illustration on real data}
\label{sec:examples}

We will illustrate cellMCD on the cars data  
obtained from the Top Gear website by 
\mbox{\cite{Alfons:robustHD}}, focusing on the 11
numerical variables \texttt{price}, 
\texttt{displacement}, \texttt{horsepower}, 
\texttt{torque}, \texttt{acceleration time}, 
\texttt{top speed}, \texttt{miles per gallon}, 
\texttt{weight}, \texttt{length}, 
\texttt{width}, and \texttt{height}.
This dataset is popular because both the 
variables and the cases (the cars) can 
easily be interpreted.
After removing two cars with mostly NA's
we have $n=295$. We also replaced the highly 
right-skewed variables \texttt{price}, 
\texttt{displacement}, \texttt{horsepower}, 
\texttt{torque}, and \texttt{top speed} by 
their logarithms.
On these data we ran cellMCD in its default
version.

To visualize the results, we first look by
variable. Consider variable $j$, say
\texttt{horsepower}.
Its $i$-th cell has observed value $x_{ij}$ 
as well as its prediction $\hx_{ij}$ 
obtained from the {\it unflagged} cells in 
the same row $i$, as in \eqref{eq:deltaij}. 
In \eqref{eq:deltaij} we also see the 
conditional variance $C_{ij}$ of this cell.
It is then natural to plot the 
{\it standardized cellwise residual}
\begin{equation} \label{eq:stdres}
  \mbox{stdres}_{ij} = 
	\frac{x_{ij} - \hx_{ij}}{\sqrt{C_{ij}}}
\end{equation}
which is NA when $x_{ij}$ is missing.
The left panel of 
Figure~\ref{fig:index_Zres/X} shows the
standardized residuals of the variable
\texttt{horsepower} versus the index
(case number) $i$.
This plot was made by the function 
\texttt{plot.cellMCD()}, which 
also draws a horizontal tolerance band 
given by $\pm\, c$ where
$c = \sqrt{\chi^2_{1,0.99}} \approx 2.57$\,.
Here, some residuals stick out below the
tolerance band. 
The Renault Twizy and Citroen DS3 are
energy savers, whereas the Caterham is 
a super lightweight fun car.
The most extreme outlier is the 
Chevrolet Volt with a standardized 
residual below $-8$. 
Top Gear lists this car's power
as 86 hp, which cellMCD says
is very low compared to what would
be expected from the other 10
characteristics of this car.
Looking it up revealed that the Volt
actually has 149 hp.
As far as we know this data error was 
not detected before.

\begin{figure}[!ht]
\centering
\vskip0.2cm
\includegraphics[width=1.0\textwidth]
  {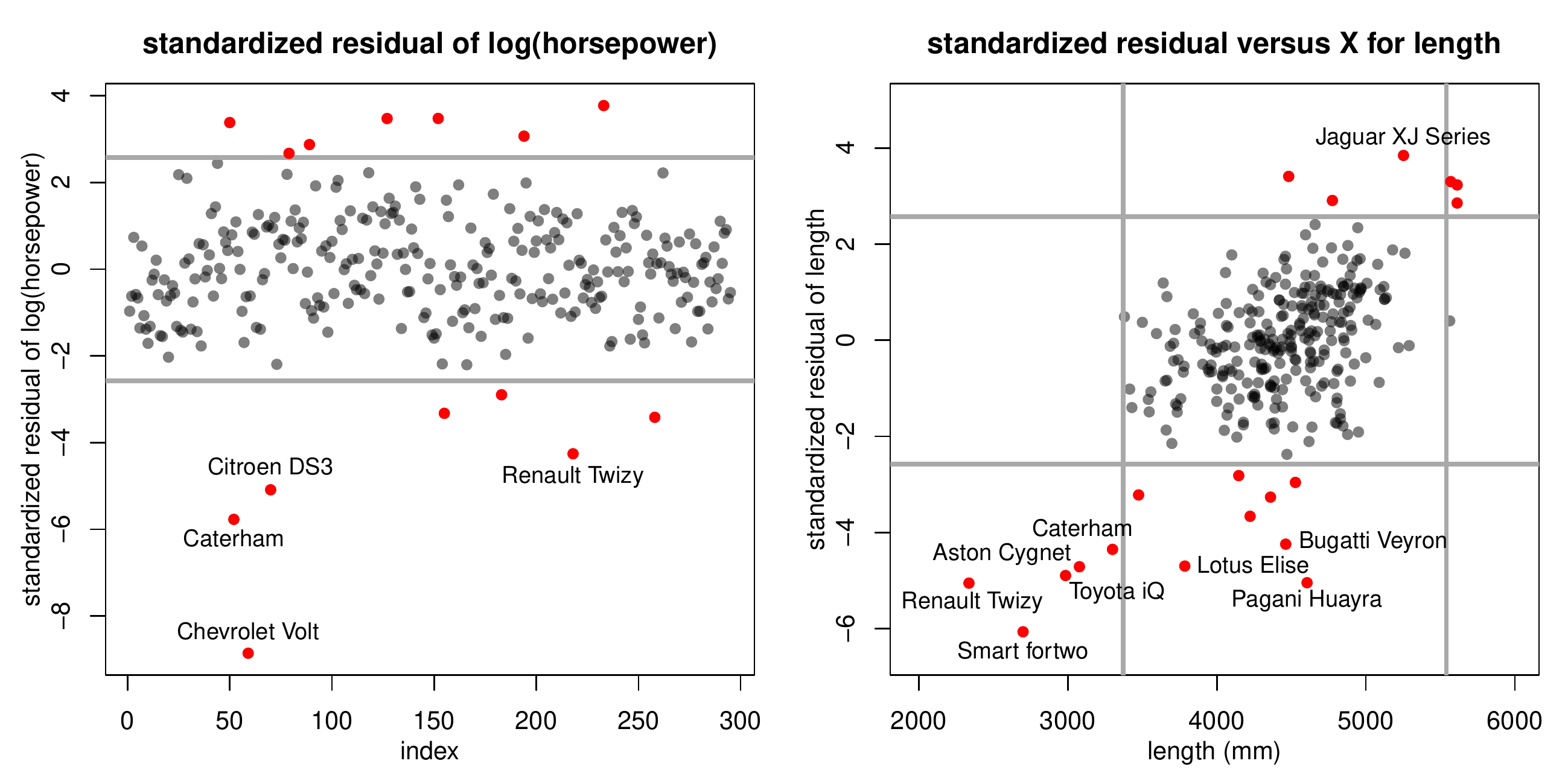}
\vskip-0.6cm
\caption{Top Gear data: (left) index plot of the
   standardized residual of 
   \texttt{log(horsepower)}; 
   (right) standardized residual of 
	 \texttt{length} versus observed 
	 \texttt{length}.}
\label{fig:index_Zres/X}
\end{figure}

The right panel of 
Figure~\ref{fig:index_Zres/X} plots the
standardized residuals of the variable
\texttt{length} versus the observed
\texttt{length} itself.
The vertical lines are at $T\pm\,cS$
where $T$ and $S$ are robust 
univariate location and scale estimates 
of \texttt{length}, obtained
from the function \texttt{estLocScale()}
in the R package {\it cellWise}.
The points to the left and right of 
such a vertical tolerance band are 
marginally outlying, i.e. their
\texttt{length} stands out by itself 
without regard to the other variables.
In the bottom left region of the plot
we see five cars that are marginal 
outliers to the left and at the same 
time have outlying negative residuals,
so they are short in absolute 
terms, as well as relative to what 
would be expected from their other
characteristics. 
The Smart fortwo, Renault Twizy, Toyota 
IQ and Aston Martin Cygnet are indeed tiny. 

However, not all cellwise outliers
are marginal outliers. 
In the middle bottom part of the plot
we marked three cars whose length is not
unusual by itself, but that are short
relative to what would be expected 
based on their other 10 variables.
They are sports cars, often built
small to achieve high speeds. 
Note that there could also be points that
lie inside the horizontal band but
(slightly) outside the vertical band.
They would correspond to cells that look a 
bit unusual in the variable $j$, but whose
observed value $x_{ij}$ is not that far
from the predicted $\hx_{ij}$ based
on its other variables.

The left panel of 
Figure~\ref{fig:Zres/pred_X/pred} plots
the standardized residual of each car's
\texttt{weight} versus its prediction.
Since all the points lie within the 
vertical tolerance band, no predictions 
are outlying.
But we do see some outlying residuals,
most of which can easily be explained.
The Bentley is a heavy luxury car, and
the Mercedes-Benz G an all-terrain 
vehicle. 
Below the horizontal tolerance band we 
see four lightweight sports cars.
What remains is the Peugeot 107 which
is small but not sporty at all.
Top Gear reports its weight as 210 kg,
which seems much too light for a car.
Based on its other characteristics,
cellMCD predicts its weight as 757 kg 
with a standard error of 89.5 kg. 
Looking up this car, its actual weight 
turns out to be 800 kg, so the value in 
the Top Gear dataset was mistaken. 

\begin{figure}[!ht]
\centering
\vspace{0.2cm}
\includegraphics[width=1.0\textwidth]
  {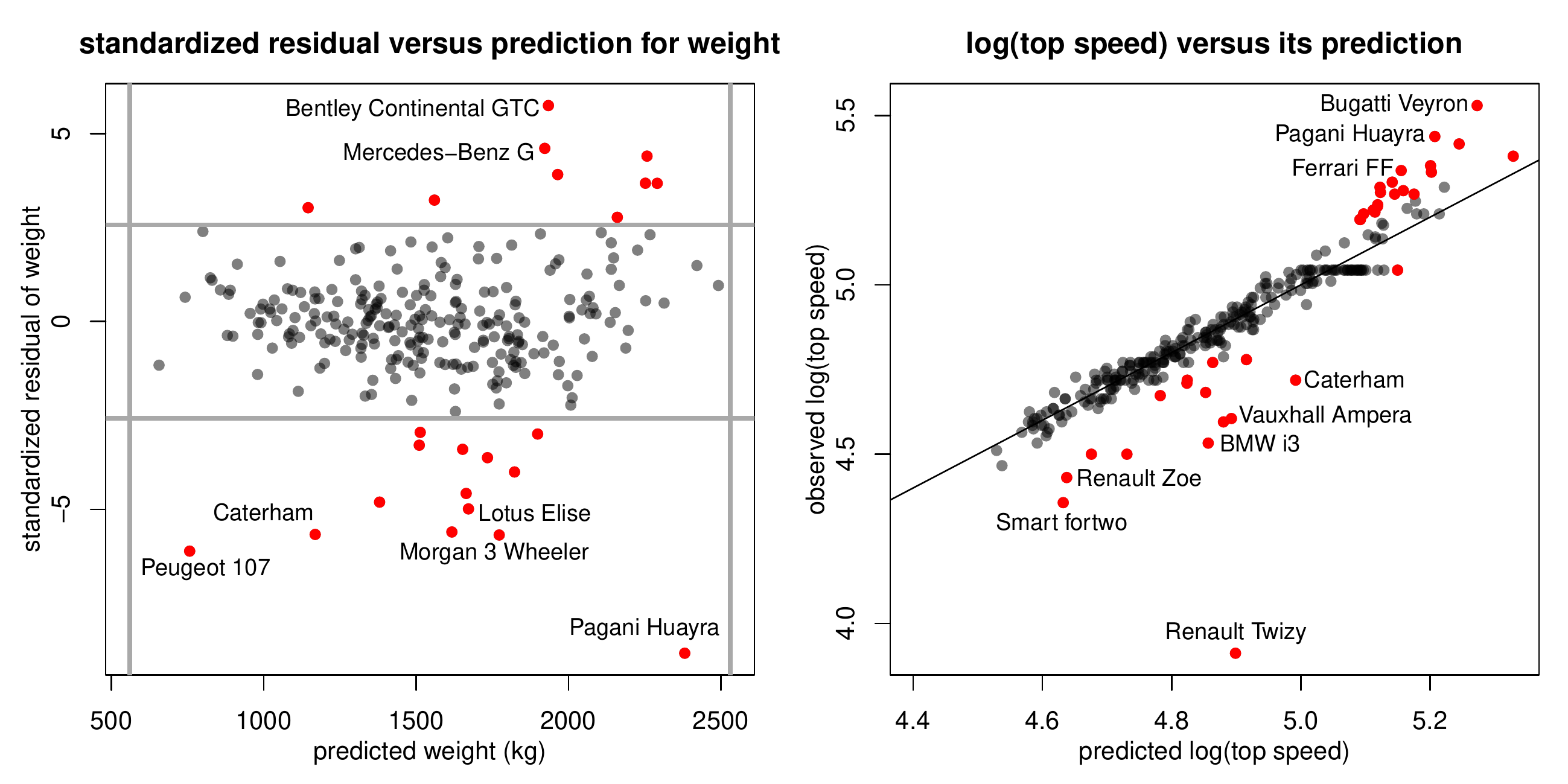}
\vskip-0.6cm  
\caption{Top Gear data: (left) standardized
  residual of \texttt{weight} versus its
  prediction; (right) observed 
  \texttt{log(top speed)} versus its
  prediction.}
\label{fig:Zres/pred_X/pred}
\end{figure}

The right panel of
Figure~\ref{fig:Zres/pred_X/pred} shows
the observed value of \texttt{top speed}
versus its prediction.
Below the superimposed $y=x$ line we find
some electric cars (BMW i3, Vauxhall
Ampera) and some small cars (Smart fortwo
and Renault Zoe). 
The one standing out most is the Renault 
Twizy, a tiny electric one-seater vehicle.
Above the line we see some extremely
fast sports cars.
Also note that some points appear to lie 
on a horizontal line.
Top Gear reports their top speed as 
155 mph, corresponding to 250 km/hour.
Many of these cars were produced by Audi, 
BMW and Mercedes with a built-in 250 
km/hour speed limiter. 

The four plot types in 
Figures~\ref{fig:index_Zres/X} 
and~\ref{fig:Zres/pred_X/pred}
all focus on a single variable.
It can also be instructive to look at a
pair of variables, say $j$ and $k$.
Figure~\ref{fig:bivariate} shows 
the variables \texttt{width} versus
\texttt{acceleration}.
The points for which $w_{ij}=0$ or
$w_{ik}=0$ or both are automatically
plotted in red.
The figure also contains an ellipse,
given by
\begin{equation} \label{eq:ellipse}
  \begin{bmatrix} x - \hmu_j & 
	y - \hmu_k \end{bmatrix}
  \begin{bmatrix}
  \hSigma_{jj} & \hSigma_{jk}\\
  \hSigma_{kj} & \hSigma_{kk}
  \end{bmatrix}^{-1}
  \begin{bmatrix} x - \hmu_j \\
	y - \hmu_k \end{bmatrix}
  = q
\end{equation}
where $q$ is the 0.99 quantile of the
$\chi_2^2$ distribution with two degrees 
of freedom.
Note that outlyingness in this type of
plot differs from cellwise outlyingness,
since the former refers to two variables
only, whereas the latter uses all 11
variables. So it is not unusual to see
some red points inside the ellipse, and
some black points outside it.

The width of the Land Rover is flagged as
this is a wide all terrain vehicle. 
The red
vertical line connects the observed point
$(x_{ij},x_{ik})$ to its predicted point
$(\hx_{ij},\hx_{ik})$ plotted in blue.
That the line is vertical means that the
\texttt{width} cell was flagged whereas 
the \texttt{acceleration} cell was not,
that is, $w_{ik} = 0$ and
$w_{ij} = 1$\,.
The acceleration of the \mbox{Ssangyong} 
Rodius and Lotus Elise 
is outlying on the left.
In fact, Top Gear lists their acceleration
time as 0 which is physically impossible:
presumably the true value was missing and
encoded as 0 instead of NA.
The same happens for the Renault Twizy.
Note that also the \texttt{width} cell of 
the Twizy is flagged, so the red line to 
its predicted point is slanted instead of
horizontal.
The Caterham also has both cells flagged,
as seen from its slanted line. 

\begin{figure}[!ht]
\centering
\includegraphics[width=0.5\textwidth]
  {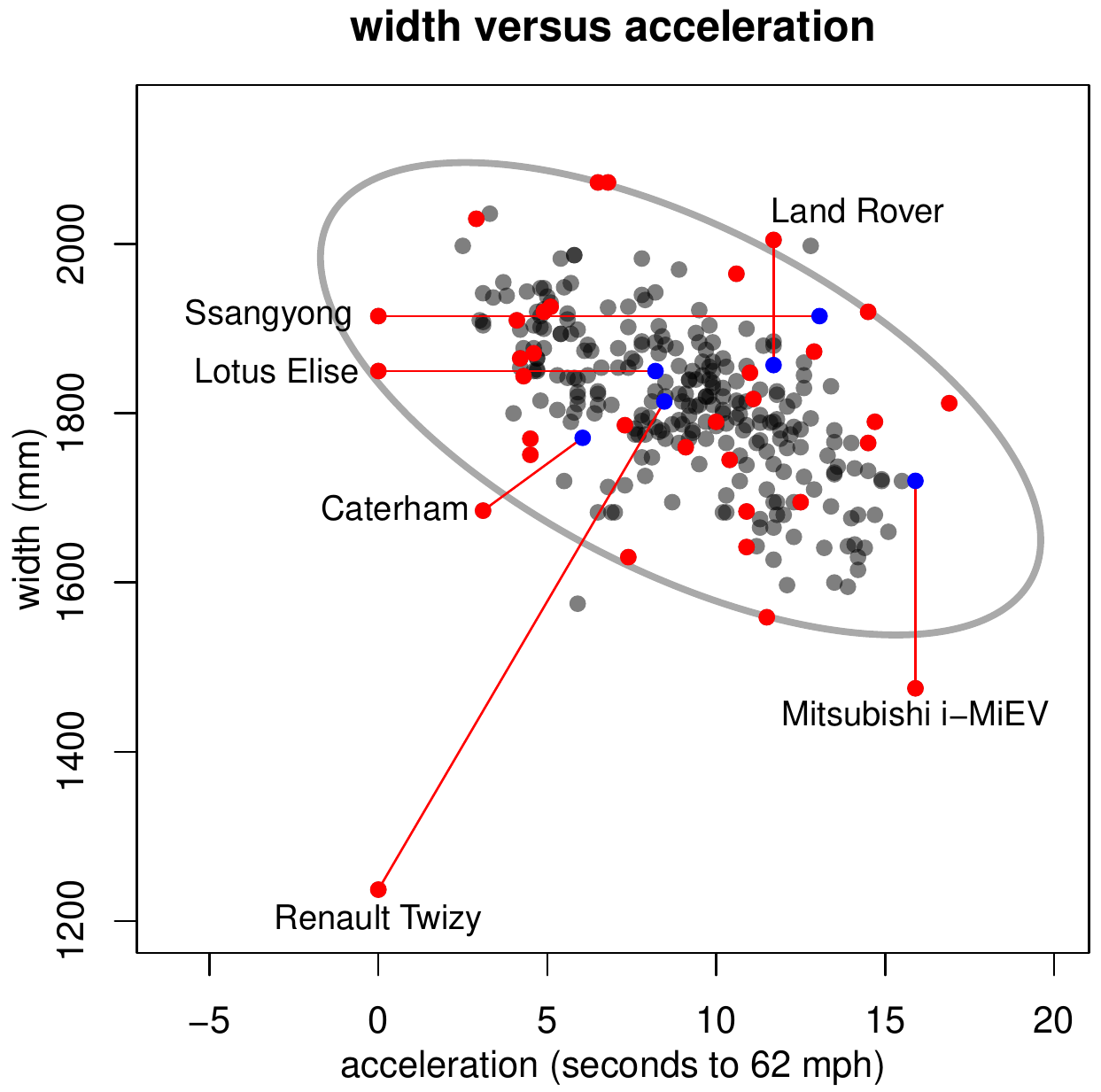}
\vskip-0.6cm  
\caption{Top Gear data: bivariate plot of
	\texttt{width} versus 
	\texttt{acceleration}.
	The 99\% tolerance ellipse is given
	by the cellMCD estimates $\bhmu$ and 
	$\bhSigma$ restricted to the variables
	in the bivariate plot, and the red lines 
	go to the predicted points shown in blue.}
\label{fig:bivariate}
\end{figure}

\section{Simulation results}
\label{sec:simul}

In this section we evaluate the performance
of cellMCD by a simulation study. 
The clean data is generated as $n$ points 
from a $d$-variate Gaussian distribution with 
mean $\bmu = \bzero$.
Since there is no affine equivariance, letting
$\bSigma$ be the identity matrix is not
sufficient. Instead we use the types ``A09'' 
and ``ALYZ''. 
The entries of the A09 correlation matrix are 
given by $\bSigma_{ij} = 0.9^{|i-j|}$, 
yielding both small and large correlations. 
The ALYZ type are randomly generated 
correlation matrices following the procedure 
of \cite{Agostinelli2015} and typically have 
mostly small absolute correlations. 
We consider three combinations of sample size 
and dimension $(n,d)$: $(100,10)$, 
$(400,20)$, and $(800,40)$.

In these clean data, we then replace a fraction 
$\varepsilon$ in $\{0.1,0.2\}$ of cells by 
contaminated cells. 
These are generated as follows. 
First, for each column in the data matrix we 
randomly sample $n\varepsilon$ indices of 
cells to be contaminated. 
In each row, say $(z_1, \ldots, z_d)$, we 
then collect the indices of the cells to be 
contaminated. 
Denote this set of size $k$ by 
$K = \{j_1, \ldots, j_k\}$. 
We next replace the cells 
$(z_{j_1}, \ldots, z_{j_k})$ by the 
$k$-dimensional vector 
$\gamma \sqrt{k}\,\bv_K/ \mbox{MD}(\bv_K, 
\bmu_K, \bSigma_K)$ where $\bmu_K$ and 
$\bSigma_K$ are $\bmu$ and $\bSigma$ 
restricted to the indices in $K$. 
The scalar $\gamma >0$ quantifies the distance 
of the outlying cells to the center of the 
distribution, and we vary $\gamma$ over 
$1, \ldots, 10$. 
The vector $\bv_K$ is the normed eigenvector 
of $\bSigma_K$ with the smallest eigenvalue. 
In each row, the outlying cells are thus 
structurally outlying in the subspace 
generated by the variables in $K$. 
Therefore, these cells will often not be 
marginally outlying, especially when $|K|$ 
is large and $\gamma$ is relatively small, 
which makes them hard to detect. 
The \texttt{R}-package \texttt{cellWise} 
\citep{cellWise} contains the function 
\texttt{generateData} which generates the 
contaminated data according to this 
procedure.

We compare the proposed method cellMCD to the 
following alternative estimators:
\begin{itemize}
\item \textbf{Grank}, \textbf{Spearman}: the 
      Gaussian and Spearman rank-based 
      estimators used in \cite{ollerer2015} and 
			\cite{croux2016}; 
\item \textbf{GKnpd}: the Gnanadesikan-Kettenring 
      estimator used in \cite{tarr2016};
\item \textbf{2SGS}: the two-step generalized 
      S-estimator of \cite{Agostinelli2015};
\item \textbf{DI}: the detection-imputation 
      algorithm of \cite{cellHandler}.
\end{itemize}

In order to evaluate the performance of the 
different estimators, we compute the 
Kullback-Leibler discrepancy between the 
estimated $\widehat{\bSigma}$ and the 
true $\bSigma$ given by
\begin{equation*}
  \mbox{KL}(\widehat{\bSigma}, 
  \bSigma) = \mbox{tr}(
	\widehat{\bSigma}\bSigma^{-1})- d - 
	\log(\det(\widehat{\bSigma} 
	\bSigma^{-1}))\;.	
\end{equation*}
For each setting of the simulation parameters 
we generate 100 random datasets, and average 
the Kullback-Leibler discrepancy over these 100 
replications. (For the variability around 
these averages see subsection A.6.1.)

\begin{figure}[!ht]
\centering
\includegraphics[width=0.93\textwidth]{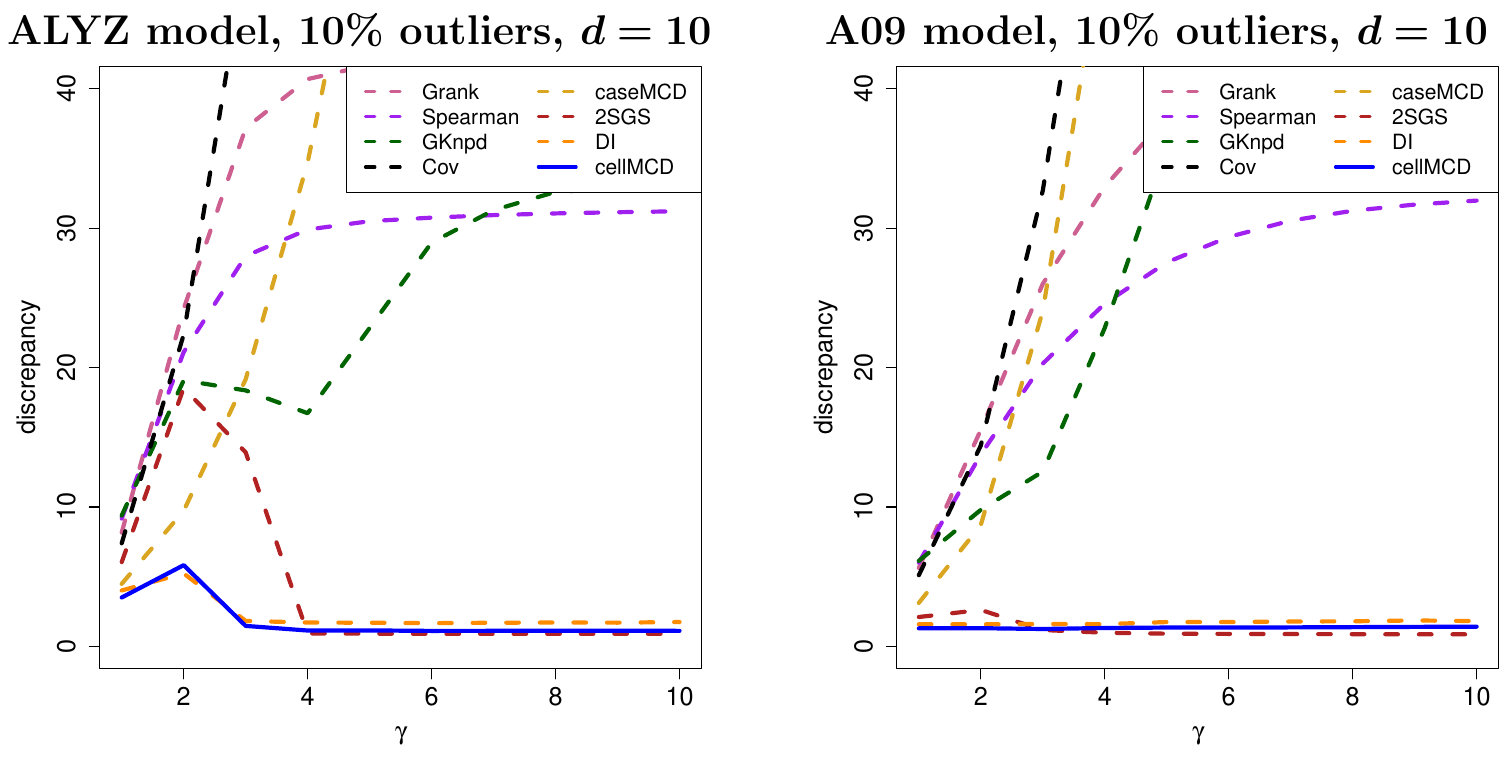}
\vspace{-0.8cm}
\caption{Discrepancy of estimated covariance
  matrices for $d = 10$ and $n = 100$.}
\label{fig:KLdiv_lowdim}
\end{figure}

Figure~\ref{fig:KLdiv_lowdim} presents the 
results for $d = 10$, $n = 100$ and 
$\varepsilon = 0.1$. (The results for 
$\varepsilon = 0.2$ were similar.)
Both cellMCD and DI perform well, as does 
2SGS provided $\gamma \geqslant 4$. 
As expected, the classical covariance
matrix (Cov) and the casewise MCD
(labeled caseMCD) were not robust to
these adversarial cellwise outliers.
Note that the performances of Grank, 
Spearman and GKnpd do not improve as $\gamma$ 
increases. While these estimators bound the 
influence that a single cell can have on the 
estimation, the effect remains substantial
as the cell becomes more outlying. 
This is in contrast to 2SGS, DI and cellMCD
in which far outliers get a zero weight. 

\begin{figure}[!ht]
\centering
\includegraphics[width=0.95\textwidth]{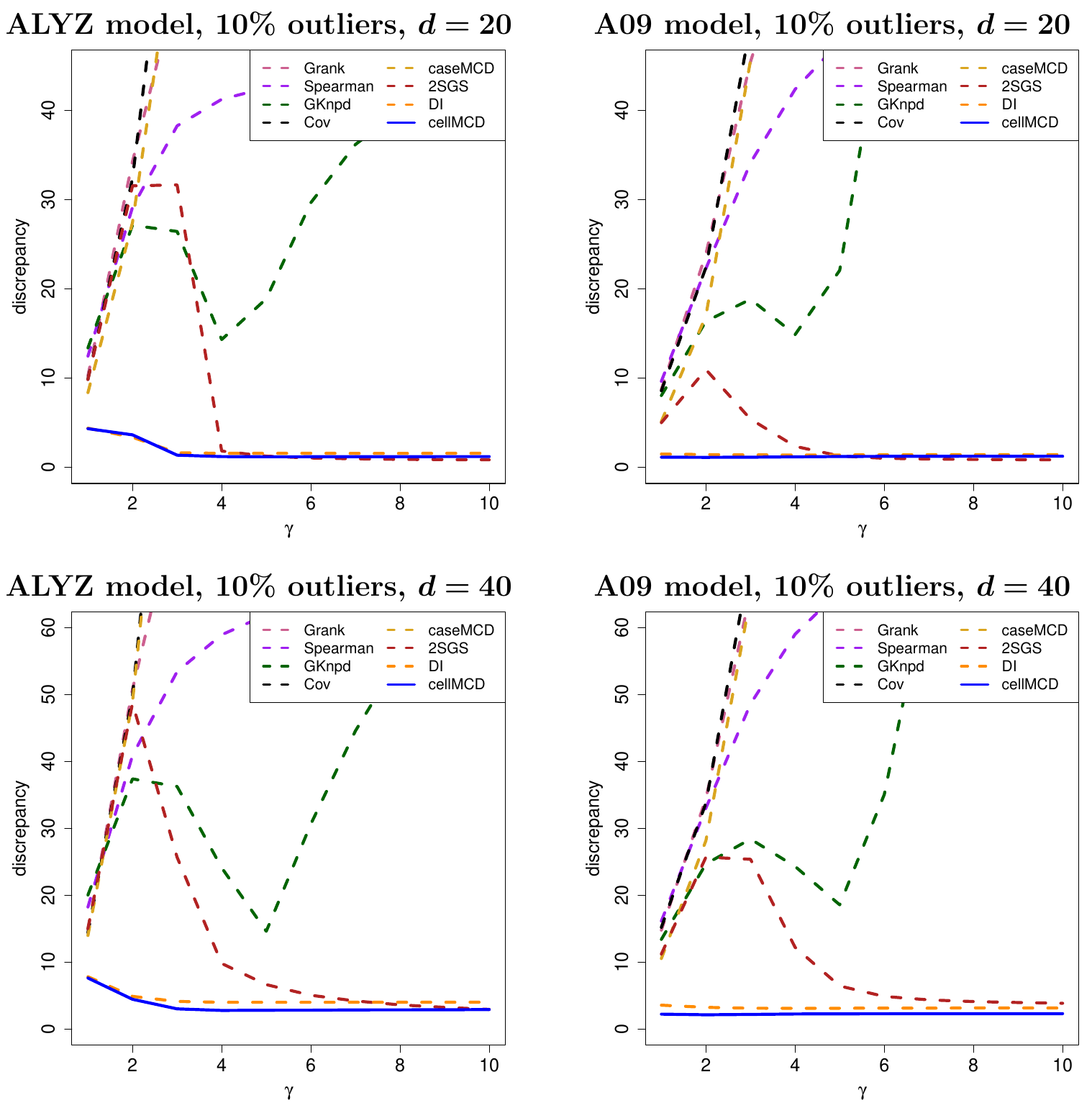}
\vspace{-0.8cm}
\caption{Discrepancy of estimated covariance
  matrices for $d = 20$ and $n = 400$ (top panels)
	and for $d = 40$ and $n = 800$ (bottom panels).}
\label{fig:KLdiv_highdim}
\end{figure}

The top panels of 
Figure~\ref{fig:KLdiv_highdim} show the 
results for $n = 400$ and $d = 20$.
The relative performances are 
similar to Figure~\ref{fig:KLdiv_lowdim}.
The 2SGS method still does well when 
$\gamma>4$, but now suffers more for low 
$\gamma$. 
The performances of DI and cellMCD are again 
very close, with cellMCD often doing slightly 
better. 

The lower panels with $n = 800$ and $d = 40$
are similar, with cellMCD performing best
for all values of $\gamma$ while DI is quite
close, and 2SGS only doing well for higher 
$\gamma$. 

Table~\ref{tab:comptimes} lists the 
computation times of the methods in the 
simulation, in seconds.
The first five methods are fast but they
performed poorly. 
The bottom three methods did better.
In dimensions 20 and 40 the cellMCD method was
the fastest among them.

\begin{table}[ht]
\centering
\caption{Computation times of the 
methods in the simulation.}
\label{tab:comptimes}
\small
\vspace{0.3cm}
\begin{tabular}{lccc}
\hline
 & d=10 & d=20 & d=40 \\ 
\hline
Cov & 0.00 & 0.00 & 0.00 \\ 
Grank & 0.00 & 0.01 & 0.05 \\ 
Spearman & 0.01 & 0.02 & 0.06 \\ 
GKnpd & 0.90 & 1.31 & 4.29 \\ 
caseMCD & 0.04 & 0.53 & 2.37 \\ 
DDCW & 0.01 & 0.03 & 0.18 \\ 
2SGS & 0.67 & 6.91 & 66.88 \\ 
DI & 0.28 & 4.72 & 41.41 \\ 
cellMCD & 0.28 & 1.83 & 22.47 \\ 
\hline
\end{tabular}
\end{table}
\normalsize

We are also interested in the performance of
these methods on data without outliers.
For this we repeated the simulation with
$\varepsilon = 0$, again with 100 
replications.
The variability of each entry of the 
covariance matrix was measured taking the 
Fisher information of that entry into account.
These results were then averaged over the
upper triangular matrix entries 
including the diagonal.
Next we divided the MSE of the classical
MLE estimator by that of each robust
method, yielding the finite-sample
efficiencies in Table~\ref{tab:effi}.

\begin{table}[!ht] 
\begin{center}
\caption{Finite-sample efficiencies of 
robust covariance estimators}
\label{tab:effi}
\small
\vspace{0.3cm}
\begin{tabular}{lcccccccc}
\hline
	& \phantom{abc} 
	& \multicolumn{3}{c}{ALYZ configuration}
	& \phantom{abc}
	& \multicolumn{3}{c}{A09 configuration} \\
\cmidrule{3-5} \cmidrule{7-9}
method  &  & $d=10$ & $d=20$ & d=40 
        &  & $d=10$ & $d=20$ & d=40\\
\hline
cellMCD  & & 0.90  &  0.90  &  0.89 
         & & 0.89 & 0.93 & 0.96\\
2SGS     & & 0.87 & 0.94 & 0.98 
         & & 0.83 & 0.91 & 0.95\\				
DI       & & 0.68 & 0.61 & 0.49 
         & & 0.87 & 0.90 & 0.90\\	
GKnpd    & & 0.74 & 0.80 & 0.81 
         & & 0.78 & 0.77 & 0.79\\
Grank    & & 0.90 & 0.96 & 0.98 
         & & 0.88 & 0.89 & 0.94\\				
Spearman & & 0.84 & 0.88 & 0.90 
         & & 0.83 & 0.82 & 0.85\\			
\hline
\end{tabular}
\end{center}
\end{table}
\normalsize

We see that the efficiency of cellMCD
averages over 90\%, which is excellent 
for a highly robust covariance estimator.
This is similar to 2SGS, and outperforms DI. 
As expected Grank has a high efficiency, 
but we just saw that it performed poorly 
under contamination, as did GKnpd and 
Spearman.
The finite-sample efficiency of cellMCD 
is much higher than that of the casewise 
MCD with the same coverage parameter 
$h=0.75n$, which is under 0.70 for this 
range of dimensions $d$.
This is due to the penalty term  
in~\eqref{eq:cellMCD}, which made the 
number of actually discarded cells much 
smaller than $0.25\,n$.

We conclude that cellMCD is about equally
robust as DI but with better efficiency,
and is about as efficient as 2SGS but with 
better robustness at contaminated data.
Moreover, it does substantially better at 
contaminated data than the remaining
methods.

\section{Discussion}
\label{sec:concl}

The cellMCD method proposed here has an
elegant formulation based on a single
objective function, making it easier to
understand than the earlier 2SGS and
DI methods. 
We proved its good breakdown properties
and consistency, 
and like the casewise MCD it can be computed 
by an algorithm based on C-steps that always 
lower the objective function and is
guaranteed to converge. 
We have illustrated cellMCD on a real data set
where the accompanying graphical displays
revealed interesting aspects of the data
that aided interpretation.
Simulations indicate that cellMCD outperforms 
earlier cellwise methods, while being 
conceptually simple and rather fast to compute.

CellMCD is cellwise robust and incorporates
a kind of sparsity penalty (on $\bone-\bW$). 
This naturally brings to mind the work 
of \cite{candes2011}.
The goals are clearly related, but there are 
also some differences. 
The first is that their work assumes that 
the cellwise outlier pattern $\bW$ is drawn
uniformly at random, whereas we adopt the
robustness paradigm that the outliers may be
placed adversarially.
Secondly, the method of \cite{candes2011} is
equivariant for transposing the data matrix, so
it treats cases and variables in the same way,
whereas in our setting they have to be treated
differently.
We do allow for some rows being flagged entirely,
whereas we cannot allow flagging an entire 
column as this would make $\bmu$ and $\bSigma$ 
not identifiable, which motivates our
constraint $||\bW_{.j}||_0 \geqslant h$
for $j=1,\ldots,d$\,.

The fact that implosion breakdown can
happen easily in the cellwise setting,
see~\eqref{eq:cov}, was not mentioned 
in the literature before. 
We feel that, apart from cellMCD,
also other cellwise robust covariance 
estimators could benefit from a 
constraint such as
$\eig_d(\bhSigma) \geqslant a$, or
similarly from a formulation in 
which $\bhSigma$ is a convex 
combination of two matrices, one of 
which is a small multiple of the 
identity matrix.

The casewise MCD is typically followed by a 
reweighting step. This works as follows.
First, the estimated covariance matrix 
$\bhSigma$ is multiplied by a correction 
factor $c_{n,d,h}$ such that 
$c_{n,d,h} \bhSigma$\, is roughly unbiased
when the original data are generated from
a Gaussian distribution.
Next, one computes the squared robust 
distances of the data points, given by 
$\RD^2_i=(\bx_i - \bhmu)^\top (c_{n,d,h}
\bhSigma)^{-1} (\bx_i - \bhmu)$.
Each case $\bx_i$ then gets a weight
$w_i$ depending on its $\RD^2_i$\,.
Typically, the weight is set to 1 when 
$\RD_i^2$ is below some quantile of the
$\chi_d^2$ distribution with $d$ degrees
of freedom, and to 0 otherwise.
The final estimates are then the 
weighted mean and the weighted covariance
matrix~\eqref{eq:wcov}.
This reweighting step increases the 
finite-sample efficiency of the estimator.

For cellMCD, the analogous
reweighting step would compute the
standardized residual~\eqref{eq:stdres}
of every cell $x_{ij}$ and compare its
square to a quantile of the $\chi_1^2$ 
distribution with 1 degree of freedom,
yielding zero-one weights $w_{ij}$. 
With these $w_{ij}$ one would then run
the EM algorithm on the original data.
But in fact, the result is not very 
different from the cellMCD result.
This is because all the ingredients are
already used in cellMCD, which contains 
the squared standardized residual 
in~\eqref{eq:deltaij}, the $\chi_1^2$ 
quantile in~\eqref{eq:lambdaj},
and the partial likelihood on which EM
is based in~\eqref{eq:cellMCD}.
So in some sense the components of a
reweighting step are already built into
cellMCD itself. This explains its rather
high finite-sample efficiency in 
Table~\ref{tab:effi}. 

\vspace{0.3cm}

\noindent{\bf Software availability:} 
The cellMCD method is implemented as the
function \texttt{cellMCD()}, and the plots
in Section~\ref{sec:examples} were drawn
by the function \texttt{plot\_cellMCD()}.
Both functions are available 
in the R package {\it cellWise} on CRAN. 
Its vignette \texttt{cellMCD\_examples} 
reproduces all results and figures in 
Section~\ref{sec:examples}.

\vspace{0.3cm}
\noindent{\bf Acknowledgment:}
We are grateful for the constructive comments
made by the Editor, Associate Editor, 
and five reviewers.

\vspace{0.3cm}
\noindent{\bf Disclosure statement:} The authors report 
there are no competing interests to declare.

%%%%%%%%% SUPPLEMENTARY MATERIAL %%%%%%%%%

\clearpage
\pagenumbering{arabic}
   % restarts page numbering from 1
%\setcounter{page}{1}
%
\appendix
%\section*{Supplementary Material}% \label{sec:A}
   % The * makes this section unnumbered
\begin{center}
%\large{\bf Supplementary Material}
\large{Supplementary Material to:\\ 
  The Cellwise Minimum Covariance Determinant 
	Estimator\\ 
  Jakob Raymaekers and Peter J. Rousseeuw}
\end{center}
\numberwithin{equation}{section} 
   % restarts equation numbering from 1
%\section*{Supplementary Material} \label{sec:A}
\renewcommand{\theequation}
   {A.\arabic{equation}}
   % labels equations as (A.1),...

\spacingset{1.45} % changed for the SuppMat only

\section*{\large A.1\;\; Proof of breakdown results}
\label{A:proofs}

\begin{proof}[Proof of Proposition~\ref{prop:bdv}.]
The proof consists of four parts.

\vspace{0.3cm}
\noindent{Part (a): this follows immediately
from the constraint 
$\eig_d(\bhSigma) \geqslant a$ for $a>0$.}

\noindent{Part (b): Explosion breakdown of 
$\bhSigma$\,.}

Denote by $\mathcal{X}_m$ the set of all 
corrupted samples $\bX^m$ obtained by 
replacing at most $m$ cells in each column 
of $\bX$ by arbitrary values, for
$m = n-h$.
Also denote
\begin{equation*}
  \mathcal{W}_h = \{\bW\in \{0,1\}^{n\times d}
  \;\;|\;\; ||\bW_{.j}||_0 \geqslant h
  \mbox{ for all } j=1,\ldots,d\}\,.
\end{equation*}
Then we can write 
$$\mathcal{X}_m = 
  \bigcup_{\bW^* \in \mathcal{W}_h}
	{\{\bX^m \in \mathcal{X}_m \,|\, 
	\bW_{ij}^* = 1 \Rightarrow 
	\bX_{ij}^m = \bX_{ij}\}}\,.$$
In other words, we can write the set of all 
corrupted samples $\mathcal{X}_m$ as a finite 
union over subsets of corrupted samples with 
the same contaminating configuration $\bW^*$.

We start by showing the existence of a 
solution with finite objective function.
Consider any such contaminating 
configuration $\bW^* \in \mathcal{W}_h$\,. 
Then take the solution 
$(\bhmu_\tEM, \bhSigma_\tEM, \bW^*)$ where the 
location and scatter are the result of the 
EM-algorithm with fixed missingness pattern 
given by $\bW^*$. Then
\begin{align*}
  &\forall \bX^m \in \{\bX^m \in 
	\mathcal{X}_m \,|\, \bW_{ij}^* = 1 
	\Rightarrow \bX_{ij}^m = \bX_{ij}\}:\\
  &\obj(\bhmu_\tEM(\bX^m), 
	\bhSigma_\tEM(\bX^m), \bW^*) = 
	\obj(\bhmu_{\tEM}(\bX), \bhSigma_{\tEM}(\bX), 
	\bW^*) = M_{\bW^*} < \infty
\end{align*}
in which $\obj(\bmu, \bSigma, W)$ denotes 
the objective function \eqref{eq:cellMCD} 
of cellMCD.
In other words, for all $\bX^m$ with the same 
contaminating configuration, we have a candidate 
solution with a finite objective function. 
Since there are finitely many such contaminating 
distributions, we can always find a candidate 
solution with a value of the objective function 
smaller than $M = 
\displaystyle \max \{M_{\bW^*}\,:\, \bW^* \in 
\mathcal{W}_h\} < \infty$.\\

We now show that $\bhSigma$ does not explode.
By construction, 
$\eig_d(\bhSigma) \geqslant a$ for 
some constant $a >0$. Then we have that 
\begin{align*}
  \ln|\bhSigma^{(\bw_i)}| 
	&= \sum_{j=1}^{\dwi}{\ln\eig_j(
	   |\bhSigma^{(\bw_i)}|)}  \\
  &= \ln\eig_1(|\bhSigma^{(\bw_i)}|) 
	   + \sum_{j=2}^{\dwi}{\ln\eig_j(
		|\bhSigma^{(\bw_i)}|)}\\
  &\geqslant \ln\eig_1(|\bhSigma^{(\bw_i)}|)
	   + \sum_{j=2}^{\dwi}{\ln \eig_{\dwi}
		   (|\bhSigma^{(\bw_i)}|)}\\
  &\geqslant \ln\max_{j} \bhSigma_{jj}^{(\bw_i)} 
	  + \sum_{j=2}^{\dwi} {\ln(\eig_d(\bhSigma))}\\
  &\geqslant \ln\max_{j} 
	 \bhSigma_{jj}^{(\bw_i)} + (d-1)\ln a
\end{align*}
where we have used that  
$\eig_1(\bhSigma^{(\bw)}) \geqslant 
\max_{j} \bhSigma_{jj}^{(\bw)}$ for any $\bw$. 
That is, the largest eigenvalue of any positive 
semi-definite (sub)matrix is at least as large 
as its largest diagonal element.

Now we can bound the first term of the 
objective from below by an increasing function 
of the largest eigenvalue. 
First note that we have at least one 
row $i^*$ for which the $j^*$-th element of 
$\bw_{i^*}$ is 1, where 
$j^* = \argmax_{j}\bhSigma_{jj}$\,. Therefore
\begin{align*}
  \sum_{i=1}^{n}{\ln |\bhSigma^{(\bw_i)}|} 
	&= \ln |\bhSigma^{(\bw_{i^*})}| 
	  + \sum_{i\neq i^*}{\ln 
		|\bhSigma^{(\bw_i)}|}\\
  &\geqslant \ln\max_{j} 
	 \bhSigma_{jj}^{\bw_{i^*}} + 
	 (d-1) \ln a + (n-1) d \ln a\\
  &= \ln\max_{j} \bhSigma_{jj} 
	  + (nd-1)\ln a\\
  &= \ln\max_{jk} |\bhSigma_{jk}| 
	  + (nd-1)\ln a\\
  &\geqslant \ln \frac{\eig_1(\bhSigma)}{d} 
	  + (nd-1)\ln a
\end{align*}
where we have used that
$\eig_1(\bhSigma) \leqslant d \max_{jk}
{|\bhSigma_{jk}|}$, i.e. the largest eigenvalue 
of a $d \times d$ positive definite matrix is 
at most $d$ times its largest absolute entry. 
Also, we have used that 
$\max_{jk}{|\bhSigma_{jk}|} = 
 \max_{j}{|\bhSigma_{jj}|}$ since $\bhSigma$ 
is a covariance matrix, so its maximum 
occurs on the diagonal.

As all other terms of the objective function 
are bounded from below by zero, we obtain:
\begin{align*}
\obj(\bhmu, \bhSigma, \bW)
  &= \sum_{i=1}^{n} {\big(
	   \ln |\bSigma^{(\bw_i)}| + 
	   \dwi\ln(2\pi) +
	   \MD^2(\bx_i^m,\bw_i,\bhmu,\bhSigma)\,
		 \big)} + \sum_{j=1}^d \penalt_j
	   ||\bone_d - \bW_{.j}||_0\\
	&\geqslant \sum_{i=1}^{n}
	 {\ln |\bhSigma^{(\bw_i)}|}
	 \geqslant \ln \frac{\eig_1(\bhSigma)}{d} 
	           + (nd-1)\ln a\;.
\end{align*}	
We thus find that the objective function 
explodes when $\eig_1(\bhSigma) \to \infty$. 
Given that for any possible contaminated 
dataset there is a candidate solution 
with objective function less than or
equal to $M < \infty$, we conclude that 
the solution cannot have an exploding 
eigenvalue.\\

\noindent{Part (c): Breakdown of $\bhmu$\;.}

Note that for all $\bW \in \mathcal{W}_h$
and each variable $j$ there is at least
one $\bW_{ij} =1$, so a cell $\bX_{ij}$
that was not replaced. Denote 
$M_2 := \max_{ij}{|\bX_{ij}|} < \infty$\,.
Then we have
\begin{align*}
  \obj(\bhmu, \bhSigma, \bW)
  &= \sum_{i=1}^{n} {\big(
	   \ln |\bSigma^{(\bw_i)}| + 
	   \dwi\ln(2\pi) +
	   \MD^2(\bx_i^m,\bw_i,\bhmu,\bhSigma)
		 \,\big)} + \sum_{j=1}^d \penalt_j
	   ||\bone_d - \bW_{.j}||_0\\
	&\geqslant n d\ln a  + \sum_{i=1}^{n}
	  {\MD^2(\bx_i^m,\bw_i,\bhmu,\bhSigma)}\\
	&= n d\ln a  + \sum_{i=1}^{n}{\left|\left|
	  \left(\bhSigma^{(\bw_{i})}\right)^{-1/2} 
		(\bx_{i, o_i}^m-\bhmu_{o_i})
		\right|\right|_2^2}\\
	&\geqslant n d\ln a +\sum_{i=1}^{n}
	  {\lm^2\left(\left(\bhSigma^{(\bw_{i})}
		\right)^{-1/2}\right)\left|\left|
		\bx_{i, o_i}^m-\bhmu_{o_i}
		\right|\right|_2^2}\\
	&= n d\ln a  +\sum_{i=1}^{n}
	  {\frac{1}{\lM(\bhSigma^{(\bw_{i})})}
		\left|\left|\bx_{i, o_i}^m-
		\bhmu_{o_i}\right|\right|_2^2}\\
	&\geqslant n d\ln a +
	   \frac{1}{\lM(\bhSigma)} \sum_{i=1}^{n}
		 {\left|\left|\bx_{i, o_i}^m-
		\bhmu_{o_i}\right|\right|_2^2}\\
	& \geqslant n d\ln a  +
	  \frac{1}{\lM(\bhSigma)} 
		\left(||\bhmu||_2^2 - dM_2^2\right).
\end{align*}
In the last line we have used that 
there is at least one uncontaminated cell 
in each variable for which $\bW_{ij}=1$, 
together with the fact that this cell is 
bounded in absolute value by $M_2$. 
From part (b) we don't have explosion of 
the covariance matrix, so 
$\lM(\bhSigma) < \infty$. 
Should $||\bhmu||_2 \to \infty$ our 
objective function would explode, but we 
know it does not.\\

\noindent{Part (d): The bound $(n-h+1)/n$ 
is sharp.} 
So far we know that 
$\varepsilon^*_n(\bhmu, \bX) \geqslant 
(n-h+1)/n$ and $\varepsilon^+_n(\bhSigma,\bX) 
\geqslant (n-h+1)/n$. 
We now show that this common lower bound
cannot be improved, by constructing an 
example which causes breakdown. For this we 
take a contaminating configuration obtained 
by replacing $n-h+1$ cells in the first 
column of the data $\bX$ by some value $c$ 
and leaving all other columns untouched. 
Unlike before, there is no way to cover all 
these cells with any $W \in \mathcal{W}_h$. 
Put $M_2 = \max_{ij}{|\bX_{ij}|}$ as before.

Consider any solution $(\bhmu,\bhSigma,\bW)$ 
with $\bW \in \mathcal{W}_h$\,. 
Denote by $\mathcal{I}$ the set of indices 
of the rows which have a contaminated cell 
equal to $c$ in their first variable. 
Denote by subscript $o_i$ the set of
variables $j$ for which $\bw_{ij} = 1$. 
By the first order conditions of the EM 
algorithm, upon convergence of the
algorithm we must have 
$\bhmu = \frac{1}{n}\sum_{i=1}^{n} \by_i$ 
where $\by_i$ are the imputed observations. 
For the first entry of $\bhmu$ we have:
\begin{align*}
\bhmu_1 
  &= \frac{1}{n}\sum_{i=1}^{n} \by_{i1}\\
  &= \frac{1}{n}\sum_{i|\bw_{i1} = 1} 
	   \bX^m_{i1} + \frac{1}{n}
		 \sum_{i|\bw_{i1} = 0} 
		 E[\bX_{i1}|\bhmu, \bhSigma, \bW]\\
  &= \frac{1}{n}\sum_{i|\bw_{i1} = 1}
	   \bX^m_{i1} + \frac{1}{n}
		 \sum_{i|\bw_{i1} = 0} {\left(
		 \bhmu_1 + \bhSigma_{1, o_i}
		 \bhSigma_{o_i, o_i}^{-1}
		 (\bX^m_{i,o_i} - \bhmu_{o_i})\right)}\\
  &= \frac{1}{n}\sum_{\{i|\bw_{i1} = 1\} 
	   \cap \mathcal{I}} \bX^m_{i1} + 
		 \frac{1}{n}\sum_{\{i|\bw_{i1} = 1\} 
		 \cap \mathcal{I}^C} \bX^m_{i1} + 
		 \frac{1}{n}\sum_{\{i|\bw_{i1} = 0\}}
		 {\left(\bhmu_1 +\bhSigma_{1, o_i}
		 \bhSigma_{o_i,o_i}^{-1}
		 (\bX^m_{i,o_i} - \bhmu_{o_i})\right)}\\
  &= \frac{c}{n}\#(\{i|\bw_{i1} = 1\} 
	   \cap \mathcal{I}) + 
		 \frac{1}{n}\sum_{\{i|\bw_{i1} = 1\}
		 \cap \mathcal{I}^C} \bX^m_{i1} + 
		 \frac{1}{n}\sum_{\{i|\bw_{i1} = 0\}}
		 {\left(\bhmu_1 +\bhSigma_{1, o_i}
		 \bhSigma_{o_i,o_i}^{-1}
		 (\bX^m_{i,o_i} - \bhmu_{o_i})\right)}\\
  &= \frac{c}{n}\#(\{i|\bw_{i1} = 1\}
	   \cap \mathcal{I}) + 
		 \frac{1}{n}\sum_{\{i|\bw_{i1} = 1\} 
		 \cap \mathcal{I}^C} \bX_{i1} + 
		 \frac{1}{n}\sum_{\{i|\bw_{i1} = 0\}}
		 {\left(\bhmu_1 +\bhSigma_{1, o_i}
		 \bhSigma_{o_i,o_i}^{-1}
		 (\bX_{i,o_i} - \bhmu_{o_i})\right)}\;.
\end{align*}
Note that we have replaced $\bX^m$ by $\bX$ 
in the last line, since all those cells
are uncontaminated.
By construction of our contaminated data, 
we have 
$\#(\{i|\bw_{i1} = 1\} \cap  \mathcal{I}) 
\geqslant 1$. 
Now take a sequence $c_k$ which diverges, 
i.e. $c_k \to \infty$ as $k \to \infty$. 
Suppose that our estimates $\bhmu$ and 
$\bhSigma$ would not break down as 
$k \to \infty$. 
Then the $\bhmu_1$ on the left hand side 
of the above equality would be bounded.
The second term on the right hand side
is just an average of uncontaminated data
so it is bounded too. 
The last term on the right hand side would 
be bounded as well, since it consists of 
the estimated $\bhmu$, $\bhSigma$, and 
the uncontaminated data. 
(Note that $\bhSigma_{o_i,o_i}^{-1}$ is
bounded since 
$\eig_1(\bhSigma_{o_i,o_i}^{-1})
\leqslant 1/a < \infty$.)
However, the first term on the right hand 
side would diverge. This is a contradiction. 
We conclude that either the location or the 
covariance matrix (or both) must diverge
as $k \to \infty$. 
\end{proof}

%%%%%%%%%%%%%%%%%%%%%%%%%%%%%%%%%%%%%%%%%%%%%%%%%%%%%%%%%%%%%%%%
\section*{A.2\;\; Asymptotic properties of cellMCD}
\label{A:asy}

\subsection*{A.2.1\;\; Introduction}
In this section we study the asymptotic properties of cellMCD for well-behaved distributions $F$ (to be specified later). In order to do so, we consider the cellMCD objective without the columnwise restriction on $W$. This is justified, since from an asymptotic perspective, the columnwise restriction on $W$ only plays a role when it is encountered asymptotically. By our choice of $q_j$, we know that this is not the case for the normal distribution, and even for much more heavy tailed distributions, such as the multivariate Cauchy with independent components, we won't flag 25 \% of the values in the marginal distributions asymptotically.\\
Now, without the columnwise restriction on $W$, there are different ways of writing the cellMCD objective, of which the version~\eqref{eq:obj1} lends itself somewhat better to asymptotic analysis:
\begin{equation*}%\label{eq:obj1}
\G(\mu, \Sigma, F) \coloneqq \int \g_{\mu, \Sigma}(x) F(\mathrm{d}x)    
\end{equation*}
where we use~\eqref{eq:obj2}:
\begin{equation*}%\label{eq:obj2}
    \g_{\mu,\Sigma}(x) \coloneqq \min_{w\in \{0,1\}^d}{\left\{\ln \left|\Sigma^{(w)}\right|+\dw\ln(2\pi) + \text{MD}^2(x,w,\mu, \Sigma)+\bpenalt\,(\bone - w)^{\top} \right\}}
\end{equation*}
in which $\bpenalt = (q_1,\ldots,q_d)$ and $w = (w_1,\ldots,w_d)$.
Our estimate is then 
$$\argmin_{\left(\mu, \Sigma\right) \in \Theta} \G(\mu, \Sigma, F_n),$$
with $F_n$ the empirical distribution and $\Theta$ the parameter space. 

Suppose, for now, that $\Theta$ is a subset of $\mathbb{R}^d \times \mathcal{P}(d)$ where $\mathcal{P}(d)$ is the (closed) cone of symmetric positive definite $d \times d$ matrices with smallest eigenvalue $\geqslant a$. We endow the product space with the metric $\Dmax$ given by $\Dmax((\mu_1,\Sigma_1),
(\mu_2,\Sigma_2)) := \max(||\mu_1-\mu_2||_2, ||\Sigma_1 - \Sigma_2||_F)$, which 
combines the Euclidean and Frobenius norms.
Other norms on the space of matrices are possible (and sometimes more natural), but in our context the Frobenius norm suffices. 

\subsection*{A.2.2\;\;Some properties of the objective function}

\begin{lemma}\label{lemma:continuitym}
For all $x$, the function $\Theta \mapsto \mathbb{R}: (\mu, \Sigma) \to  \g_{\mu,\Sigma}(x)$ is continuous. 
\end{lemma}
\begin{proof}
Note that for each fixed $w \in \{0,1\}^d$, the function $\Theta \mapsto \mathbb{R}: (\mu, \Sigma) \to \ln \left|\Sigma^{(w)}\right|+\dw\ln(2\pi) + \text{MD}^2(x,w,\mu, \Sigma)+\bpenalt\,(\bone - w)^{\top}$ is continuous. For any fixed $x$, $\g_{\mu,\Sigma}(x)$ is thus a minimum of a finite number of continuous functions, which is continuous. 
\end{proof}

\begin{lemma}\label{lemma:dominating}
For $\penalt_j > \max\{0, \ln(a)\}$ and all $x$, the function $\g_{\mu,\Sigma}(x)$ is uniformly bounded over all possible $\mu, \Sigma$ in $\mathbb{R}^d \times \mathcal{P}(d)$.
\end{lemma}
\begin{proof}
We have that $d\ln(a) \leqslant \g_{\mu,\Sigma}(x) \leqslant \sum_{j=1}^d{\penalt_j}$\;.\\
The first inequality stems from minimizing each term of $\ln \left|\Sigma^{(w)}\right|+\dw\ln(2\pi) + \text{MD}^2(x,w,\mu, \Sigma)+\bpenalt\,(\bone - w)^{\top}$ individually. In particular, the last three terms are always positive, so we set them to zero. The first term is bounded from below by $d\ln(a)$, which is negative by our choice of $a$.\\
The second inequality stems from the instance $w=(0,\ldots,0)$, in which case $\g_{\mu,\Sigma}(x) = \bpenalt\,(\bone - w)^{\top}=\sum_{j=1}^d{\penalt_j}$.
\end{proof}

\begin{lemma}\label{lemma:continuityM}
The function $\Theta \mapsto \mathbb{R}: (\mu, \Sigma) \to \G(\mu, \Sigma, F)$ is continuous. 
\end{lemma}
\begin{proof}
Denote $(\mu_n, \Sigma_n)_{n\in \mathbb{N}}$ a sequence which converges to $(\mu, \Sigma)$ for $n \to \infty$. We have
\begin{align*}
    \lim_{n\to \infty}{\G(\mu_n, \Sigma_n, F)} 
    =&\lim_{n\to \infty}{\int \g_{\mu_n, \Sigma_n}(x) F(\mathrm{d}x)}\\
    =&\int \lim_{n\to \infty}{\g_{\mu_n, \Sigma_n}(x)} F(\mathrm{d}x)\\
    = &\int \g_{\mu, \Sigma}(x) F(\mathrm{d}x)\\
    = &\; \G(\mu, \Sigma, F)
\end{align*}
where we have used the dominated convergence theorem in the second equality, which is possible thanks to Lemma \ref{lemma:dominating} and the (pointwise) continuity of $\g_{\mu,\Sigma}$ of Lemma \ref{lemma:continuitym} in the third equality.
\end{proof}

\subsection*{A.2.3\;\; Wald-type consistency proof}
% Instead of relying on empirical process theory, 
We can set up a Wald-type consistency proof.
Assume that $\Theta$ is a compact subset of $\mathbb{R}^d \times \mathcal{P}(d)$ where $\mathcal{P}(d)$ is the (closed) cone of symmetric positive definite $d \times d$ matrices with smallest eigenvalue $\geqslant a$\,.
Denote by $\Theta^* = \{(\mu^*, \Sigma^*) \in \Theta: \G(\mu^*, \Sigma^*, F) = \inf_{(\mu, \Sigma) \in \Theta}{\G(\mu, \Sigma, F)}\}$ the set of minima of $\G$. It is nonempty due to compactness of $\Theta$, and the continuity of $G$ from Lemma~\ref{lemma:continuityM}.

\begin{proposition}\label{prop:consistency_Wald}
    Let $(\hat{\mu}_n, \hat{\Sigma}_n)$ be a sequence of estimators which nearly minimize $\G(\cdot, \cdot, F_n)$, i.e. for which $$\G(\hat{\mu}_n, \hat{\Sigma}_n, F_n) \leqslant \G(\mu^*, \Sigma^*, F_n)+o_P(1)$$
    for some $(\mu^*, \Sigma^*) \in \Theta^*$. Then it holds for all $\varepsilon > 0$ and for any compact set $K \subset \Theta$ that
    $$P(\,\Dmax((\hat{\mu}_n, \hat{\Sigma}_n), \Theta^*) \geqslant \varepsilon \; \mbox{ and } \; (\hat{\mu}_n, \hat{\Sigma}_n) \in K\,)\rightarrow 0\,. $$
\end{proposition}
\begin{proof}
    This follows from a direct application of Theorem 5.14 in \cite{VdVaart2000}. Its conditions are satisfied because
    \begin{itemize}
        \item For all $x$, the function $\Theta \mapsto \mathbb{R}: (\mu, \Sigma) \to  \g_{\mu,\Sigma}(x)$ is continuous by Lemma \ref{lemma:continuitym}.
        \item For any sufficiently small ball $U \subset \Theta$ the function $x \to \inf_{(\mu, \Sigma) \in U}{\g_{\mu,\Sigma}(x)}$ is measurable and satisfies $\int \inf_{(\mu, \Sigma) \in U}{\g_{\mu,\Sigma}(x)} F(\mathrm{d}x) > -\infty$. This holds because, by Lemma \ref{lemma:dominating}, we have $\int \inf_{(\mu, \Sigma) \in U}{\g_{\mu,\Sigma}(x)} F(\mathrm{d}x) \geq  d\ln(a)\int F(\mathrm{d}x) > -\infty\;.$
    \end{itemize}
\end{proof}

\subsection*{A.2.4\;\; Compactness}
Proposition \ref{prop:consistency_Wald} gives consistency of the cellMCD estimator to the set of true minimizers of the asymptotic objective. The consistency holds on any compact set  $K \subset \Theta$ of $\mathbb{R}^d \times \mathcal{P}(d)$.
Ideally, we would like the consistency of cellMCD on the whole parameter space, but that in itself is not compact. Fortunately, we can indeed obtain the desired consistency result by virtue of the following proposition.

\begin{proposition}\label{prop:compact}
    Let $\bm x_1, \ldots, \bm x_n$ be a random sample from $F$ with empirical cdf $F_n$. Let $(m_n, S_n)$ be an optimal set of parameters for the sample. Then there exists an $M >0$ so that\linebreak $\lim_{n\to \infty} P((m_n, S_n) \in B(M)) = 1$, where $B(M)$ is the closed ball of radius $M$ around $0$.
\end{proposition}

\begin{proof}

Consider a sequence of solutions $(m_n, S_n)$ minimizing $\G(m_n, S_n, F_n)$.
Denote the diagonal elements of $S_n$ by $\diag(S_n) = (s_{1, n}, s_{2,n}, \ldots, s_{d,n})$ and the components of $m_n$ by $(m_{1,n}, \ldots, m_{d,n})$. We want to show that the diagonal elements of $S_n$ and the components of $m_n$ are bounded eventually: 
$$\exists M_1>0: \lim_{n\to \infty} P((m_n, \diag(S_n)) \in B(M_1)) = 1\;.$$
This suffices because the off-diagonal entries of $S_n$ are bounded by the diagonal entries due to $((S_n)_{ij})^2 \leqslant s_{i,n} s_{j,n}$ since $S_n$ is PSD.

We will first show that the $s_{1, n}, s_{2,n}, \ldots, s_{d,n}$ cannot diverge.
Suppose the opposite, so that w.l.o.g. the first $d^*$ diagonal elements of $S_n$ are not bounded in this way. Therefore $\forall r > 0, \forall j \in \{1, \ldots, d^*\}: \lim_{n\to \infty} P(s_{j,n} \notin B(r)) > 0$.

If $d^* < d$ we denote by $\tilde{m_n}$ and $\tilde{S_n}$ the location and scatter estimates restricted to all but the first $d^*$ coordinates. Denote $$\tg_{m_n, S_n} (x)\coloneqq\min_{w \in \{0, 1\}^{d-d^*}}{\left( \ln \left|\tilde{S}_n^{(w)}\right|+\dw\ln(2\pi) +  \text{MD}^2(x,w,\tilde{m}_n, \tilde{S}_n) + \bpenalt\,(\bone - w)^{\top}\right) }\;.$$
If $d^* = d$ we set the term $\tg_{m_n, S_n} (x)$ equal to zero.
Note that now, for every $\varepsilon >0$, there is a $P_{\varepsilon} > 0$ so that $$\lim_{n\to \infty} P\left(\forall x : \g_{m_n, S_n}(x) >  \left(\sum_{j=1}^{d^*}\penalt_j - \varepsilon\right) + \tg_{m_n, S_n}(x)\right) = P_{\varepsilon}\;.$$
Intuitively this means that, no matter the value of $x$, by inflating $s_1, \ldots, s_{d^*}$ we can get arbitrarily close to the value of the objective function obtained by dropping the first $d^*$ coordinates. The contribution of these first $d^*$ coordinates to the objective function becomes arbitrarily close to $\sum_{j=1}^{d^*}{\penalt_j}$.

Now consider a new sequence of estimates given by $(m_n^*, S_n^*)$ where $m_n* = [\bm{0}_{d^*}, \tilde{m}_n]$ and $S_n^* = \begin{bmatrix}\bm{I}_{d^*} &\bm{0}\\ \bm{0} & \tilde{S}_n\end{bmatrix}$. Note that we must have $\G(m_n, S_n, F_n) \leq  \G(m_n^*, S_n^*, F_n)$ for all $n$, since $(m_n, S_n)$ is a sequence of minimizers of the objective.
For this new sequence of estimates we have
\begin{align*}
    & \g_{m_n^*, S_n^*}(x)\\
    =&\min_{w\in \{0,1\}^d}{\ln \left|S_n^{*(w)}\right|+\dw\ln(2\pi) +  \text{MD}^2(x,w,m_n^*, S_n^*) + \bpenalt'(\bm{1} - w)}\\
     =&   \min_{\underset{\widetilde{}}{w}\in \{0,1\}^{d^*}}{ \ln\left|\underset{\widetilde{}}{S_n}^{*(\underset{\widetilde{}}{w})}\right| + d^{(\underset{\widetilde{}}{w})} \ln(2\pi) + \text{MD}^2(\underset{\widetilde{}}{x},\underset{\widetilde{}}{w},\underset{\widetilde{}}{m_n}^*, \underset{\widetilde{}}{S_n}^*) +\underset{\widetilde{}}{\bpenalt}'(1-\underset{\widetilde{}}{w}) }\\
    +&\min_{\tilde{w}\in \{0,1\}^{d-d^{*}}}{ \ln \left|\tilde{S}_n^{(\tilde{w})}\right|+ d^{(\tilde{w})}\ln(2\pi) +\text{MD}^2(\tilde{x},\tilde{w},\tilde{m}_n, \tilde{S}_n) + \bm{\tilde{\penalt}}'(\bm{1} - \tilde{w})}\\
    =&  \min_{\underset{\widetilde{}}{w}\in \{0,1\}^{d^*}}{\Big( \ln\left|\underset{\widetilde{}}{S_n}^{*(\underset{\widetilde{}}{w})}\right| + d^{(\underset{\widetilde{}}{w})} \ln(2\pi) + \text{MD}^2(\underset{\widetilde{}}{x},\underset{\widetilde{}}{w},\underset{\widetilde{}}{m_n}^*, \underset{\widetilde{}}{S_n}^*) +\underset{\widetilde{}}{\bpenalt}'(1-\underset{\widetilde{}}{w})\Big)}
    +\tg_{m_n, S_n}(x)\\
      =&  \min_{\underset{\widetilde{}}{w}\in \{0,1\}^{d^*}}{ \Big( d^{(\underset{\widetilde{}}{w})} \ln(2\pi) + ||\underset{\widetilde{}}{x}^{(\underset{\widetilde{}}{w})}||_2^2 +\underset{\widetilde{}}{\bpenalt}'(1-\underset{\widetilde{}}{w}) \Big) }+\tg_{m_n, S_n}(x)
\end{align*}
where the tildes below $m_n, S_n, x, w, \penalt$ indicate the restriction of the quantity to the first $d^*$ coordinates, so $\ln|\underset{\widetilde{}}{S_n}^{*(\underset{\widetilde{}}{w})}| = 0$ and $\text{MD}^2(\underset{\widetilde{}}{x},\underset{\widetilde{}}{w},\underset{\widetilde{}}{m_n}^*, \underset{\widetilde{}}{S_n}^*) = ||\underset{\widetilde{}}{x}^{(\underset{\widetilde{}}{w})}||_2^2$\,.

Now denote $C \coloneqq \int  \min_{\underset{\widetilde{}}{w}\in \{0,1\}^{d^*}}{ (d^{(\underset{\widetilde{}}{w})} \ln(2\pi) + ||\underset{\widetilde{}}{x}^{(\underset{\widetilde{}}{w})}||_2^2 +\underset{\widetilde{}}\bpenalt\,(1-\underset{\widetilde{}}{w})^{\top}) }F(\mathrm{d}x)$ and note that $C < \sum_{j=1}^{d^*}\penalt_j$ by the assumptions on $F$. Take $\varepsilon^* < \sum_{j=1}^{d^*}\penalt_j - C$. Then we can find a $P_{\varepsilon^*} > 0$ so that
\begin{align*}
    &\lim_{n\to \infty}P\left(\g_{m_n, S_n}(x) > \left(\sum_{j=1}^{d^*}\penalt_j - \varepsilon^*\right) +  \tg_{m_n, S_n}(x)\right) = P_{\varepsilon^*}\\
    \Rightarrow&\lim_{n\to \infty}P\left(\g_{m_n, S_n}(x) >  C +  \tg_{m_n, S_n}(x)\right) = P_{\varepsilon^*}\\
    \Rightarrow &\lim_{n\to \infty}P\left(\int \g_{m_n, S_n}(x)F_n(\mathrm{d}x) >  C +  \int \tg_{m_n, S_n}(x)F_n(\mathrm{d}x)\right) = P_{\varepsilon^*}\\
        \Rightarrow&\lim_{n\to \infty}P\left(\G(m_n, S_n, F_n) >  \G(m_n^*, S_n^*, F_n) + o_P(1)\right) = P_{\varepsilon^*}> 0
\end{align*}
where the $o_P(1)$ appears because $C = \int  \min_{\underset{\widetilde{}}{w}\in \{0,1\}^{d^*}} (d^{(\underset{\widetilde{}}{w})} \ln(2\pi) + ||\underset{\widetilde{}}{x}^{(\underset{\widetilde{}}{w})}||_2^2 \;+$ \linebreak $\underset{\widetilde{}}\bpenalt\,(1-\underset{\widetilde{}}{w})^{\top}) F_n(\mathrm{d}x) + o_P(1)$ by the law of large numbers. 
We thus obtain a contradiction, since $(m_n, S_n)$ is supposed to be a sequence of minimizers and we find that our newly constructed sequence attains a lower value of the objective function with non-zero probability for $n$ large enough.

Now we know that $s_{1,n}, \ldots, s_{d,n}$ cannot diverge. It remains to show that the same is true for $m_{1,n}, \ldots, m_{d,n}$. Suppose w.l.o.g. that the first $d^*$ elements of $m_n$ are not appropriately bounded, i.e. that  $\forall R > 0: \forall j \in \{1, \ldots, d^*\}: \liminf_{n\to \infty} P(m_{j,n} \notin B(R)) > 0$.
Note that now, for every $\varepsilon >0$, there exists a $ P_{\varepsilon} > 0$ so that  $$ \lim_{n\to \infty} P\left(\g_{m_n, S_n}(x) >  \left(\sum_{j=1}^{d^*}\penalt_j - \varepsilon\right) + \tg_{m_n, S_n}(x)\right) =  P_{\varepsilon}$$
Intuitively this means that, for any fixed value of $x$, by increasing $m_1, \ldots, m_{d^*}$ we can get arbitrarily close to the value of the objective function obtained by dropping the first $d^*$ coordinates. The contribution of these first $d^*$ coordinates to the objective function becomes arbitrarily close to $\sum_{j=1}^{d^*}{\penalt_j}$. It is worth noting the subtle difference from the scale case, where we had the same identity uniformly for all $x$. In this case, we cannot make the exact same statement. However, as long as we bound $||x||$, we can get uniformity. More specifically, we can strengthen the previous statement as follows. For every $\varepsilon >0$ and $R > 0$, we have that there exists a $ P_{\varepsilon, R} > 0$ so that $$ \lim_{n\to \infty} P\left(\forall x \in B(R): \g_{m_n, S_n}(x) >  \left(\sum_{j=1}^{d^*}\penalt_j - \varepsilon\right) + \tg_{m_n, S_n}(x)\right) = P_{\varepsilon, R} > 0.$$ Note that we used here that all  $s_1, \ldots, s_n$ remain bounded (in probability).

We now consider a new sequence of estimates, like before, given by $(m_n^*, S_n^*)$ where $m_n* = [\bm{0}_{d^*}, \tilde{m}_n]$ and $S_n^* = \begin{bmatrix}\bm{I}_{d^*} &\bm{0}\\ \bm{0} & \tilde{S}_n\end{bmatrix}$. Additionally, take $\varepsilon^* < \frac{1}{2}\left(\sum_{j=1}^{d^*}\penalt_j - C\right)$. Then take $R>0$ such that $C + \left(\sum_{j=1}^{d^*}\penalt_j-|d^*\ln(a)|\right) \int_{B(R)^C} F(\mathrm{d}x) < \varepsilon^*\int_{B(R)}F(\mathrm{d}x)$. Then:
\small
\begin{align*}
    &\lim_{n\to \infty}P\left(\forall x \in B(R):\g_{m_n, S_n}(x) >  \left(\sum_{j=1}^{d^*}\penalt_j - \varepsilon^*\right) +  \tg_{m_n, S_n}(x)\right) = P_{\varepsilon^*, R}\\
    \Rightarrow&\lim_{n\to \infty}P\left(\forall x \in B(R):\g_{m_n, S_n}(x) >  (C+\varepsilon*) +  \tg_{m_n, S_n}(x)\right) = P_{\varepsilon^*, R}\\
    \Rightarrow &\lim_{n\to \infty}P\left(\int_{B(R)} \g_{m_n, S_n}(x)F_n(\mathrm{d}x) >  (C+\varepsilon*) \int_{B(R)} F_n(\mathrm{d}x) + \int_{B(R)} \tg_{m_n, S_n}(x)F_n(\mathrm{d}x)\right) = P_{\varepsilon^*, R}\\
    \Rightarrow &\lim_{n\to \infty}P\left(\int \g_{m_n, S_n}(x)F_n(\mathrm{d}x) - \int_{B(R)^C} \g_{m_n, S_n}(x)F_n(\mathrm{d}x) \; F(\mathrm{d}x)>\right. \\
    &\left.C + \varepsilon^* \int_{B(R)} F_n(\mathrm{d}x) +  \int \tg_{m_n, S_n}(x)F_n(\mathrm{d}x) - \int_{B(R)^{C}} C + \tg_{m_n, S_n}(x)F_n(\mathrm{d}x)\right) = P_{\varepsilon^*, R}\\
    \Rightarrow &\lim_{n\to \infty}P\left( \G(m_n, S_n, F_n) - \int_{B(R)^C} \g_{m_n, S_n}(x)F_n(\mathrm{d}x) \; F(\mathrm{d}x)>\right. \\
    &\left.\G(m_n^*, S_n^*, F_n) + o_P(1) + \varepsilon^*\int_{B(R)}F_n(\mathrm{d}x) - \int_{B(R)^{C}} C + \tg_{m_n, S_n}(x)F_n(\mathrm{d}x)\right) = P_{\varepsilon^*, R}\\
    \Rightarrow &\lim_{n\to \infty}P\left( \G(m_n, S_n, F_n)  >\right. \\
    &\left.\G(m_n^*, S_n^*, F_n) + o_P(1) +\varepsilon^*\int_{B(R)}F_n(\mathrm{d}x) - \int_{B(R)^{C}} C+ \tg_{m_n, S_n}(x) - \g_{m_n, S_n}(x)F_n(\mathrm{d}x)\right) = P_{\varepsilon^*, R}\\
    \Rightarrow &\lim_{n\to \infty}P\left( \G(m_n, S_n, F_n)  >\G(m_n^*, S_n^*, F_n) + o_P(1) +A_n\right) = P_{\varepsilon^*, R}
\end{align*}
\normalsize
where $A_n = \varepsilon^*\int_{B(R)}F_n(\mathrm{d}x) - \int_{B(R)^{C}}(C+ \tg_{m_n, S_n}(x) - \g_{m_n, S_n}(x))F_n(\mathrm{d}x)$. First note that $|\tg_{m_n, S_n}(x) - \g_{m_n, S_n}(x)| \leq \left(\sum_{j=1}^{d^*}\penalt_j-|d^*\ln(a)|\right)$. Therefore, $P(A_n > 0) \to 1$ by our choice of $R$ and the law of large numbers. 
Therefore we obtain
\begin{equation*}
    \lim_{n\to \infty}P\left( \G(m_n, S_n, F_n)  >\G(m_n^*, S_n^*, F_n) + o_P(1)\right) = P_{\varepsilon^*, R}\;.
\end{equation*}
This is again a contradiction, since $(m_n, S_n)$ is supposed to be a sequence of minimizers and we find that our newly constructed sequence attains a lower value for the objective function with nonzero probability for $n$ large enough. 
\end{proof}

From Propositions \ref{prop:consistency_Wald} and \ref{prop:compact}, 
we now obtain Proposition~\ref{prop:consistency} in the paper.

\subsection*{A.2.5\;\; Fisher consistency of the cellMCD location estimator}

The previous parts were concerned with the consistency of 
the estimators for the set of population minimizers.
The population minimizer for $\Sigma$ is not quite the 
underlying parameter, since a small fraction of cells is always
given weight zero due to the penalty term in the objective. 
But for the location $\mu$ we can prove that the unique
minimizer is indeed the underlying parameter vector, so
the cellMCD functional for location is Fisher consistent.
Below we will keep $\Sigma$ fixed at its minimizer, so only
$\mu$ varies. We furthermore assume that $F$ is a strictly
unimodal elliptical distribution which allows a density 
function. Finally, we assume w.l.o.g. that the center of
symmetry of $F$ is $\mu = 0$.

The following Lemma states the relevant properties of the function $\g_{\mu,\Sigma}$\,. We will use the notation $l_w(x, \mu, \Sigma) \coloneqq \ln \left|\Sigma^{(w)}\right|+\dw\ln(2\pi) + \text{MD}^2(x,w,\mu, \Sigma)+\bpenalt\,(\bone - w)^{\top}$.
\begin{lemma}\label{lemma:gfunction}
The function $\g_{\mu,\Sigma}: \mathbb{R}^d \to \mathbb{R}: x \to\g_{\mu,\Sigma}(x)$:
\begin{enumerate}
    \item is minimal in $x = \mu$;
    \item This minimum is unique as long as $\exists\, \delta > 0$ such that $\forall x \in B(\mu, \delta)$ it holds that $\argmin_w  l_w(x, \mu, \Sigma) = \bm{1}_d$\,;
    \item only shifts when $\mu$ changes, i.e. $\g_{\mu,\Sigma}(x) = \g_{0,\Sigma}(x-\mu)$\;
    \item is point symmetric around $\mu$ and, for every $v \in S^{d-1}$, weakly monotone increasing in $||(x-\mu)'v||$.
\end{enumerate}
\end{lemma}
\begin{proof}
Suppose we fix $\mu$ and $w$ for a moment. Note that $l_w$\,, as a function of $x = (x_1,\ldots, x_d)$, has the following properties:
\begin{itemize}
    \item it is quadratic in those $x_j$ for which $w_j = 1$, and constant in the other $x_j$. It is thus strictly monotone increasing in $||x^{(w)}-\mu^{(w)}||$;
    \item it is a point symmetric function in $x-\mu$;
    \item it is minimal in $x_j = \mu_j$ for all $x_j$ for which $w_j = 1$. So, unless $w = 1$ for all $j$, the minimum is not unique;
    \item changing $\mu$ only shifts this function.
\end{itemize}

So, each function $l_w$ is a quadratic function with a minimum at $x = 0$. This minimum is unique only if $w_j = 1$ for all $j$, in other cases we have some dimensions in which the minimum is not unique (the function is constant there).
%Let's denote these functions as $h_1(\mu, x), \ldots, h_{2^d}(\mu, x)$.
Now the function we are interested in, is 
$$\g_{\mu,\Sigma}: \mathbb{R}^d \to \mathbb{R}: x \to \g_{\mu,\Sigma}(x) = \min_{w \in \{0,1\}^d} l_w(x, \mu, \Sigma).$$
%$$\g_{\mu,\Sigma}(x) \coloneqq \min_j h_j(\mu, x).$$

The first claim now immediately follows. Since each $l_w$ is minimized in $x = \mu$, this also holds for the minimum of these functions.
Now, this minimum need not be unique in principle. However, if $\exists\, \delta > 0$ such that $\forall x \in B(\mu, \delta)$ it holds that $\argmin_w l_w(x, \mu, \Sigma) = \bm{1}_d$, then we know that we have an open ball around $\mu$ so that for all $x$ in this ball, $w = \bm{1}_d$ attains the lowest value for $l$. In that case, we do have that $x=\mu$ is a unique minimizer of $g_{\mu, \Sigma}$. Intuitively, if we have a region of observations (centered around $\mu$) where no cells are flagged, we obtain a unique minimum for $g_{\mu, \Sigma}$.\\
The third property follows from the fact that for each of the $l_w$ we have that $l_w(x, \mu, \Sigma) = l_w(x-\mu, 0, \Sigma)$, hence it also holds for $g_{\mu,\Sigma}$\;.\\
Finally, $g_{\mu, \Sigma}$ is point symmetric around $\mu$ and weakly monotone increasing in $||(x-\mu)'v||$, because all $l_w$ have these same properties.
\end{proof}

Now that we have the relevant properties of $\g_{\mu,\Sigma}$, we need to prove that the expected value of this function w.r.t. a strictly unimodal elliptical density function centered at 0 is minimized at $\mu = 0$. For this, we first show this in the univariate case.

\begin{lemma}[univariate case]\label{lemma:univariate}
    Let $g_0: \mathbb{R} \rightarrow \mathbb{R}$ be a symmetric function around the origin and assume there is a $\delta > 0$ such that $g_0(|x|)$ is strictly increasing for $|x| < \delta$ and monotone increasing for
    $|x| \geqslant \delta$. Put $g_{\mu}(x) \coloneqq g_0(x-\mu)$. Let $f$ be a strictly unimodal density function symmetric around 0. When $g_{\mu}(x)f(x)$ is integrable, the integral
        $$\int g_{\mu}(x)f(x)\mathrm{d}x$$
    attains its unique minimum at $\mu = 0$.
\end{lemma}
\begin{proof}
Note that 
%\small
\begin{align*}
    &\int g_{\mu}(x)f(x)\mathrm{d}x - \int g_{0}(x)f(x)\mathrm{d}x \\
     =&\int g_{0}(x-\mu)f(x)\mathrm{d}x - \int g_{0}(x)f(x)\mathrm{d}x \\
     =&\int g_{0}(x-\mu/2)f(x+\mu/2)\mathrm{d}x - \int g_{0}(x-\mu/2)f(x-\mu/2)\mathrm{d}x \\
     =&\int g_{0}(x-\mu/2)\left\{f(x+\mu/2) - f(x-\mu/2)\right\}\mathrm{d}x\\
     =&\int_{\mathbb{R}^-} g_{0}(x-\mu/2)\left\{f(x+\mu/2) - f(x-\mu/2)\right\}\mathrm{d}x  \;+\\
     &\int_{\mathbb{R}^+} g_{0}(x-\mu/2)\left\{f(x+\mu/2) - f(x-\mu/2)\right\}\mathrm{d}x   \\
     =&\int_{\mathbb{R}^+} g_{0}(-x-\mu/2)\left\{f(-x+\mu/2) - f(-x-\mu/2)\right\}\mathrm{d}x \; + \\
     & \int_{\mathbb{R}^+} g_{0}(x-\mu/2)\left\{f(x+\mu/2) - f(x-\mu/2)\right\}\mathrm{d}x   \\
     %=&\int_{\mathbb{R}^+} \left\{g_{0}(-x-\mu/2) -  g_{0}(x-\mu/2)\right\}\left\{f(-x+\mu/2) - f(-x-\mu/2)\right\}\mathrm{d}x \\
     =&\int_{\mathbb{R}^+} \left\{g_{0}(x+\mu/2) - g_{0}(x-\mu/2)\right\}\left\{f(x-\mu/2) - f(x+\mu/2)\right\}\mathrm{d}x
    \end{align*}
%\normalsize
where we have used the symmetry of $f$ in the last equality.
    
 If $\mu \geqslant 0$, then the last line is $\geqslant 0$ because both factors are $\geqslant 0$ due to $g_0$ being monotone increasing and symmetric and $f$ being unimodal and symmetric, and $|x + \mu/2| \geqslant |x-\mu/2|$ for $x \geqslant 0$. For $\mu > 0$ we obtain $|x + \mu/2| > |x-\mu/2|$ in all $x > 0$ so almost everywhere, and by strict unimodality of $f$ the integrated inequality is strict. If $\mu \leqslant 0$ then the last line is  still $\geqslant 0$, as now both factors are $\leqslant 0$ because of the same reasons. So we find
    $$\int g_{\mu}(x)f(x)\mathrm{d}x - \int g_{0}(x)f(x)\mathrm{d}x \geqslant 0$$
for all $\mu$. Note that the equality is reached for $\mu = 0$, and this is the unique minimizer due to the strict monotonicity of $g_0$ in its central region and the strict unimodality of $f$.
\end{proof}

Now we need a multivariate version of the above, which is stated below:

\begin{proposition}[multivariate case]\label{prop:multivariatemu}
    Let the function $g_0: \mathbb{R}^d \rightarrow \mathbb{R}$ be point symmetric around the origin and assume there is a $\delta > 0$ such that for every direction $v$ on the unit sphere $S^{d-1}$ it holds that $g_0(tv)$ is strictly increasing for $0 \leqslant t < \delta$ and weakly monotone increasing for $t \geqslant \delta$. Put $g_{\mu}(x) \coloneqq g_0(x-\mu)$.
    Let $f$ be a strictly unimodal density function which is 
    elliptical around 0. When $g_{\mu}(x)f(x)$ is integrable, 
          $$\int g_{\mu}(x)f(x)\mathrm{d}x$$
    attains its unique minimum at $\mu = 0$.
\end{proposition}
\begin{proof}
Note that
$$\int g_{\mu}(x)f(x)\mathrm{d}x =\int g_{0}(x)f(x+\mu)\mathrm{d}x\;.$$
By switching to hyperspherical coordinates this multivariate integral
becomes
$$\frac{1}{2} \int_{v \in S^{d-1}}\left\{ 
     \int_{t \in \mathbb{R}} g_{0}(t v)f(t v+\mu) |t|^{d-1} \mathrm{d}t \right\}d\eta(v)$$
where $\eta$ is the uniform probability measure on the unit sphere $S^{(d-1)}$. This is a change of variables: x is written as $t v$ with $t \in \mathbb{R}$ and $v \in S^{d-1}$. The factor $|t|^{d-1}$ is the Jacobian.

Now consider the inner integral
$$\int_{t \in \mathbb{R}} g_{0}(t v)f(t v+\mu) 
  |t|^{d-1} \mathrm{d}t = \int_{t \in \mathbb{R}} 
  h(t) f(t v+\mu) \mathrm{d}t$$
where the univariate function $h(t) \coloneqq g_0(t v)|t|^{d-1}$ is symmetric around $t = 0$ and monotone increasing in $|t|$, and even strictly monotone increasing for $|t| < \delta$.

The other function in the inner integral is $f(t v+\mu)$ where $t v+\mu$ forms a straight line. Due to the properties of $f$ this function is symmetric about $t_0 = -\mu'v$ and strictly unimodal. It is in fact
a constant multiple of the conditional density on that line.
Consider the  univariate function $f^{(v)}(t) \coloneqq f(tv)$. Then $f(t v + \mu) = f^{(v)}(t - t_0)$ is a univariate function of $t$. So the entire inner integral becomes
$$ \int_{t \in \mathbb{R}} 
  h(t) f^{(v)}(t - t_0)\, \mathrm{d}t =
  \int_{t \in \mathbb{R}} 
  h(t - (-t_0)) f^{(v)}(t)\, \mathrm{d}t\;.$$
To this integral we can apply Lemma~\ref{lemma:univariate}, which tells us that the integral is minimized when $t_0 = 0$ and that this minimizer is unique.
So we know that for any direction $v \in S^{d-1}$ the inner integral is minimal
when $\mu'v = -t_0 = 0$. Therefore the entire integral is minimal when $\mu = 0$. This minimizer is unique because the only vector $\mu$ that is orthogonal to every direction $v$ on the unit sphere is the origin. 
\end{proof}

Proposition~\ref{prop:Fisherconsistency} in the paper now 
follows from Lemma~\ref{lemma:gfunction} and 
Proposition~\ref{prop:multivariatemu}. It covers typical 
model distributions $f$ such as multivariate Gaussians 
and elliptical $t$-distributions.

\clearpage
%%%%%%%%%%%%%%%%%%%%%%%%%%%%%%%%%%%%%%%%%%%%%%%%%%%%%%%%%%%%%%%%%
\section*{A.3\;\;About the algorithm in Section~\ref{sec:algo}}

\vspace{1cm} % These are results about the algorithm.

\begin{proof}[Proof of Proposition~\ref{prop:split}.]
Put $\bmu=\bzero$ without loss of generality. 
Following Petersen and Pedersen (2012), p. 47, 
we can write 
  \[ \bSigma^{-1}=\bA \bB \bA^\top\]
with 
  \[\begin{array}{ll}
  \bA=\begin{bmatrix}
  \bI & \bzero\\
  -\bSigma_{22}^{-1} \bSigma_{21} & \bI
\end{bmatrix}
\;\;\mbox{ and }
& \bB=\begin{bmatrix}
  \bC_1^{-1} & \bzero\\
  \bzero & \bSigma_{22}^{-1}
  \end{bmatrix}\;.
  \end{array}\]
Note that 
  \[\bx^\top\bA=\begin{bmatrix} \bx_1^\top-
  \bx_2^\top\bSigma_{22}^{-1}\bSigma_{21} & 
	\bx_2^\top\end{bmatrix} = 
  \begin{bmatrix}\bx_1^\top-\bhx_1^\top &
  \bx_2^\top \end{bmatrix}\]
and so
\begin{align*}
  \MD^2(\bx,\bzero,\bSigma) &=
  \bx^\top \bSigma^{-1} \bx = 
	(\bx^\top \bA) \bB(\bx^\top \bA)^\top\\
  &= (\bx_1-\bhx_1)^\top \bC_1^{-1}
	   (\bx_1-\bhx_1)+ 
		 \bx_2^\top\bSigma_{22}^{-1} \bx_2\\
  &= \MD^2(\bx_1,\bhx_1,\bC_1) + 
	   \MD^2(\bx_2,\bzero,\bSigma_{22})\;.
\end{align*}
For \eqref{eq:addlikelihood}, we verify that 
  \[ |\bSigma^{-1}|=|\bA|\,|\bB|\,|\bA|=
  1\,|\bC_1^{-1}|\,|\bSigma_{22}^{-1}|\,1\]
so \[|\bSigma|=|\bC_1|\,|\bSigma_{22}|\;.\]
Finally,
\begin{align*}
  & L(\bx,\bmu,\bSigma)-
	L(\bx_2,\bmu_2,\bSigma_{22})\\ 
	&= \MD^2(\bx,\bmu,\bSigma)-
	   \MD^2(\bx_2,\bmu_2,\bSigma_{22})+
		 d\ln(2\pi) - d^{(\bx_2)}\ln(2\pi) + 
		 \ln |\bSigma|-
		 \ln |\bSigma_{22}|\\
  &= \MD^2(\bx_1,\bhx_1,\bC_1)+
	   d^{(\bx_1)}\ln(2\pi) + \ln |\bC_1|
		=L(\bx_1,\bhx_1,\bC_1)\;.
\end{align*} \end{proof}

\newpage
\noindent{\bf Pseudocode of the cellMCD algorithm}

\vspace{0.6cm}
For the purpose of clarity, we assume throughout the pseudocode algorithms that $n, d, q_j$, and $h$ are global constants and that the input data has already been standardized robustly. The function \texttt{getObjective} simply computes and returns the cellMCD objective given the current estimates of the parameters.\\

\begin{algorithm}[hbt!]
\caption{The cellMCD algorithm}\label{alg:cellMCD}
\begin{algorithmic}%[1]
\Require A dataset $\bX \in \mathbb{R}^{n \times d}$, initial estimates $\hmu_{\mbox{\small{init}}}$ and $\hSigma_{\mbox{\small{init}}}$ of location and scatter and a matrix $W_{\mbox{\small{init}}} \in \{0,1\}^{n \times d}$ of flagged cells.
\State $\hmu \gets \hmu_{\mbox{\small{init}}}$
\State $\hSigma \gets \hSigma_{\mbox{\small{init}}}$
\State $W \gets W_{\mbox{\small{init}}}$
\State $\texttt{Objective} \gets \texttt{getObjective}(\bX, \hmu, \hSigma, W)$ \Comment{Initial value of the objective}

\While{\texttt{!Converged}}
\State $W \gets \texttt{updateW}(\bX, \hmu, \hSigma, W)$  \Comment{Part a of the C-step}
\State $(\hmu, \hSigma) \gets \texttt{updatemS}(\bX, \hmu, \hSigma, W)$  \Comment{Part b of the C-step}
\State $\texttt{Converged} \gets  (\texttt{Objective} - \texttt{getObjective}(\bX, \hmu, \hSigma, W)) < 1\text{e-}10$
\State $\texttt{Objective} \gets \texttt{getObjective}(\bX, \hmu, \hSigma, W)$
\EndWhile
\State \textbf{return} $(\hmu, \hSigma, W)$
\end{algorithmic}
\end{algorithm}

\begin{algorithm}[hbt!]
\caption{Part a of the C-step}\label{alg:Cstepa}
\begin{algorithmic}%[1]
\Function{updateW}{$\bX, \hmu, \hSigma, W$}
\State $\texttt{ordering} \gets\texttt{order}( |\bX| \bm{1}_d)$  \Comment{Order of variable updates}
\For{j $\in$ \texttt{ordering}}\Comment{Cycle through the variables in order}
\State $\Delta  \gets \bm{0}_d$
\For{i= 1:n}
\State $\texttt{obs} \gets (W_{i,\cdot} == 1)$ \Comment{Indices of preserved cells}
\State $\texttt{xhat} \gets \hmu_{j} + \hSigma_{j,\texttt{obs}} \left(\hSigma_{\texttt{obs},\texttt{obs}}\right)^{-1} (\bX_{i, \texttt{obs}} - \hmu_{\texttt{obs}}) $ \Comment{Conditional expectation}
\State $C  \gets \hSigma_{j,j} - \hSigma_{j,\texttt{obs}} \left(\hSigma_{\texttt{obs},\texttt{obs}}\right)^{-1} \hSigma_{\texttt{obs},j}$ \Comment{Conditional variance}
\State $\Delta_{i}  \gets \ln(C) + \ln(2\pi) + (\bX_{ij} - \texttt{xhat})^2/C - \penalt_j$
\EndFor
\State $\texttt{cutoff} \gets \max\{\Delta_{(h)}, 0\}$ \Comment{Flag no more than $n-h$ positive Deltas}
\For{i= 1:n}
\If{$\Delta_{i} \geq \texttt{cutoff}$}
\State $W_{i, j} \gets 0$
\Else
\State $W_{i, j} \gets 1$
\EndIf
\EndFor
\EndFor
\State \textbf{return} $W$
\EndFunction
\end{algorithmic}
\end{algorithm}

\begin{algorithm}[hbt!]
\caption{Part b of the C-step}\label{alg:Cstepb}
\begin{algorithmic}%[1]
\Function{updatemS}{$\bX, \hmu, \hSigma, W$}
\State $\bX_{\texttt{imp}}\gets \bX$ \Comment{Initialize imputed data matrix}
\State $B\gets \bm{0}^{d \times d}$ \Comment{Initialize bias correction matrix}
\For{i= 1:n}
\State $\texttt{mis} \gets (W_{i,\cdot} == 0)$ \Comment{Indices of flagged cells}
\State $\texttt{obs} \gets (W_{i,\cdot} == 1)$ \Comment{Indices of preserved cells}
\State $\bX_{\texttt{imp}, i, \texttt{mis}} \gets \hmu_{\texttt{mis}} + \hSigma_{\texttt{mis},\texttt{obs}} \left(\hSigma_{\texttt{obs},\texttt{obs}}\right)^{-1} (\bX_{i, \texttt{obs}} - \hmu_{\texttt{obs}}) $
\State $B_{\texttt{mis}, \texttt{mis}} \gets B_{\texttt{mis}, \texttt{mis}}  + \left(\hSigma_{\texttt{mis},\texttt{mis}} - \hSigma_{\texttt{mis},\texttt{obs}} \left(\hSigma_{\texttt{obs},\texttt{obs}}\right)^{-1} \hSigma_{\texttt{obs},\texttt{mis}}\right)$
\EndFor
\State $\hmu \gets \frac{1}{n} \bX_{\texttt{imp}} \bm{1}_d $
\State $\hSigma \gets \frac{1}{n} (\bX_{\texttt{imp}}  - \bm{1}_n \; \hmu')'(\bX_{\texttt{imp}}  - \bm{1}_n\; \hmu') + \frac{1}{n}B$
\State \textbf{return} $(\hmu, \hSigma)$
\EndFunction
\end{algorithmic}
\end{algorithm}

\clearpage
\begin{proof}[Proof of Proposition~\ref{prop:Cstep}.]
We first prove statement (i).
Part (a) of the C-step repeatedly updates
one column of $\BtW$, say column $j$. 
It sets $\Btw_{ij}=1$ for all $i$ with 
negative $\Delta_{ij}$. 
If that number exceeds $h$ the constraint is
satisfied, and otherwise it takes the $h$ 
smallest values of $\Delta_{ij}$.
In either case we obtain the lowest
sum of the terms of the 
objective~\eqref{eq:cellMCD} in column $j$
that satisfies the constraint, so that sum
has to be less than or equal to before.
This remains true after repeating the
procedure on other columns.

Part (b) starts by performing the standard
E-step.
Next, the M-step is carried out and the
constraint $\eig_d(\bhSigma) \geqslant a$ is
applied by truncating all eigenvalues of
$\bhSigma$ at $a$ from below.
This combination nevertheless reduces the 
objective~\eqref{eq:cellMLE} or keeps it 
the same, following section 11.3 of 
\cite{LR2020} on the Gaussian model with
a restricted covariance matrix.
This is because the E-step is unchanged, 
whereas the constraint acts on the M-step 
which is the same as if the result of the 
E-step came from complete data.
For our specific constraint this was also
shown in Proposition 1 of 
Aubry et al. (2021), see in particular 
their formulas (33) and (34).
Since the objective~\eqref{eq:cellMLE}
is reduced or stays the same, this also
follows for the total 
objective~\eqref{eq:cellMCD}.

We now prove statement (ii).
The algorithm iterates C-steps, and 
converges because the objective decreases 
in each C-step (when it remains the same 
the algorithm is done) and there is a 
finite lower bound on the 
objective~\eqref{eq:cellMCD}.
To see the latter, first consider a fixed 
matrix $\bW$. 
Then the first term satisfies
$\ln(|\bSigma^{(w_i)}|) = 
\ln(\prod_j{\eig_j(\bSigma^{(w_i)})})
=
\sum_j{\ln(\eig_j(\bSigma^{(w_i)}))}
\geqslant
||w_i||_0 \ln(\eig_d(\bSigma))
\geqslant ||w_i||_0 \ln(a)$
which is finite, and all the other
terms are bounded below by zero.
The overall lower bound is the minimum
of such lower bounds over the finite 
number of possible matrices $\bW$
that satisfy the constraint, 
so it is finite.
\end{proof}

\clearpage
%%%%%%%%%%%%%%%%%%%%%%%%%%%%%%%%%%%%%%%%%%%%%%%%%%%%%%%%%%%
\section*{A.4\;\; The initial estimator DDCW}
\label{A:DDCW}

The C-step iterations of section
\ref{sec:algo} need initial cellwise robust
estimates $\bhmu^0$ and $\bhSigma^0$ of 
location and covariance.
For this purpose we developed an initial 
estimator called DDCW, described here.
Its steps are:
\begin{enumerate}
\item Drop variables with too many missing values
      or zero median absolute deviation, and
			continue with the remaining columns.
\item Run the DetectDeviatingCells (DDC) method
      \citep{DDC2018} with the constraint
			that no more than $n - h$ cells are
			flagged in any variable.
			DDC also rescales the variables, and may
			delete some cases.
			Continue with the remaining imputed and
			rescaled cases denoted as $\bz_i$\,.
\item Project the $\bz_i$ on the axes of
      their principal components, yielding the 
			transformed data points	$\btz_i$\;.
\item Compute the wrapped location $\bhmu_w$
      and covariance matrix $\bhSigma_w$
			\citep{raymaekers2021fast} of these 
			$\btz_i$\,. 
			Next,	compute the temporary points
			$\bu_i = (u_{i1},...,u_{id})$
			given by $u_{ij} = \max(\min(
			\tilde{z}_{ij}-(\bhmu_w)_j,2),-2)$.
			Then remove all cases for which the 
			squared robust distance
			$\RD^2(i) = 
			 \bu_i' \bhSigma_w^{-1} \bu_i$
			exceeds\; 
			$\chi^2_{d,0.99} 
			 \median_h(\RD^2(h))/\chi^2_{d,0.5}$\;.
\item Project the remaining $\btz_i$ on the
      eigenvectors of $\bhSigma_w$ and
			again compute a wrapped location and
			covariance matrix.
\item Transform these estimates back to
      the original coordinate system of
			the imputed data, and undo the scaling.
			This yields the estimates 
			$\bhmu^0$ and $\bhSigma^0$\,.
\end{enumerate}
Note that DDCW can handle missing values since
the DDC method in Step 2 imputes them.
The reason for the truncation in the rejection
rule in Step 4 is that otherwise the robust
distance $\RD$ could be inflated by a single
outlying cell.
Step 4 tends to remove rows which deviate
strongly from the covariance structure. 
These are typically rows which cannot be shifted 
towards the majority of the data without 
changing a large number of cells.

%%%%%%%%%%%%%%%%%%%%%%%%%%%%%%%%%%%%%%%%%%%%%%%%%%%%%%%%%%
\section*{A.5\;\; Choice of the tuning constant}

The $q_j$ in the objective function~\eqref{eq:cellMCD}
are given by expression~\eqref{eq:lambdaj} which contains 
the single tuning constant $p$. This tuning constant
determines how many cells are flagged, which has 
consequences for efficiency and robustness. Therefore, we
have to choose its default value carefully.

Based on the discussion around~\eqref{eq:deltaij},  
the condition for flagging a cell $x_{ij}$ is 
 \begin{align*}
\ln(C_{ij}) + \ln(2\pi) +
	   (x_{ij} - \hx_{ij})^2/C_{ij}
		 > \penalt_j
\end{align*}
where the scalars $\hx_{ij}$ and $C_{ij}$ are the estimated conditional 
mean and variance of the cell $X_{ij}$ given the observed cells 
in row $i$, i.e. those with $\tw_{i\cdot}=1$. 
Together with our choice of $q_j$ in~\eqref{eq:lambdaj}, we 
see that $x_{ij}$ is flagged if and only if 
  $$\frac{(x_{ij} - \hx_{ij})^2}{C_{ij}}\, > \chi^2_{1,p}\;,$$ 
i.e. its squared conditional residual exceeds the $p$-th 
quantile of the chi-square distribution.

There is no simple analytic expression for the 
population cellMCD covariance matrix. We can write it as 
the minimizer of the objective function, as we have 
done in~\eqref{eq:cellMCD}.
For a given weight matrix $\bW$ it satisfies (keeping 
$\bmu = \bzero$ fixed for simplicity of notation):
$$\hSigma_{jk} = \frac{1}{n}\sum_{i=1}^{n} \by_{ij} \by_{ik} + c_{jki}$$
where $$\by_{ij}= \begin{cases}
    \bx_j &\mbox{ if } \bw_{ij} = 1\\
    E[\bx_j|\bw_i, \bhSigma] &\mbox{ if } \bw_{ij} = 0\\
\end{cases}$$ and
$$c_{jki}= \begin{cases}
    0 &\mbox{ if } \bw_{ij} = 1 \mbox{ or } \bw_{ik} = 1\\
    \Cov[\bx_{ij},\bx_{ik}|\bw_i, \bhSigma] 
    &\mbox{ if } \bw_{ij} =  \bw_{ik}=0\\
\end{cases}$$
This is the maximum likelihood estimate for incomplete data, which is consistent for cells missing completely at random, but here the $w_{ij}$ are not of that type since they correspond to cells that were flagged due to being extreme in some sense.

From these formulas, the cellMCD covariance matrix can be seen as a classical covariance matrix computed on the imputed data, with an additional correction. If the imputed data is the original data, i.e. if we flag no cells, we recover the classical covariance matrix.

To illustrate how the flagging of cells depends on the choice of $p$ we look at regions where one or both cells are flagged. For this we considered a bivariate normal distribution with center $\bmu=\bzero$. The diagonal entries of its scatter matrix $\bSigma$ are $1$, and its off-diagonal entries equal $\rho=0.9$. The resulting ``domains of attraction'' are shown in the figure below, for different values of $p$. In the central region no cells are flagged. In the horizontal region the first cell is flagged, and in the vertical region the second cell is. In the `corner' regions both cells are flagged. We see that the central region expands with $p$.\\

\begin{figure}[!ht]
    \centering
\includegraphics[width = 0.22\columnwidth]{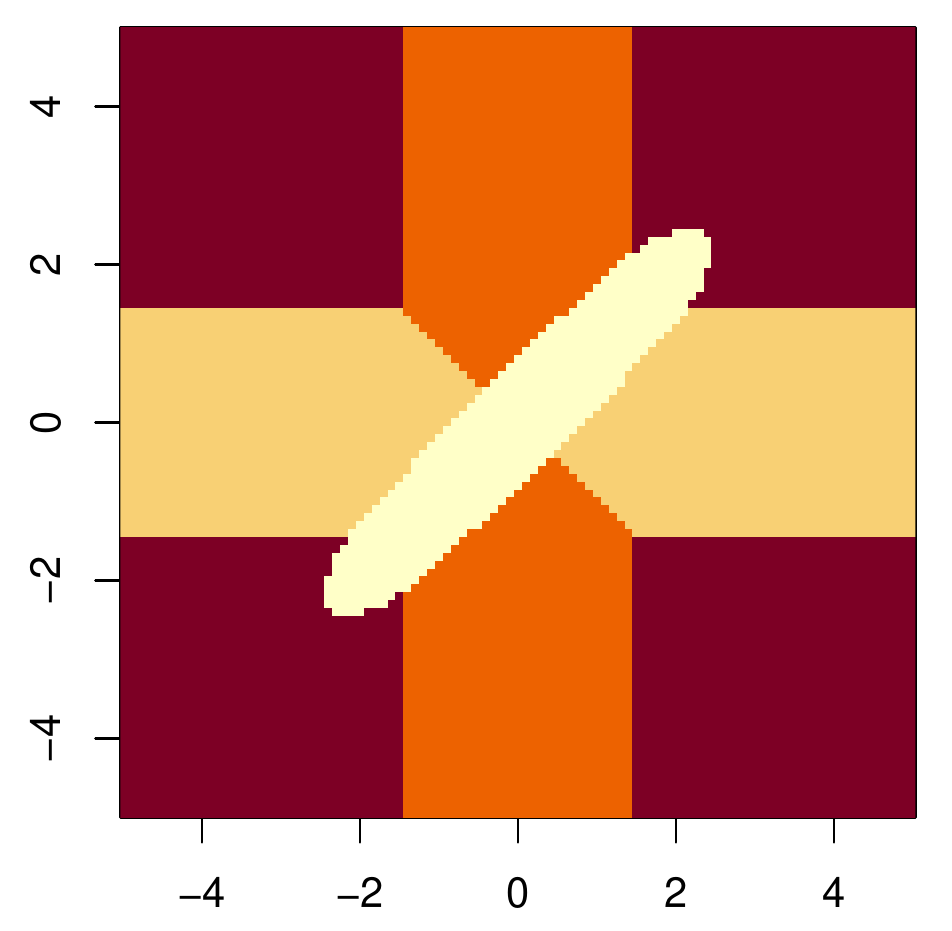}
\includegraphics[width = 0.22\columnwidth]{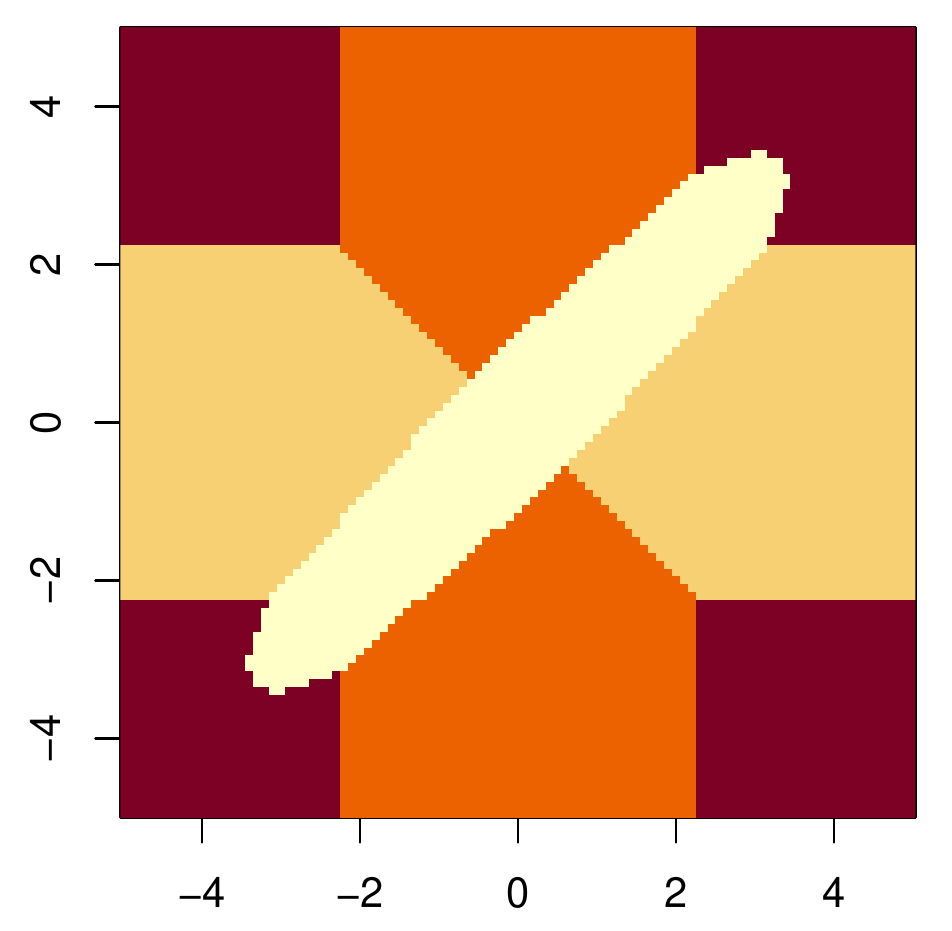}
\includegraphics[width = 0.22\columnwidth]{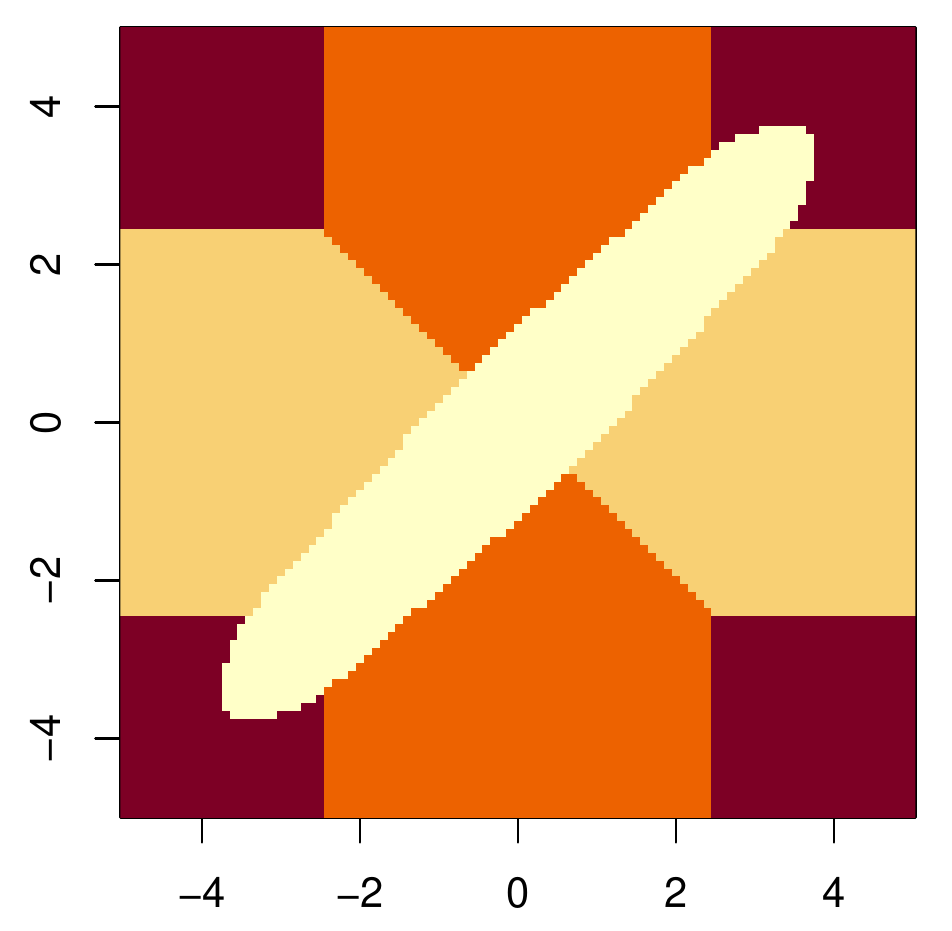}
\includegraphics[width = 0.22\columnwidth]{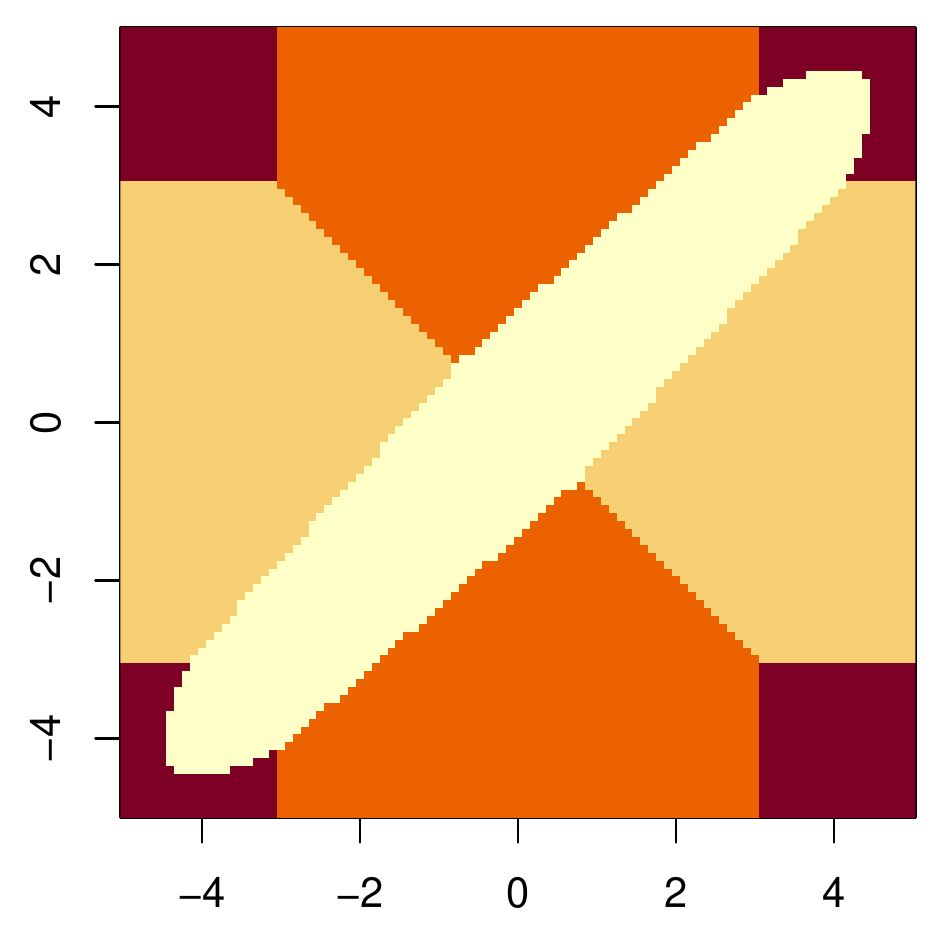}\\
\vspace{-4mm}
\caption{Domain of attraction for $\rho = 0.9$ and 
        quantiles of 0.95, 0.99, 0.995, 0.999.}
\label{fig:rho9}
\end{figure}

As long as $\hx_{ij}$ and $C_{ij}$ remain bounded (which happens under a wide variety of distributions, including heavy-tailed and contaminated ones, due to the good breakdown value), this implies that the cellMCD estimator will yield the classical maximum likelihood estimator of the covariance matrix if $p$ is taken large enough.

This might tempt us to look for a rate at which $p$ can diverge while achieving estimation consistency. 
This would be similar to letting the tuning parameter in Huber-type estimators diverge as in \cite{Sun2020}.
This can work well under specific assumptions on the contamination. 
The drawback of this exercise for us, and the reason why we chose not to pursue this direction, is that the breakdown value will be lost no matter the rate at which $p$ diverges. 
To see why the breakdown value is lost, let $p(n)$ tend to 1 at some rate depending on the sample size $n$. Then we can replicate part (d) of the proof of Proposition 2 on the breakdown value in Section A.1 of the Supplementary Material, where we can pick a fixed percentage of cells, say $\varepsilon n$ for some $0 < \varepsilon < (n-h)/n$, and set them equal to the sequence $c_n$ where $c_n$ does not diverge too fast, while $\bw_{ij} =1$. It suffices to take $c_n = \mathrm{o}\left(\sqrt{\chi^2_{1,p(n)}}\right)$. But then the method would break down.

The default choice of $p=0.99$ was guided by a tradeoff between robustness and efficiency. This is a typical approach in robust statistics, for instance when choosing the tuning constant of Huber's M-estimator or Tukey's bisquare. In Figures~\ref{fig:KLdiv_lowdim} and~\ref{fig:KLdiv_highdim} in Section~\ref{sec:simul} of the paper we saw that the cellMCD method with this $p$ is very robust to outliers, and Table 2 showed its good finite-sample efficiency. We can also look at the efficiency for varying $p$. The table below shows some approximate large-sample efficiencies as a function of $p$. They were obtained by repeatedly generating $n=10000$ data points from the multivariate normal distribution in dimension $d = 3$ with covariance matrix of type A09, and running cellMCD with cutoff $q_j$ given by different $p$-th quantiles. The variances of the entries of the resulting matrices $\hSigma$ were then compared to those of the classical covariance matrix. This rough result illustrates that the efficiency goes up with increasing $p$, and reaches a satisfactory value for $p=0.99$.\\ 

\begin{table}[h!]
\centering
\begin{tabular}{c|cccccc}
\hline
quantile $p$ & 0.95 & 0.975 & 0.99 & 0.995 & 0.999 & 0.9999\\
efficiency & 0.84 & 0.92 & 0.93 & 0.94 & 0.99 & 0.99\\
\hline
\end{tabular}
\label{tab:eff_p}
\caption{Approximate efficiency of the entries of $\bhSigma$ 
         as a function of the quantile $p$.}  
\end{table}

%%%%%%%%%%%%%%%%%%%%%%%%%%%%%%%%%%%%%%%%%%%%%%%%%%%%%%%%%%%
%\clearpage
\section*{A.6\;\; More simulation results}
% The * makes this section unnumbered
\label{A:simul}

%%%%%%%%%%%%%%%%%%%%%%%%%%%%%%%%%%%%%%%%%%%%%%%%%%%%%%%%%%%
\subsection*{A.6.1\;\;Variability of the simulation results}
\label{A:simulsdv}

A reviewer asked for the variability of the Kullback-Leibler
discrepancy in the simulation, the averages of which are shown 
in Figures~\ref{fig:KLdiv_lowdim} and~\ref{fig:KLdiv_highdim} 
in the paper. Their standard deviations are listed in 
Table~\ref{cellMCD_sds_n_corrTypeALYZ_eps10} below for ALYZ,
and in Table~\ref{cellMCD_sds_n_corrTypeA09_eps10} for A09.\\

\begin{table}[ht]
\footnotesize
\centering
\begin{tabular}{||ll||rrrrrrrrrr||}
  \hline
  &&  &  &  &  &  $\gamma$ &  &  &  &  & \\
  \hline
 d&method& 1 & 2 & 3 & 4 & 5 & 6 & 7 & 8 & 9 & 10 \\ 
  \hline \hline
10&Grank & 2.71 & 7.02 & 10.19 & 10.85 & 11.01 & 11.05 & 11.10 & 11.13 & 11.15 & 11.17 \\ 
  10&Spearman & 3.15 & 6.36 & 8.07 & 8.43 & 8.54 & 8.56 & 8.59 & 8.61 & 8.62 & 8.63 \\ 
  10&GKnpd & 3.86 & 6.78 & 7.31 & 7.24 & 7.63 & 9.12 & 9.35 & 10.35 & 10.53 & 11.07 \\ 
  10&Cov & 2.48 & 6.83 & 13.89 & 23.65 & 36.13 & 51.34 & 69.30 & 90.00 & 113.44 & 139.64 \\ 
  10&caseMCD & 2.41 & 5.53 & 12.47 & 23.06 & 33.57 & 50.53 & 76.77 & 102.79 & 134.07 & 168.88 \\ 
  10&2SGS & 2.45 & 6.71 & 11.60 & 0.20 & 0.20 & 0.19 & 0.19 & 0.19 & 0.19 & 0.19 \\ 
  10&DI & 2.47 & 4.21 & 0.73 & 0.66 & 0.60 & 0.64 & 0.63 & 0.64 & 0.63 & 0.71 \\ 
  10&cellMCD & 2.14 & 5.97 & 0.94 & 0.29 & 0.30 & 0.29 & 0.29 & 0.29 & 0.29 & 0.29 \\ 
  \hline
  20&Grank & 2.30 & 7.15 & 10.89 & 11.88 & 12.23 & 12.43 & 12.56 & 12.64 & 12.72 & 12.76 \\ 
  20&Spearman & 2.91 & 6.24 & 7.90 & 8.40 & 8.61 & 8.73 & 8.81 & 8.86 & 8.91 & 8.94 \\ 
  20&GKnpd & 3.24 & 5.98 & 5.89 & 3.87 & 7.07 & 10.71 & 13.10 & 14.50 & 14.80 & 15.55 \\ 
  20&Cov & 2.24 & 6.82 & 14.49 & 25.20 & 38.96 & 55.76 & 75.60 & 98.50 & 124.44 & 153.42 \\ 
  20&caseMCD & 2.15 & 6.94 & 14.95 & 28.32 & 48.00 & 68.38 & 99.64 & 130.81 & 152.51 & 186.62 \\ 
  20&2SGS & 2.32 & 6.87 & 8.49 & 0.53 & 0.31 & 0.20 & 0.16 & 0.14 & 0.12 & 0.12 \\ 
  20&DI & 1.79 & 1.54 & 0.29 & 0.28 & 0.25 & 0.27 & 0.27 & 0.26 & 0.26 & 0.27 \\ 
  20&cellMCD & 1.75 & 3.46 & 0.21 & 0.18 & 0.18 & 0.18 & 0.18 & 0.20 & 0.19 & 0.19 \\ 
  \hline
  40&Grank & 3.51 & 11.66 & 17.98 & 20.23 & 21.30 & 21.94 & 22.36 & 22.67 & 22.91 & 23.09 \\ 
  40&Spearman & 4.26 & 9.37 & 12.09 & 13.24 & 13.85 & 14.23 & 14.49 & 14.68 & 14.83 & 14.94 \\ 
  40&GKnpd & 4.69 & 8.59 & 8.29 & 5.31 & 5.00 & 10.67 & 15.81 & 18.57 & 19.58 & 21.56 \\ 
  40&Cov & 3.41 & 11.31 & 24.74 & 43.60 & 67.86 & 97.53 & 132.60 & 173.07 & 218.94 & 270.21 \\ 
  40&caseMCD & 3.24 & 11.07 & 24.44 & 43.20 & 68.25 & 101.38 & 143.31 & 176.56 & 242.36 & 311.15 \\ 
  40&2SGS & 3.54 & 11.22 & 7.07 & 2.49 & 1.60 & 1.13 & 0.89 & 0.73 & 0.59 & 0.50 \\ 
  40&DI & 2.20 & 0.93 & 0.50 & 0.45 & 0.43 & 0.40 & 0.40 & 0.42 & 0.39 & 0.39 \\ 
  40&cellMCD & 2.28 & 1.10 & 0.31 & 0.36 & 0.35 & 0.35 & 0.38 & 0.36 & 0.34 & 0.39 \\ 
\hline
\end{tabular}
\caption{Standard deviations of the discrepancy for $\Sigma = \Sigma_{\mbox{ALYZ}}$}
\label{cellMCD_sds_n_corrTypeALYZ_eps10}
\end{table}

\begin{table}[ht]
\footnotesize
\centering
\begin{tabular}{||ll||rrrrrrrrrr||}
  \hline
  &&  &  &  &  &  $\gamma$ &  &  &  &  & \\
  \hline
 d&method& 1 & 2 & 3 & 4 & 5 & 6 & 7 & 8 & 9 & 10 \\ 
  \hline \hline
10&Grank & 1.22 & 2.60 & 4.10 & 4.93 & 5.62 & 6.20 & 6.64 & 7.00 & 7.17 & 7.27 \\ 
 10& Spearman & 1.20 & 2.36 & 3.34 & 3.93 & 4.37 & 4.75 & 5.00 & 5.19 & 5.28 & 5.34 \\ 
  10&GKnpd & 2.58 & 1.83 & 4.81 & 7.70 & 7.76 & 8.12 & 7.44 & 6.92 & 7.69 & 6.88 \\ 
  10&Cov & 1.19 & 2.34 & 3.99 & 6.13 & 8.79 & 12.00 & 15.76 & 20.10 & 25.00 & 30.48 \\ 
  10&caseMCD & 0.77 & 2.42 & 7.19 & 12.35 & 19.00 & 24.31 & 32.86 & 43.60 & 56.26 & 71.06 \\ 
  10&2SGS & 0.62 & 0.86 & 0.30 & 0.23 & 0.18 & 0.18 & 0.17 & 0.18 & 0.18 & 0.17 \\ 
  10&DI & 0.36 & 0.34 & 0.39 & 0.42 & 0.43 & 0.42 & 0.41 & 0.42 & 0.47 & 0.43 \\ 
  10&cellMCD & 0.30 & 0.30 & 0.31 & 0.32 & 0.31 & 0.30 & 0.37 & 0.35 & 0.39 & 0.40 \\ 
 \hhline{|=|=|=|=|=|=|=|=|=|=|=|=|}
  20&Grank & 0.85 & 1.67 & 2.83 & 3.61 & 4.23 & 4.75 & 5.14 & 5.36 & 5.51 & 5.61 \\ 
  20&Spearman & 0.79 & 1.42 & 2.01 & 2.52 & 2.91 & 3.22 & 3.43 & 3.56 & 3.65 & 3.72 \\ 
  20&GKnpd & 0.72 & 1.17 & 1.57 & 1.72 & 6.78 & 9.10 & 8.58 & 8.39 & 8.94 & 8.13 \\ 
  20&Cov & 0.84 & 1.46 & 2.44 & 3.79 & 5.53 & 7.65 & 10.15 & 13.04 & 16.30 & 19.95 \\ 
  20&caseMCD & 0.48 & 1.46 & 4.00 & 7.01 & 10.57 & 15.09 & 20.43 & 26.51 & 36.33 & 44.09 \\ 
  20&2SGS & 0.50 & 1.38 & 1.03 & 0.51 & 0.20 & 0.13 & 0.11 & 0.11 & 0.10 & 0.10 \\ 
  20&DI & 0.19 & 0.18 & 0.18 & 0.15 & 0.16 & 0.18 & 0.18 & 0.18 & 0.17 & 0.20 \\ 
  20&cellMCD & 0.13 & 0.15 & 0.14 & 0.15 & 0.15 & 0.16 & 0.17 & 0.17 & 0.16 & 0.15 \\ 
  \hhline{|=|=|=|=|=|=|=|=|=|=|=|=|}
  40&Grank & 0.79 & 1.27 & 2.23 & 3.05 & 3.66 & 4.07 & 4.34 & 4.53 & 4.66 & 4.74 \\ 
  40&Spearman & 0.68 & 1.14 & 1.68 & 2.14 & 2.45 & 2.66 & 2.79 & 2.89 & 2.96 & 3.01 \\ 
  40&GKnpd & 0.53 & 0.91 & 1.34 & 1.56 & 1.31 & 8.14 & 9.18 & 10.61 & 11.60 & 10.15 \\ 
  40&Cov & 0.77 & 1.11 & 1.73 & 2.68 & 3.97 & 5.61 & 7.57 & 9.86 & 12.47 & 15.40 \\ 
  40&caseMCD & 0.51 & 0.96 & 2.10 & 5.11 & 7.19 & 10.37 & 13.51 & 17.46 & 21.48 & 25.84 \\ 
  40&2SGS & 0.63 & 1.31 & 2.47 & 1.95 & 0.77 & 0.42 & 0.32 & 0.29 & 0.26 & 0.28 \\ 
  40&DI & 0.37 & 0.28 & 0.29 & 0.24 & 0.25 & 0.24 & 0.25 & 0.27 & 0.28 & 0.26 \\ 
  40&cellMCD & 0.14 & 0.12 & 0.13 & 0.13 & 0.14 & 0.15 & 0.15 & 0.16 & 0.15 & 0.15 \\ 
\hline \hline
\end{tabular}
\caption{Standard deviations of the discrepancy for $\Sigma = \Sigma_{\mbox{A09}}$} 
\label{cellMCD_sds_n_corrTypeA09_eps10}
\end{table}

\clearpage
\normalsize

In these tables we note that the standard deviations
differ a lot by method and value of $\gamma$.
Of course, the same is true for the averages as well.
Figures~\ref{fig:SNR_lowdim} and~\ref{fig:SNR_highdim}
below plot the signal to noise ratio of the discrepancy,
that is, their average divided by their standard deviation.
The roughly horizontal nature of these curves indicates
that the variability and the average typically grow 
together.\\

\begin{figure}[!ht]
\centering
\vskip0.3cm
\begin{minipage}{0.49\linewidth}
  \centering 
    \textbf{ALYZ model, 10\% outliers, $\bm{d = 10}$}
	\includegraphics[width=0.9\textwidth]
	  {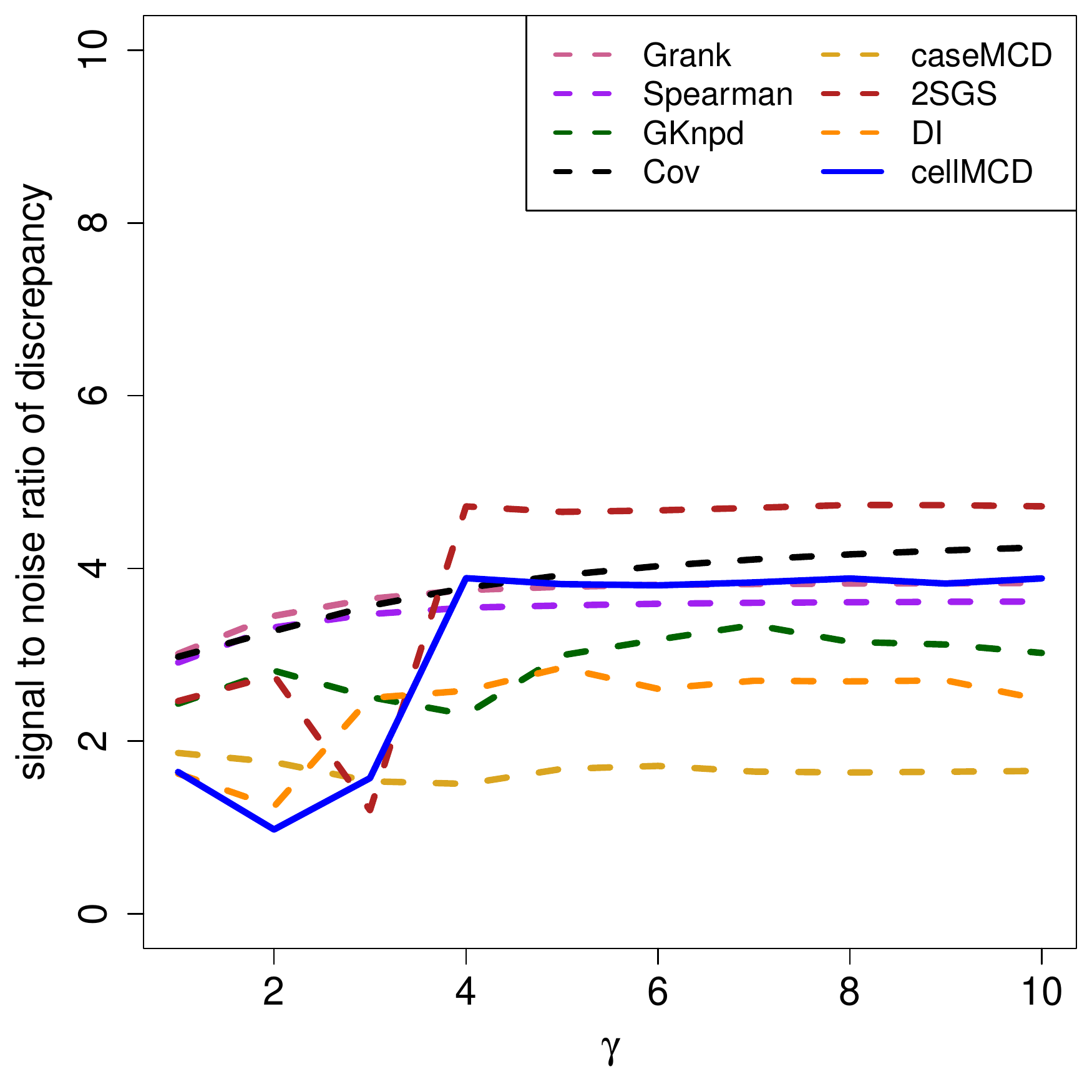} 
\end{minipage}
\begin{minipage}{0.49\linewidth}
  \centering
	  \textbf{A09 model, 10\% outliers, $\bm{d = 10}$}
  \includegraphics[width=0.9\textwidth]
	  {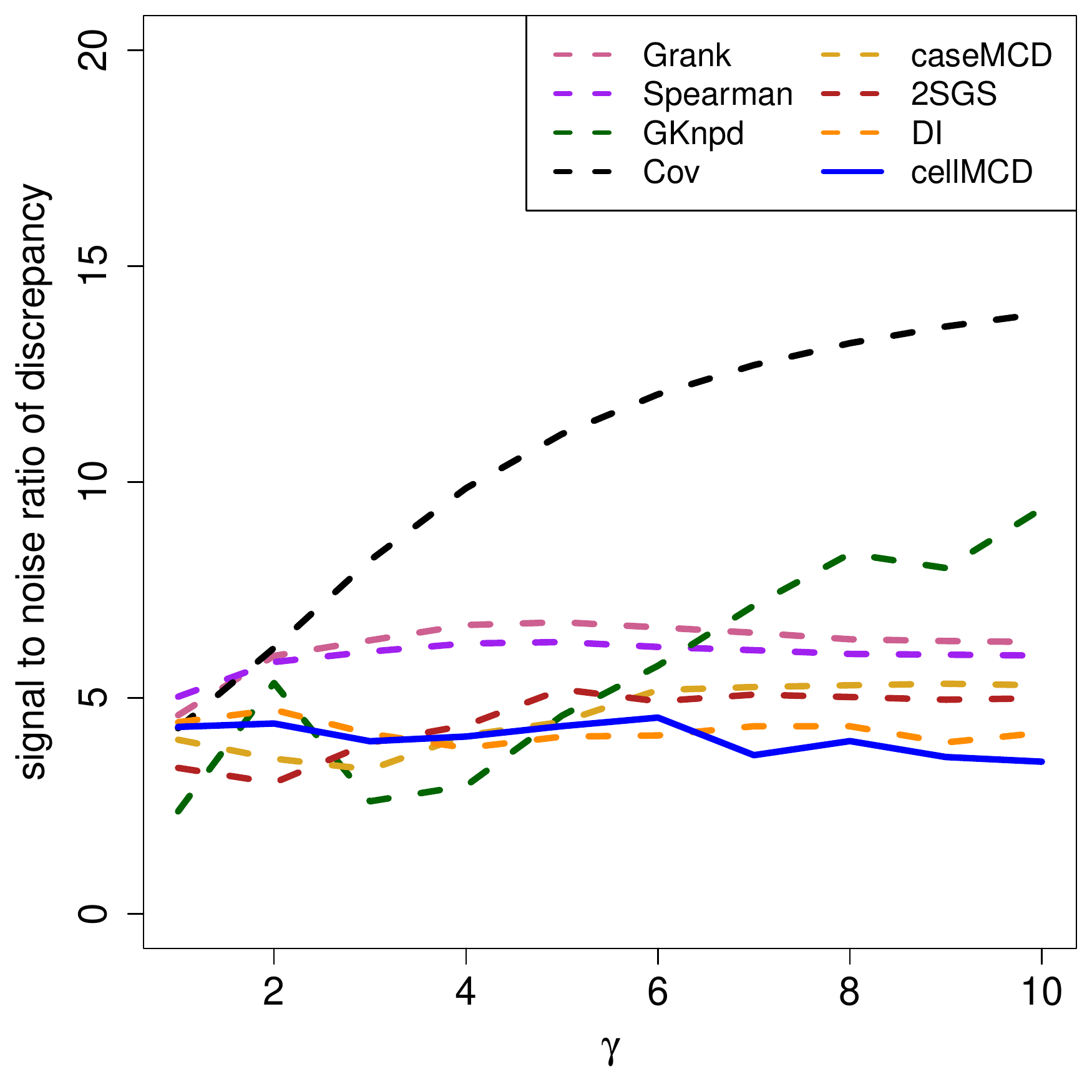} 
\end{minipage}
\vskip-0.2cm
\caption{Signal to noise ratio of the discrepancy of 
 estimated covariance matrices for $d = 10$ and 
 $n = 100$.}
\label{fig:SNR_lowdim}
\end{figure}

\begin{figure}[!ht]
\centering
\vskip0.2cm
\begin{minipage}{0.49\linewidth}
\centering 
    \textbf{ALYZ model, 10\% outliers, $\bm{d = 20}$}
	\includegraphics[width=0.9\textwidth]
	  {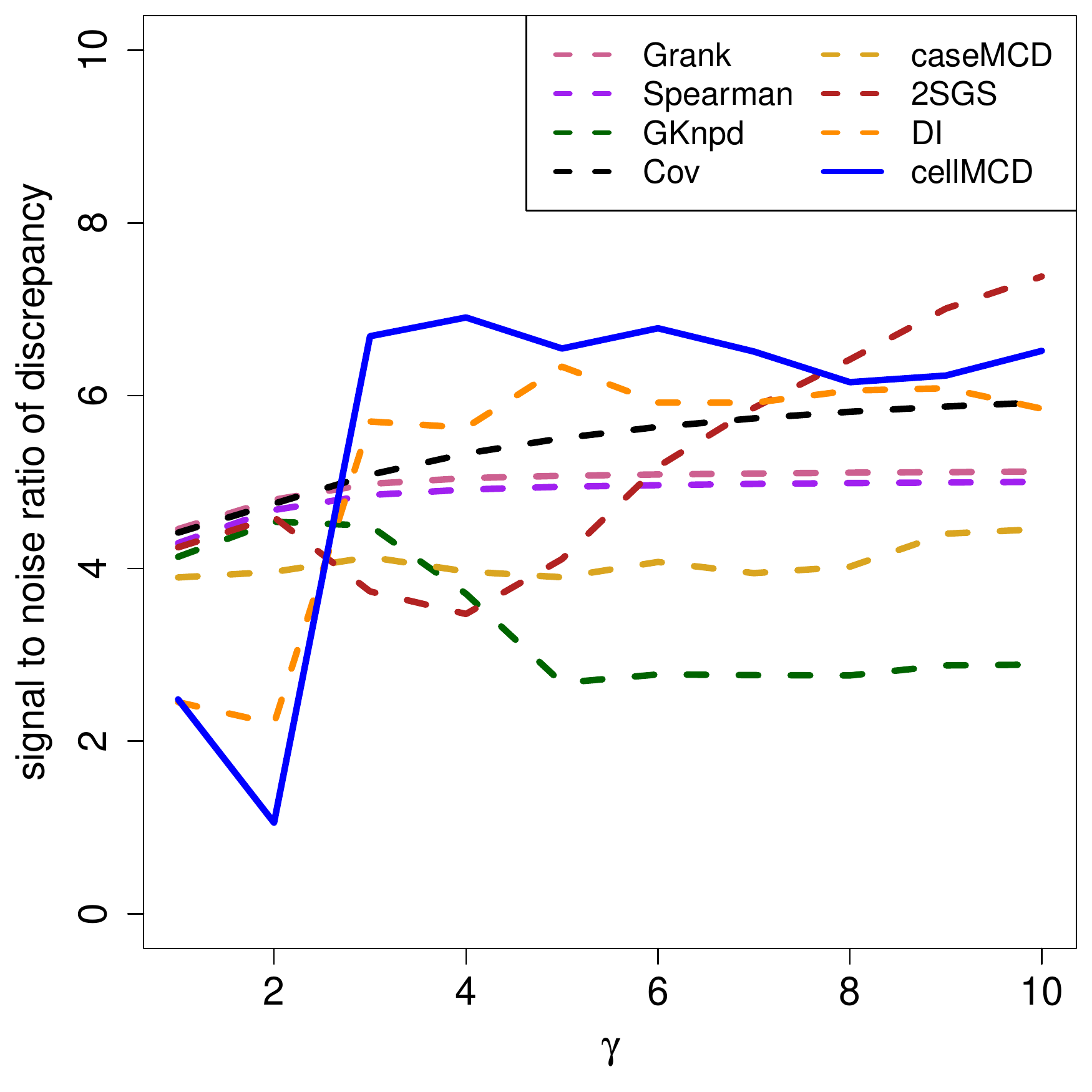} 
\end{minipage}
\begin{minipage}{0.49\linewidth}
  \centering
	  \textbf{A09 model, 10\% outliers, $\bm{d = 20}$}
  \includegraphics[width=0.9\textwidth]
	  {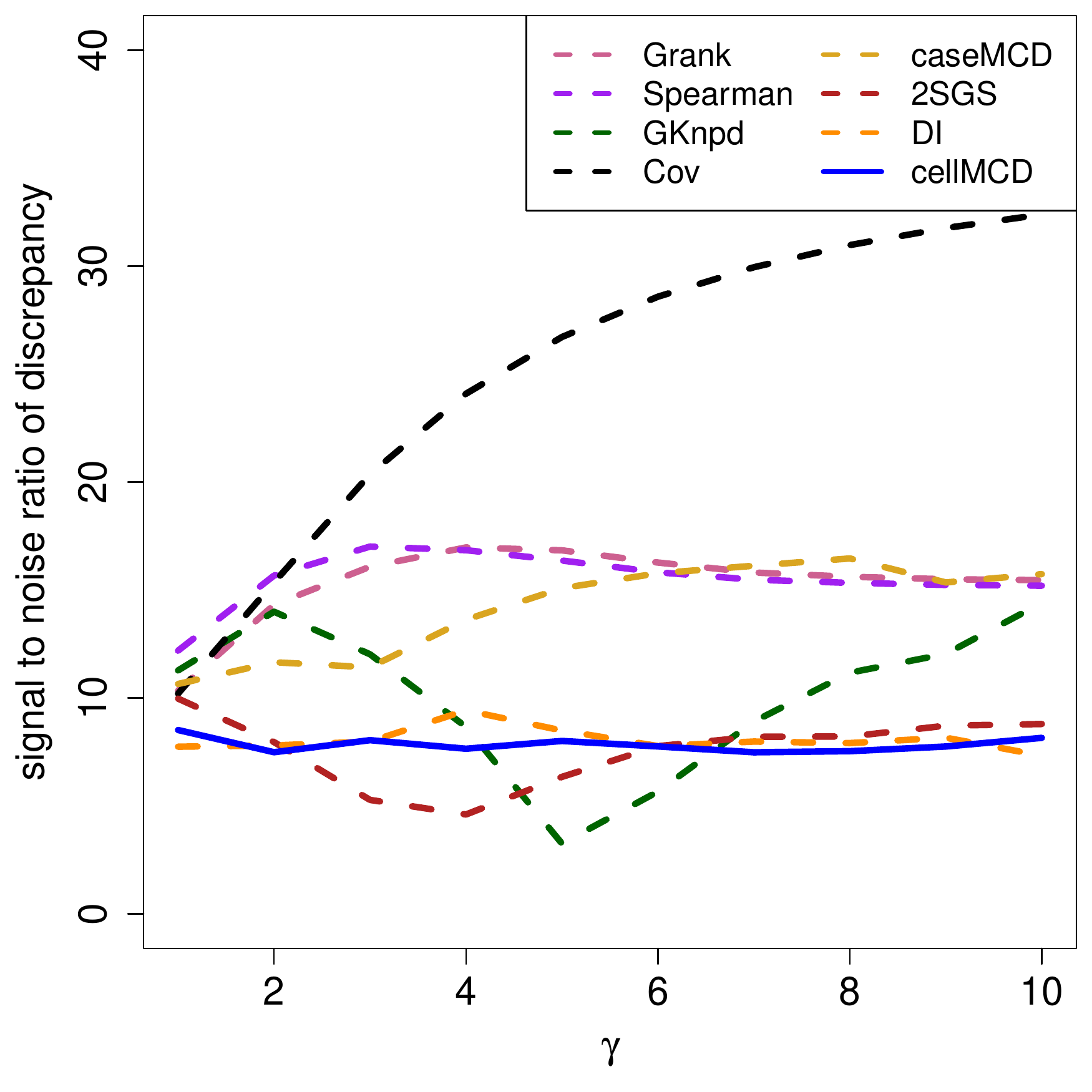} 
\end{minipage}
\vskip0.3cm
\begin{minipage}{0.49\linewidth}
\centering 
    \textbf{ALYZ model, 10\% outliers, $\bm{d = 40}$}
	\includegraphics[width=0.9\textwidth]
	  {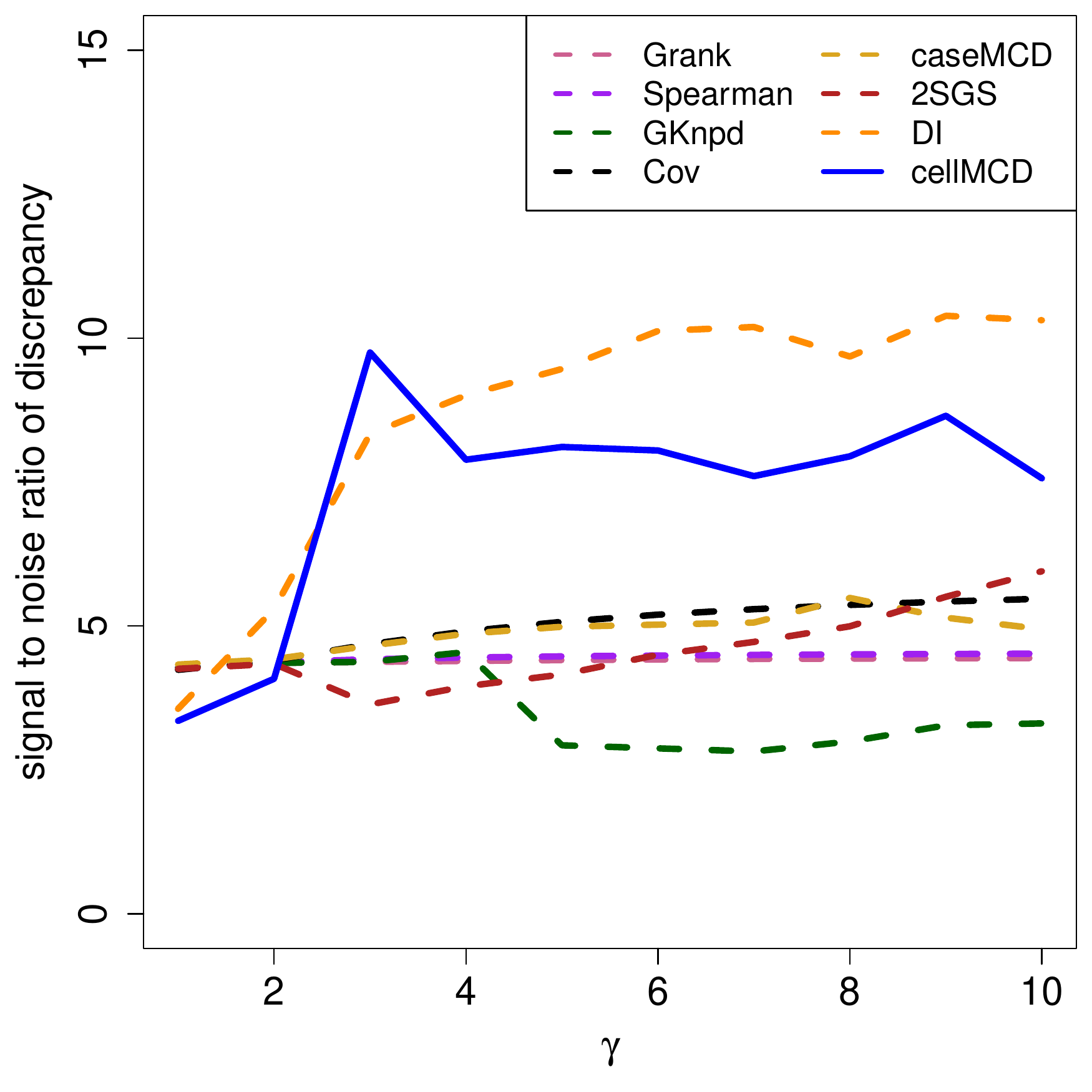} 
\end{minipage}
\begin{minipage}{0.49\linewidth}
  \centering
	  \textbf{A09 model, 10\% outliers, $\bm{d = 40}$}
  \includegraphics[width=0.9\textwidth]
	  {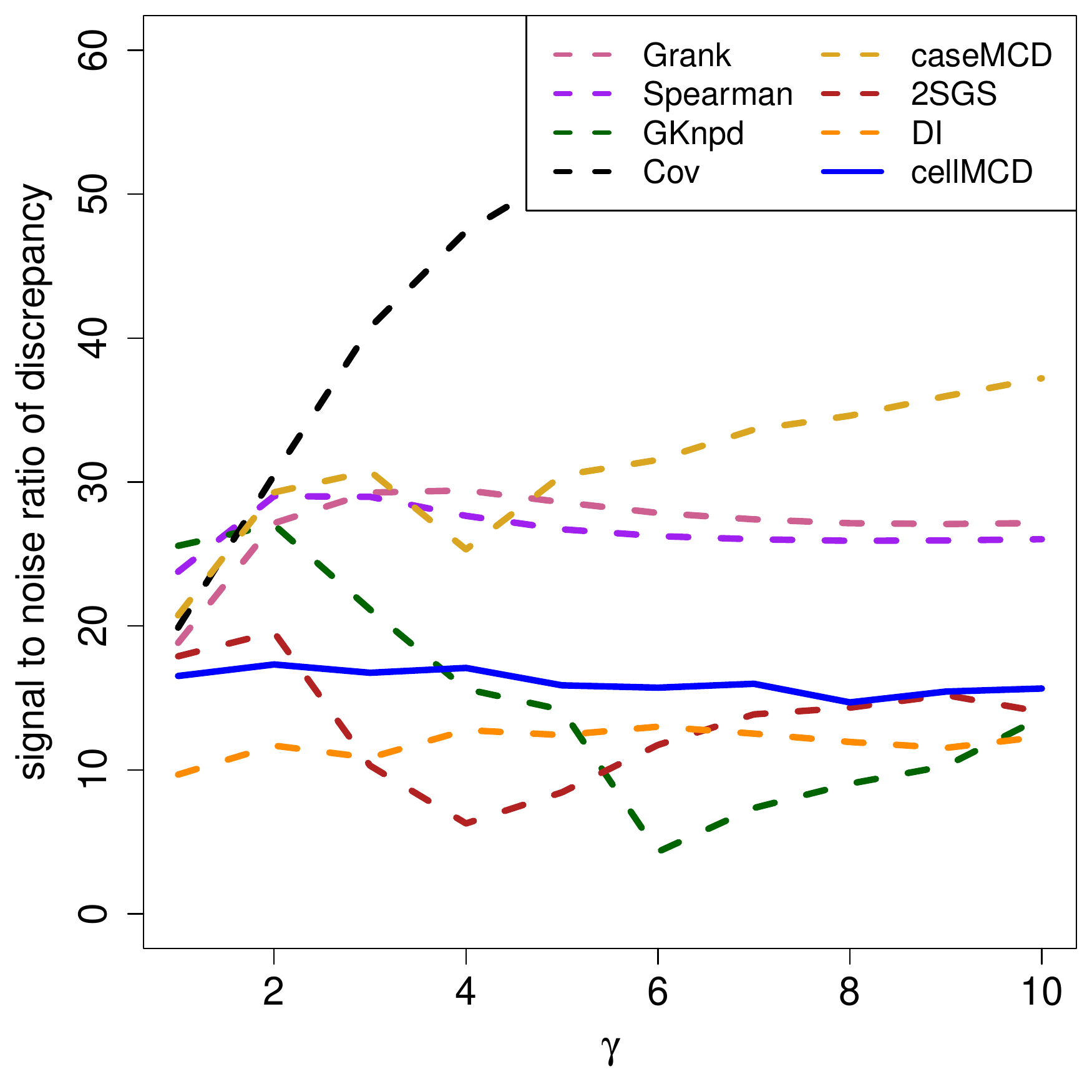} 
\end{minipage}
\vskip-0.2cm
\caption{Signal to noise ratio of the discrepancy 
  of estimated covariance
  matrices for $d = 20$ and $n = 400$ (top panels)
	and for $d = 40$ and $n = 800$ (bottom panels).}
\label{fig:SNR_highdim}
\end{figure}

\clearpage
%%%%%%%%%%%%%%%%%%%%%%%%%%%%%%%%%%%%%%%%%%%%%%%%%%%%%%%%%%%%%%%%%%%
\subsection*{A.6.2\;\;Simulation on the ordering of the variables}
\label{A:ordering}

In section~\ref{sec:algo}, part (a) of the C-step 
updates the matrix $\bW$ in \eqref{eq:cellMCD} 
while keeping $\bhmu^{(k)}$ and $\bhSigma^{(k)}$ 
unchanged. We start the new pattern $\BtW$ 
as $\BtW = \bW^{(k)}$, and then we
modify $\BtW$ column by column, by
cycling over the variables $j=1,\ldots,d$.
A referee asked to check the effect of the order 
in which the variables are updated on the result
of the algorithm, in the simulations with A09 
in section~\ref{sec:simul}.

In order to study this we ran an experiment in which
different orderings of the variables (columns) were 
tried, while keeping the remainder of the algorithm 
unchanged. We considered 5 options:
\begin{itemize}
\item \textbf{original}: the variables are updated 
  in the original ordering of the data (1 to $d$).
\item \textbf{T descending}: the variables are 
  updated in the order of descending tail weight, as 
  measured by $\sum_{i=1}^{n} |X_{ij}|$.
\item \textbf{T ascending}: the variables are 
  updated in the order of ascending tail weight, 
  measured in the same way.
\item \textbf{W descending}: the variables are 
  updated in the order of descending number of 
  unflagged cells, as measured 
  by $\sum_{i=1}^{n} W_{ij}$\,.
\item \textbf{W ascending}: the variables are 
  updated in the order of ascending number of 
  unflagged cells, measured in the same way.
\end{itemize}
For each version we carried out 100 replications,
and computed the averaged MSE of the estimated center, 
the KLdiv of the estimated covariance matrix, 
and the value of the objective function. This
yielded the tables below.
Based on these results, the choice of the ordering 
appears to have only a tiny effect.

\begin{table}[ht]
\centering
\begin{tabular}{cccc}
  \hline
 & MSE($\hmu$) & KLdiv($\hSigma$) & objective \\ 
  \hline
  original & 0.01000 & 1.297 & 1285.94 \\ 
  T descending & 0.00985 & 1.289 & 1286.28 \\ 
  T ascending & 0.00993 & 1.291 & 1286.93 \\ 
  W descending & 0.01002 & 1.298 & 1286.56 \\ 
  W ascending & 0.00990 & 1.281 & 1286.19 \\ 
   \hline
\end{tabular}
\caption{$d=10$, $n=100$, $\Sigma = \Sigma_{\mbox{A09}}$ 
  with $\varepsilon = 0$.} 
\vspace{1cm}  
\label{tab_variableordering_n100_p10_A09_eps0_gamma0}
\end{table}

\begin{table}[ht]
\centering
\begin{tabular}{cccc}
  \hline
 & MSE($\hmu$) & KLdiv($\hSigma$) & objective \\ 
  \hline
  original & 0.01046 & 1.301 & 1614.71 \\ 
  T descending & 0.01053 & 1.300 & 1614.88 \\ 
  T ascending & 0.01052 & 1.301 & 1614.79 \\ 
  W descending & 0.01063 & 1.301 & 1616.70 \\ 
  W ascending & 0.01039 & 1.292 & 1615.02 \\ 
   \hline
\end{tabular}
\caption{$d=10$, $n=100$, $\Sigma = \Sigma_{\mbox{A09}}$ 
  with $\varepsilon = 0.1$ and $\gamma = 4$.} 
\vspace{1cm}    
\label{tab_variableordering_n100_p10_A09_eps0.1_gamma4}
\end{table}

\begin{table}[ht]
\centering
\begin{tabular}{cccc}
  \hline
 & MSE($\hmu$) & KLdiv($\hSigma$) & objective\\ 
  \hline
  original & 0.01185 & 2.749 & 2034.36 \\ 
  T descending & 0.01179 & 2.707 & 2033.03 \\ 
  T ascending & 0.01182 & 2.740 & 2032.88 \\ 
  W descending & 0.01201 & 2.838 & 2036.42 \\ 
  W ascending & 0.01179 & 2.749 & 2034.73 \\ 
   \hline
\end{tabular}
\caption{$d=10$, $n=100$, $\Sigma = \Sigma_{\mbox{A09}}$ 
  with $\varepsilon = 0.2$ and $\gamma = 4$.}
\vspace{1cm}  
\label{tab_variableordering_n100_p10_A09_eps0.2_gamma4}
\end{table}

\clearpage
%%%%%%%%%%%%%%%%%%%%%%%%%%%%%%%%%%%%%%%%%%%%%%%%%%%%%%%%%%%%%%%%%%%%%%%%%
\subsection*{A.6.3\;\; Simulations on the number of W-steps and EM-steps}
\label{A:WandEM}

The concentration step (C-step) of the cellMCD
algorithm in section~\ref{sec:algo} consists
of two parts. Part (a) updates the matrix 
$\bW$ in \eqref{eq:cellMCD} while keeping 
$\bhmu$ and $\bhSigma$ as they are. Let us
call this a W-step. Part (b) uses the new
`missingness' pattern $\bW$ and carries out
an EM-step to update $\bhmu$ and $\bhSigma$.
So each C-step contains a single W-step and 
a single EM-step.

A referee asked what would be the effect of
increasing the number of W-steps and/or
EM-steps. We studied this by considering 5 
settings. The original setting is denoted
as 1W+1EM, so a single W-step and a single
EM-step. The other four settings are, with
similar notation, 5W+1EM, 1W+5EM, 5W+5EM,
and 10W+10EM. The convergence criterion
and everything else in the algorithm was
left unchanged.
We ran 100 replications of the algorithm
versions, for the following combinations
of choices. The dimension $d$ is either
10 (with $n=100$), 20 (with $n=400$), or
40 (with $n=800$), and the true $\Sigma$
is either A09 or ALYZ. The contamination 
fraction $\varepsilon$ is 0, 0.1, or 0.2\,.
And finally, the position of the cellwise 
outliers is given by $\gamma$ equal to
4 or 10.

As expected, increasing the number of
W-steps and/or EM-steps increases the
overall computation time. This is seen
in Table~\ref{tab:algocomptimes} which
provides the averaged computation time
over all settings in each of the three
dimensions. Note that W-steps are more
expensive than EM-steps, due to frequent
evaluations of the objective function.
The table shows that increasing the
number of steps is costly, especially
in the higher dimensions.

\begin{table}[ht]
%\small
\centering
\caption{Average computation times (in seconds) 
of algorithm versions.} 
\label{tab:algocomptimes}
\vspace{0.3cm}
\begin{tabular}{llll}
  \hline
 version & d=10 & d=20 & d=40 \\ 
  \hline
  1W + 1EM & 0.42 & 2.72 & 26.24 \\ 
  5W + 1EM & 0.83 & 10.02 & 118.78 \\ 
  1W + 5EM & 0.61 & 4.31 & 32.32 \\ 
  5W + 5EM & 1.13 & 11.74 & 125.77 \\ 
  10W + 10EM & 2.20 & 23.16 & 250.44 \\ 
   \hline
\end{tabular}
\end{table}

The main 
question is of course whether
the faster 1W+1EM version pays a price in
estimation accuracy. The next three %large
tables say that it does not, as 
% in all settings 
the differences in MSE($\hmu$), 
KLdiv($\hSigma$) and the objective function
were tiny. Also, the effect is not
systematic: more computation time does not
necessarily yield a lower MSE($\hmu$), 
KLdiv($\hSigma$), or objective. 
 
\begin{table}[ht]
\centering
\begin{tabular}{||ccc||ccc||ccc||}
\hline 
 & & & & A09 & & & ALYZ & \\
\hline 
$\varepsilon$ & $\gamma$ & method & MSE($\hmu$) & KLdiv($\hSigma$) & 
  objective & MSE($\hmu$) & KLdiv($\hSigma$) & objective\\ 
  \hline \hline
  0 & -- & 1W + 1EM & 0.01001 & 1.228 & 1289.27 & 0.00998 & 0.846 & 2286.53 \\ 
  0 & -- & 5W + 1EM & 0.00998 & 1.243 & 1289.66 & 0.01000 & 0.844 & 2286.49 \\ 
  0 & -- & 1W + 5EM & 0.00998 & 1.233 & 1289.19 & 0.01000 & 0.846 & 2286.53 \\ 
  0 & -- & 5W + 5EM & 0.00998 & 1.244 & 1289.55 & 0.01002 & 0.845 & 2286.41 \\ 
  0 & -- & 10W + 10EM & 0.00998 & 1.244 & 1289.55 & 0.01002 & 0.845 & 2286.41 \\ \hhline{|=|=|=|=|=|=|=|=|=|}
  0.1 & 4 & 1W + 1EM & 0.01050 & 1.323 & 1613.00 & 0.01139 & 1.141 & 2789.84 \\ 
  0.1 & 4 & 5W + 1EM & 0.01046 & 1.303 & 1612.90 & 0.01132 & 1.138 & 2789.59 \\ 
  0.1 & 4 & 1W + 5EM & 0.01050 & 1.300 & 1613.27 & 0.01137 & 1.136 & 2790.15 \\ 
  0.1 & 4 & 5W + 5EM & 0.01040 & 1.293 & 1612.79 & 0.01134 & 1.138 & 2789.58 \\ 
  0.1 & 4 & 10W + 10EM & 0.0104 & 1.293 & 1612.79 & 0.01134 & 1.138 & 2789.58 \\  \hline
  0.1 & 10 & 1W + 1EM & 0.01043 & 1.418 & 1714.75 & 0.01094 & 1.118 & 2828.22 \\ 
  0.1 & 10 & 5W + 1EM & 0.01043 & 1.416 & 1714.92 & 0.01092 & 1.121 & 2828.31 \\ 
  0.1 & 10 & 1W + 5EM & 0.01038 & 1.421 & 1714.67 & 0.01096 & 1.121 & 2828.22 \\ 
  0.1 & 10 & 5W + 5EM & 0.01041 & 1.418 & 1714.89 & 0.01092 & 1.121 & 2828.32 \\ 
  0.1 & 10 & 10W + 10EM & 0.01041 & 1.418 & 1714.89 & 0.01092 & 1.121 & 2828.32 \\ \hhline{|=|=|=|=|=|=|=|=|=|}
  0.2 & 4 & 1W + 1EM & 0.01180 & 2.710 & 2034.65 & 0.01568 & 3.473 & 3119.95 \\ 
  0.2 & 4 & 5W + 1EM & 0.01181 & 2.748 & 2034.04 & 0.01559 & 3.359 & 3119.06 \\ 
  0.2 & 4 & 1W + 5EM & 0.01179 & 2.745 & 2036.07 & 0.01549 & 3.455 & 3118.36 \\ 
  0.2 & 4 & 5W + 5EM & 0.01184 & 2.768 & 2035.55 & 0.01571 & 3.369 & 3119.30 \\ 
  0.2 & 4 & 10W + 10EM & 0.01184 & 2.768 & 2035.55 & 0.01571 & 3.369 & 3119.30 \\  \hline
  0.2 & 10 & 1W + 1EM & 0.01109 & 1.795 & 2109.86 & 0.01223 & 2.009 & 3324.95 \\ 
  0.2 & 10 & 5W + 1EM & 0.01099 & 1.794 & 2109.43 & 0.01224 & 2.015 & 3324.98 \\ 
  0.2 & 10 & 1W + 5EM & 0.01104 & 1.774 & 2110.05 & 0.01222 & 2.025 & 3325.00 \\ 
  0.2 & 10 & 5W + 5EM & 0.01096 & 1.780 & 2109.32 & 0.01226 & 2.025 & 3325.05 \\ 
  0.2 & 10 & 10W + 10EM & 0.01096 & 1.779 & 2109.33 & 0.01229 & 2.021 & 3325.09 \\
\hline \hline
\end{tabular}
\caption{Comparison of different algorithm versions 
on data with d=10 and n=100.} 
\label{tab:algoversions_d10}
\end{table}

\begin{table}[ht]
\centering
\begin{tabular}{||ccc||ccc||ccc||}
\hline 
 & & & & A09 & & & ALYZ & \\
\hline 
$\varepsilon$ & $\gamma$ & method & MSE($\hmu$) & KLdiv($\hSigma$) & 
  objective & MSE($\hmu$) & KLdiv($\hSigma$) & objective\\
  \hline \hline
  0 & -- & 1W + 1EM & 0.00242 & 1.151 & 9694.82 & 0.00247 & 0.830 & 19346.52 \\ 
  0 & -- & 5W + 1EM & 0.00242 & 1.154 & 9694.30 & 0.00246 & 0.831 & 19346.57 \\ 
  0 & -- & 1W + 5EM & 0.00243 & 1.153 & 9694.68 & 0.00247 & 0.830 & 19346.57 \\ 
  0 & -- & 5W + 5EM & 0.00243 & 1.155 & 9694.24 & 0.00246 & 0.829 & 19346.72 \\ 
  0 & -- & 10W + 10EM & 0.00243 & 1.155 & 9694.24 & 0.00246 & 0.829 & 19346.72 \\  \hhline{|=|=|=|=|=|=|=|=|=|}
  0.1 & 4 & 1W + 1EM & 0.00251 & 1.185 & 12451.15 & 0.00283 & 1.216 & 23240.19 \\ 
  0.1 & 4 & 5W + 1EM & 0.00252 & 1.186 & 12446.70 & 0.00283 & 1.211 & 23240.16 \\ 
  0.1 & 4 & 1W + 5EM & 0.00251 & 1.189 & 12450.67 & 0.00284 & 1.213 & 23240.50 \\ 
  0.1 & 4 & 5W + 5EM & 0.00251 & 1.187 & 12446.52 & 0.00283 & 1.215 & 23240.09 \\ 
  0.1 & 4 & 10W + 10EM & 0.00251 & 1.187 & 12446.51 & 0.00283 & 1.215 & 23240.09 \\ \hline
  0.1 & 10 & 1W + 1EM & 0.00248 & 1.256 & 13289.50 & 0.00267 & 1.214 & 23693.20 \\ 
  0.1 & 10 & 5W + 1EM & 0.00248 & 1.257 & 13288.34 & 0.00266 & 1.215 & 23693.08 \\ 
  0.1 & 10 & 1W + 5EM & 0.00248 & 1.259 & 13289.29 & 0.00267 & 1.214 & 23693.38 \\ 
  0.1 & 10 & 5W + 5EM & 0.00248 & 1.255 & 13288.59 & 0.00267 & 1.220 & 23692.84 \\ 
  0.1 & 10 & 10W + 10EM & 0.00248 & 1.255 & 13288.59 & 0.00267 & 1.219 & 23692.84 \\  \hhline{|=|=|=|=|=|=|=|=|=|}
  0.2 & 4 & 1W + 1EM & 0.00270 & 2.086 & 15592.22 & 0.00386 & 3.346 & 25257.25 \\ 
  0.2 & 4 & 5W + 1EM & 0.00270 & 2.102 & 15590.99 & 0.00385 & 3.375 & 25258.53 \\ 
  0.2 & 4 & 1W + 5EM & 0.00271 & 2.081 & 15592.68 & 0.00383 & 3.338 & 25258.51 \\ 
  0.2 & 4 & 5W + 5EM & 0.00271 & 2.092 & 15590.41 & 0.00386 & 3.354 & 25257.58 \\ 
  0.2 & 4 & 10W + 10EM & 0.00270 & 2.091 & 15590.34 & 0.00386 & 3.345 & 25257.52 \\ \hline
  0.2 & 10 & 1W + 1EM & 0.00260 & 1.593 & 16725.02 & 0.00321 & 1.952 & 27383.61 \\ 
  0.2 & 10 & 5W + 1EM & 0.00259 & 1.588 & 16722.75 & 0.00321 & 1.957 & 27383.51 \\ 
  0.2 & 10 & 1W + 5EM & 0.00259 & 1.589 & 16724.92 & 0.00321 & 1.953 & 27383.79 \\ 
  0.2 & 10 & 5W + 5EM & 0.00259 & 1.586 & 16723.12 & 0.00322 & 1.970 & 27383.19 \\ 
  0.2 & 10 & 10W + 10EM & 0.00259 & 1.586 & 16723.17 & 0.00322 & 1.969 & 27383.08 \\
   \hline
\end{tabular}
\caption{Comparison of different algorithm versions 
on data with d=20 and n = 400.} 
\label{tab:algoversions_d20}
\end{table}

\begin{table}[ht]
\centering
\begin{tabular}{||ccc||ccc||ccc||}
\hline 
 & & & & A09 & & & ALYZ & \\
\hline 
$\varepsilon$ & $\gamma$ & method & MSE($\hmu$) & KLdiv($\hSigma$) & 
  objective & MSE($\hmu$) & KLdiv($\hSigma$) & objective\\
  \hline \hline
  0 & -- & 1W + 1EM & 0.00122 & 2.184 & 37556.92 & 0.00126 & 1.937 & 79272.01 \\ 
  0 & -- & 5W + 1EM & 0.00122 & 2.188 & 37556.93 & 0.00126 & 1.931 & 79272.70 \\ 
  0 & -- & 1W + 5EM & 0.00122 & 2.186 & 37556.64 & 0.00126 & 1.939 & 79271.40 \\ 
  0 & -- & 5W + 5EM & 0.00122 & 2.188 & 37557.10 & 0.00126 & 1.931 & 79272.81 \\ 
  0 & -- & 10W + 10EM & 0.00122 & 2.188 & 37557.10 & 0.00126 & 1.931 & 79272.81 \\ \hhline{|=|=|=|=|=|=|=|=|=|}
  0.1 & 4 & 1W + 1EM & 0.00126 & 2.279 & 49674.77 & 0.00153 & 2.808 & 92440.95 \\ 
  0.1 & 4 & 5W + 1EM & 0.00126 & 2.273 & 49659.15 & 0.00153 & 2.801 & 92441.82 \\ 
  0.1 & 4 & 1W + 5EM & 0.00126 & 2.281 & 49673.60 & 0.00153 & 2.816 & 92441.39 \\ 
  0.1 & 4 & 5W + 5EM & 0.00126 & 2.275 & 49658.97 & 0.00153 & 2.809 & 92439.42 \\ 
  0.1 & 4 & 10W + 10EM & 0.00126 & 2.275 & 49658.97 & 0.00153 & 2.804 & 92440.26 \\ \hline
  0.1 & 10 & 1W + 1EM & 0.00125 & 2.325 & 51862.73 & 0.00142 & 2.943 & 95432.80 \\ 
  0.1 & 10 & 5W + 1EM & 0.00126 & 2.321 & 51857.96 & 0.00142 & 2.922 & 95434.68 \\ 
  0.1 & 10 & 1W + 5EM & 0.00125 & 2.329 & 51862.62 & 0.00141 & 2.933 & 95433.13 \\ 
  0.1 & 10 & 5W + 5EM & 0.00125 & 2.320 & 51858.69 & 0.00142 & 2.934 & 95433.04 \\ 
  0.1 & 10 & 10W + 10EM & 0.00125 & 2.320 & 51858.69 & 0.00142 & 2.936 & 95433.03 \\ \hhline{|=|=|=|=|=|=|=|=|=|}
  0.2 & 4 & 1W + 1EM & 0.00129 & 4.083 & 65008.14 & 0.00201 & 9.206 & 99283.76 \\ 
  0.2 & 4 & 5W + 1EM & 0.00130 & 4.121 & 65001.05 & 0.00201 & 9.311 & 99286.22 \\ 
  0.2 & 4 & 1W + 5EM & 0.00129 & 4.087 & 65007.01 & 0.00201 & 9.261 & 99286.68 \\ 
  0.2 & 4 & 5W + 5EM & 0.00130 & 4.119 & 64999.17 & 0.00202 & 9.297 & 99285.07 \\ 
  0.2 & 4 & 10W + 10EM & 0.00130 & 4.120 & 64999.05 & 0.00202 & 9.297 & 99285.13 \\ \hline
  0.2 & 10 & 1W + 1EM & 0.00127 & 3.515 & 67014.48 & 0.00170 & 4.313 & 108819.48 \\ 
  0.2 & 10 & 5W + 1EM & 0.00127 & 3.523 & 67009.90 & 0.00169 & 4.330 & 108819.88 \\ 
  0.2 & 10 & 1W + 5EM & 0.00126 & 3.521 & 67016.49 & 0.00170 & 4.338 & 108819.34 \\ 
  0.2 & 10 & 5W + 5EM & 0.00127 & 3.517 & 67009.13 & 0.00169 & 4.333 & 108819.46 \\ 
  0.2 & 10 & 10W + 10EM & 0.00127 & 3.517 & 67009.25 & 0.00169 & 4.321 & 108819.65 \\
   \hline
\end{tabular}
\caption{Comparison of different algorithm versions on data 
with d=40 and n = 800.} 
\label{tab:algoversions_d40}
\end{table}

\clearpage
%%%%%%%%%%%%%%%%%%%%%%%%%%%%%%%%%%%%%%%%%%%%%%%%%%%%%%%%%%%%%%%%%
\subsection*{A.6.4\;\;Results for other variations on the method}

Figures~\ref{fig:KLdiv_lowdim} and \ref{fig:KLdiv_highdim} 
in the paper show the Kullback-Leibler discrepancy of 
several covariance estimators in dimensions 10, 20, and 40.
Referees requested two more estimators to be considered:
\begin{itemize}
\item the initial estimator DDCW, which is fast as
    seen from its entries in Table~\ref{tab:comptimes} 
    in the paper;
\item cellMCD without the penalty term, that is, 
    setting all $q_j = 0$. We denote this as cellMCD\_q0\,.
\end{itemize}
Figure~\ref{fig:A_KLdiv_0p1} below plots both versions, 
together with the curve for cellMCD shown in 
Figures~\ref{fig:KLdiv_lowdim} and \ref{fig:KLdiv_highdim} 
in the paper.
We see that DDCW and cellMCD\_q0 do not explode in the sense 
that the effect of contaminated cells remains bounded,
which is what we want for our initial estimator DDCW. 
However, cellMCD\_q0 has a large discrepancy, because 
always taking out 25\% of the cells in each variable hurts
efficiency. DDCW does reasonably well by itself under A09,
but very poorly under ALYZ. Neither version can thus be 
considered a competitive alternative to cellMCD.

\begin{figure}[!ht]
\centering
\vskip0.3cm
\includegraphics[width=0.44\textwidth]{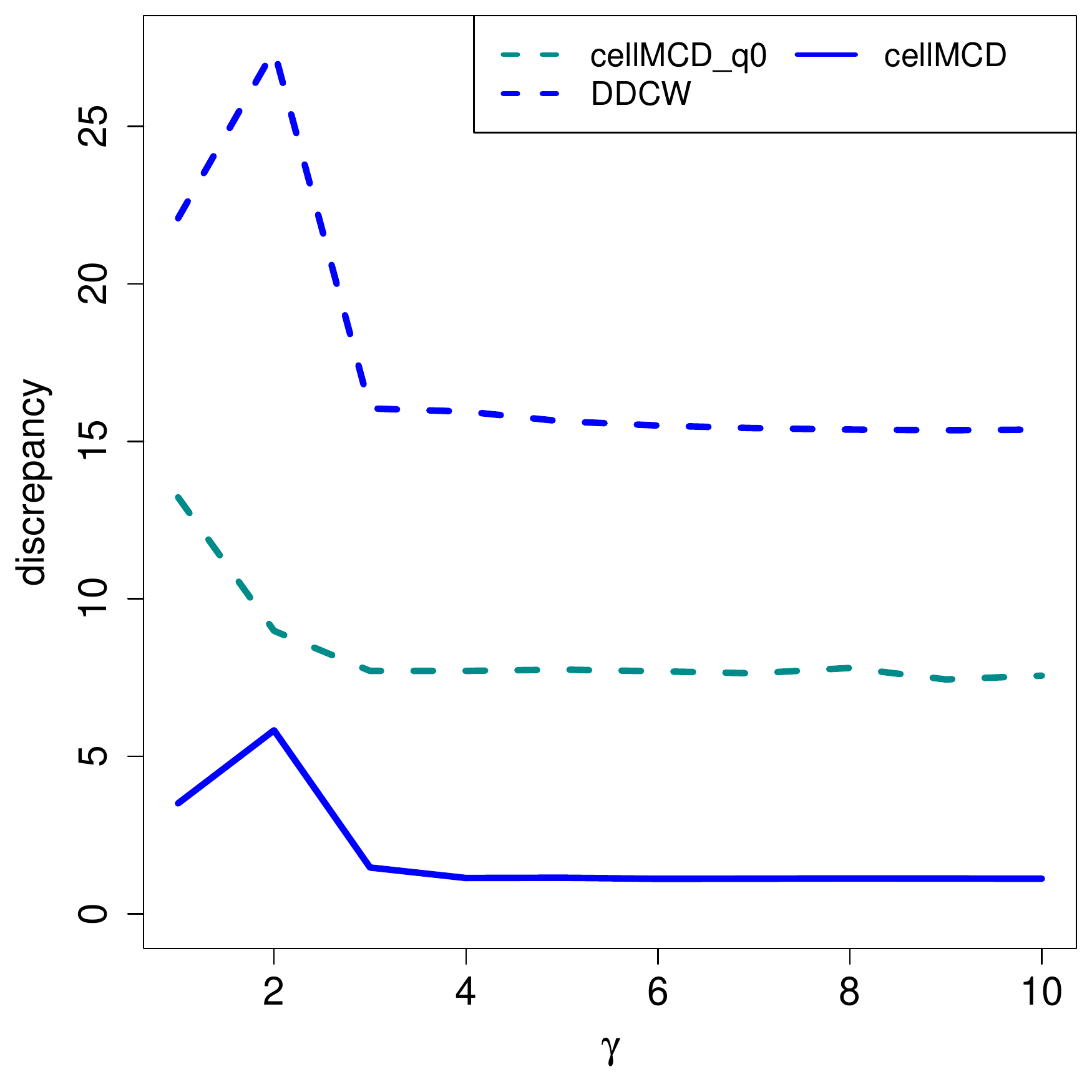}
\includegraphics[width=0.44\textwidth]{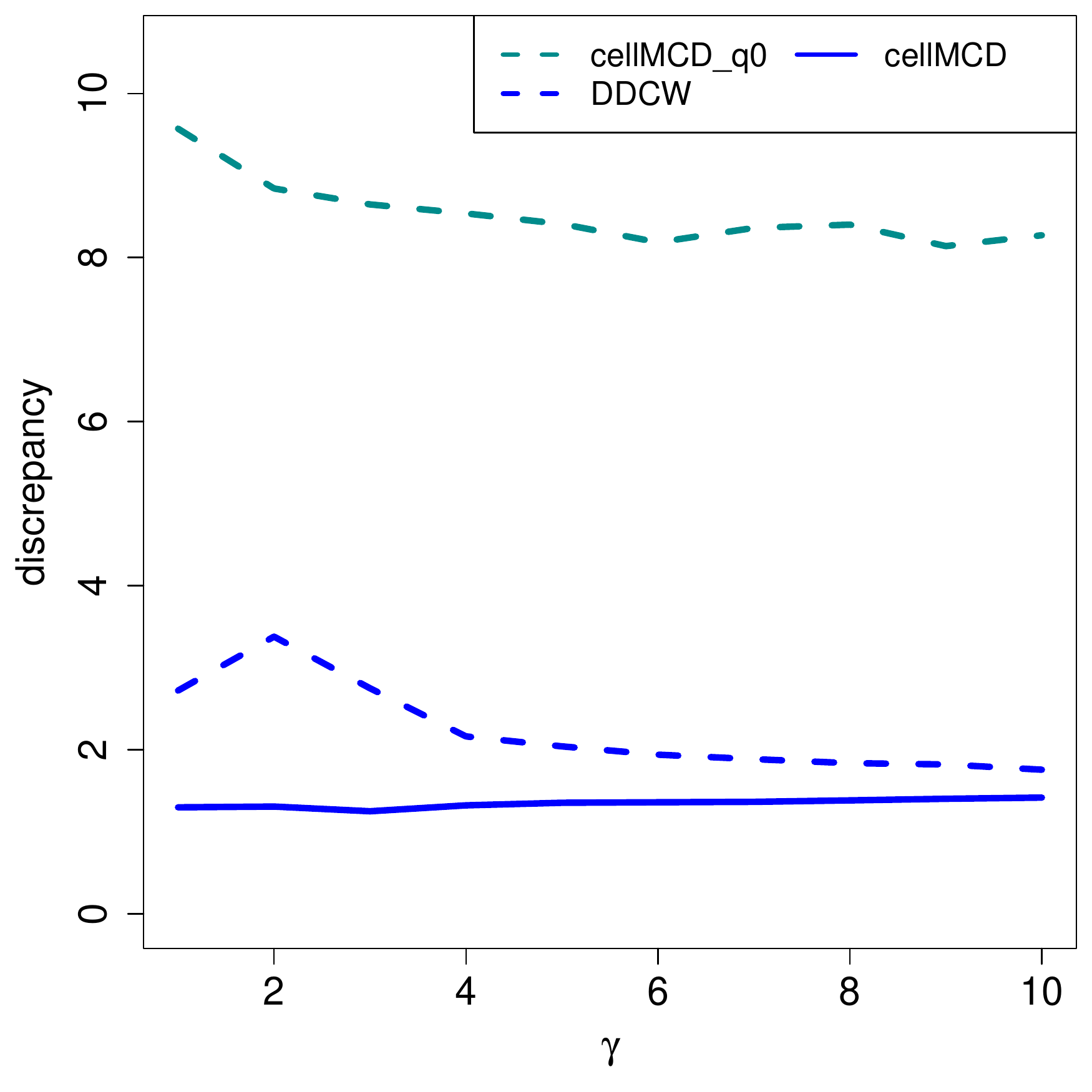}
\includegraphics[width=0.44\textwidth]{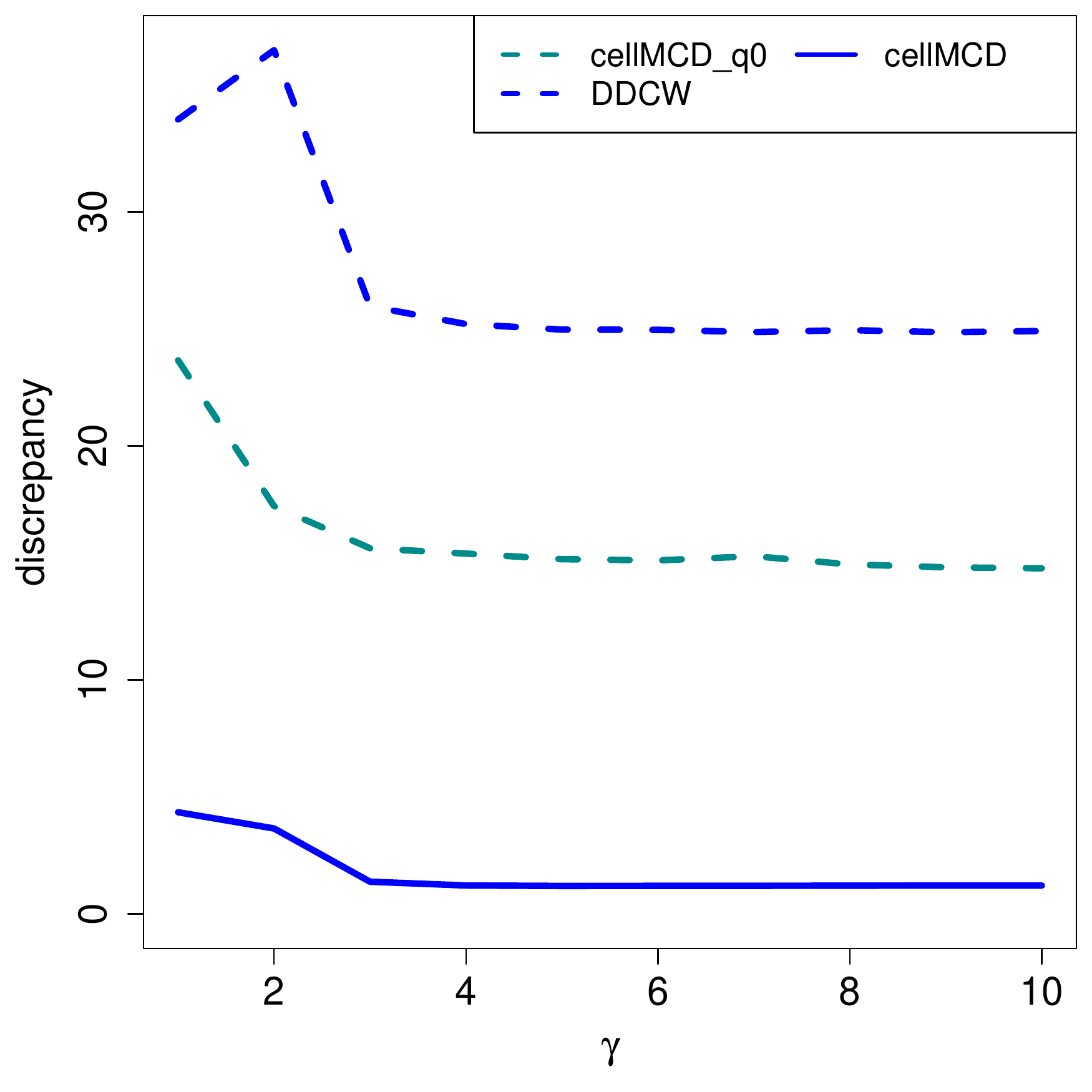}
\includegraphics[width=0.44\textwidth]{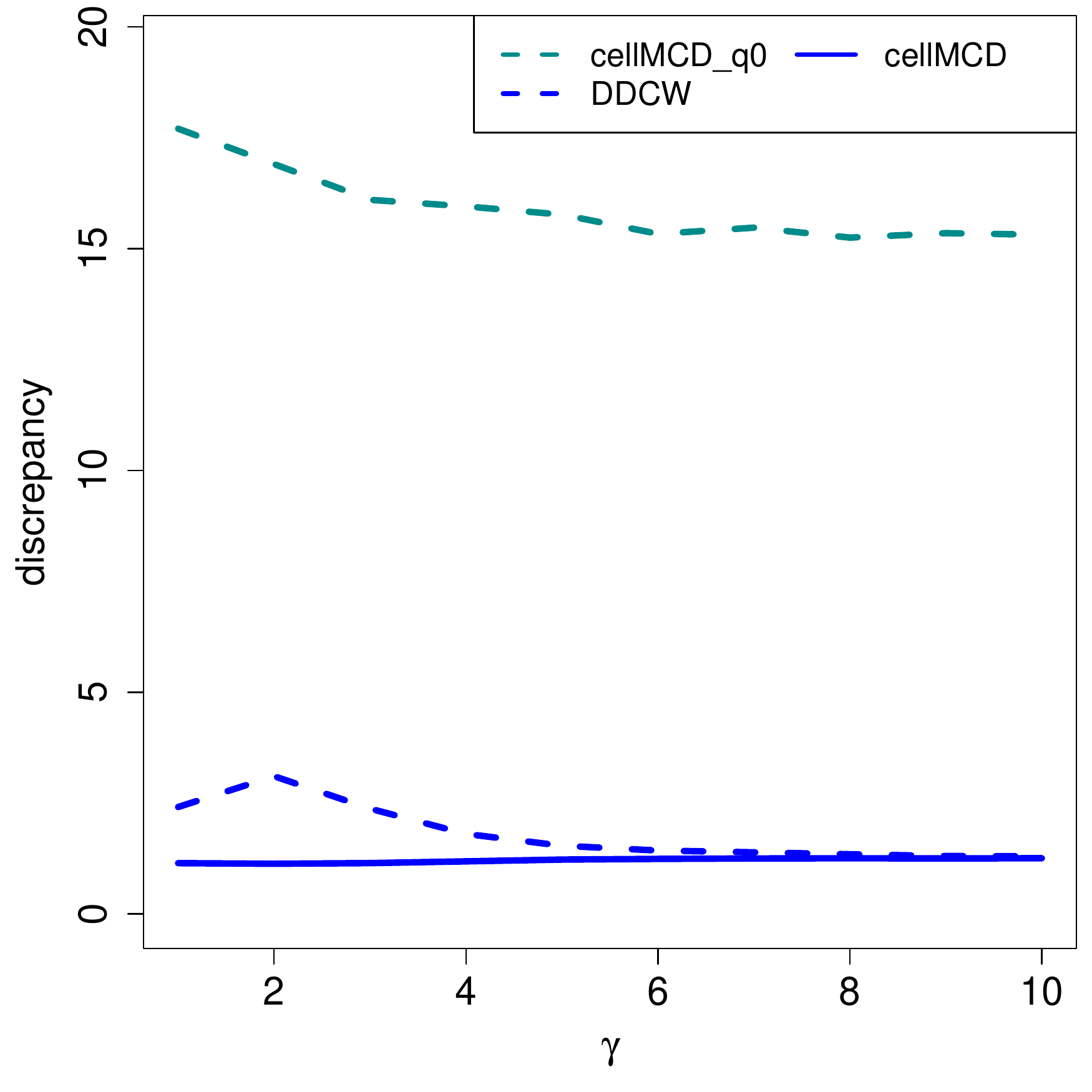}
\includegraphics[width=0.44\textwidth]{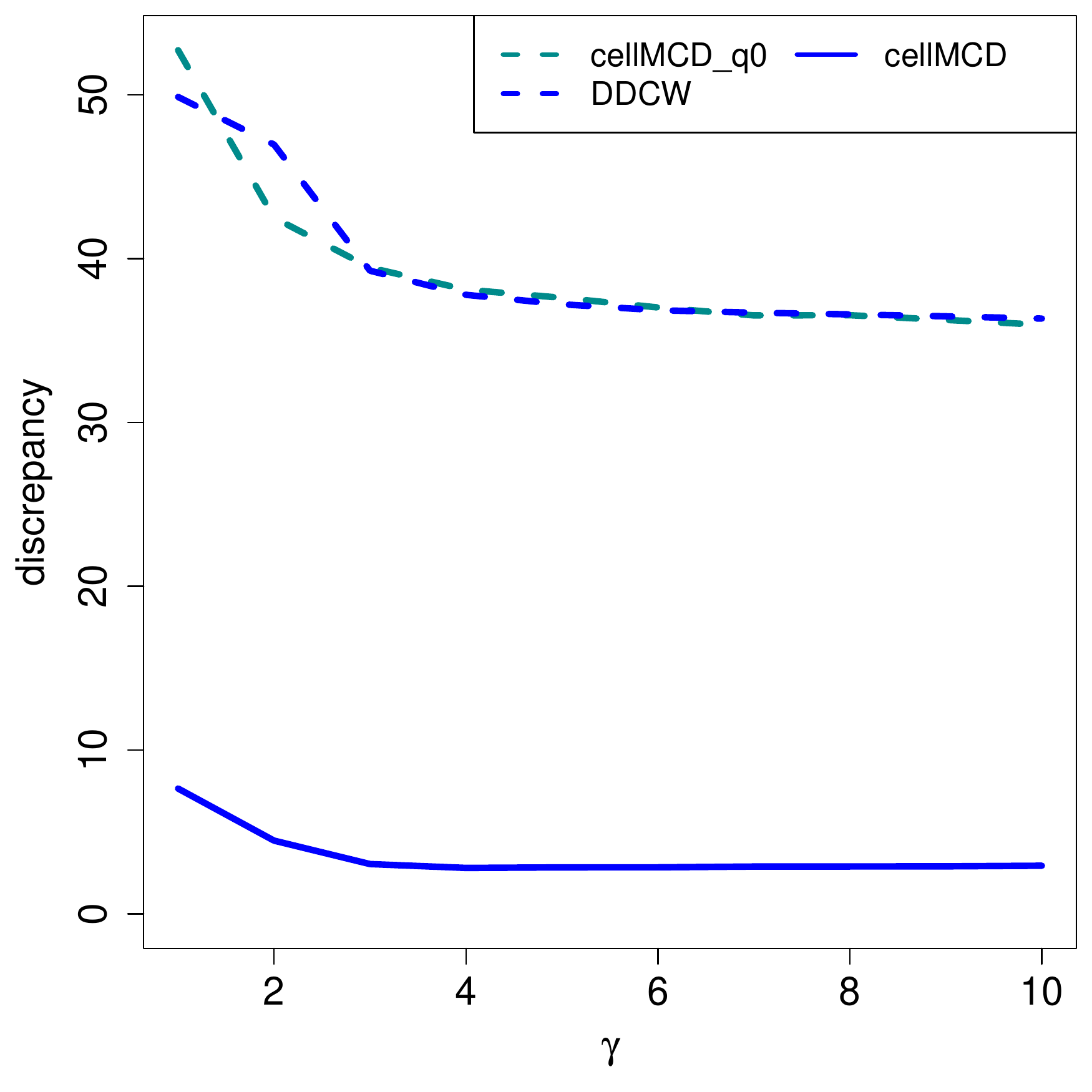}
\includegraphics[width=0.44\textwidth]{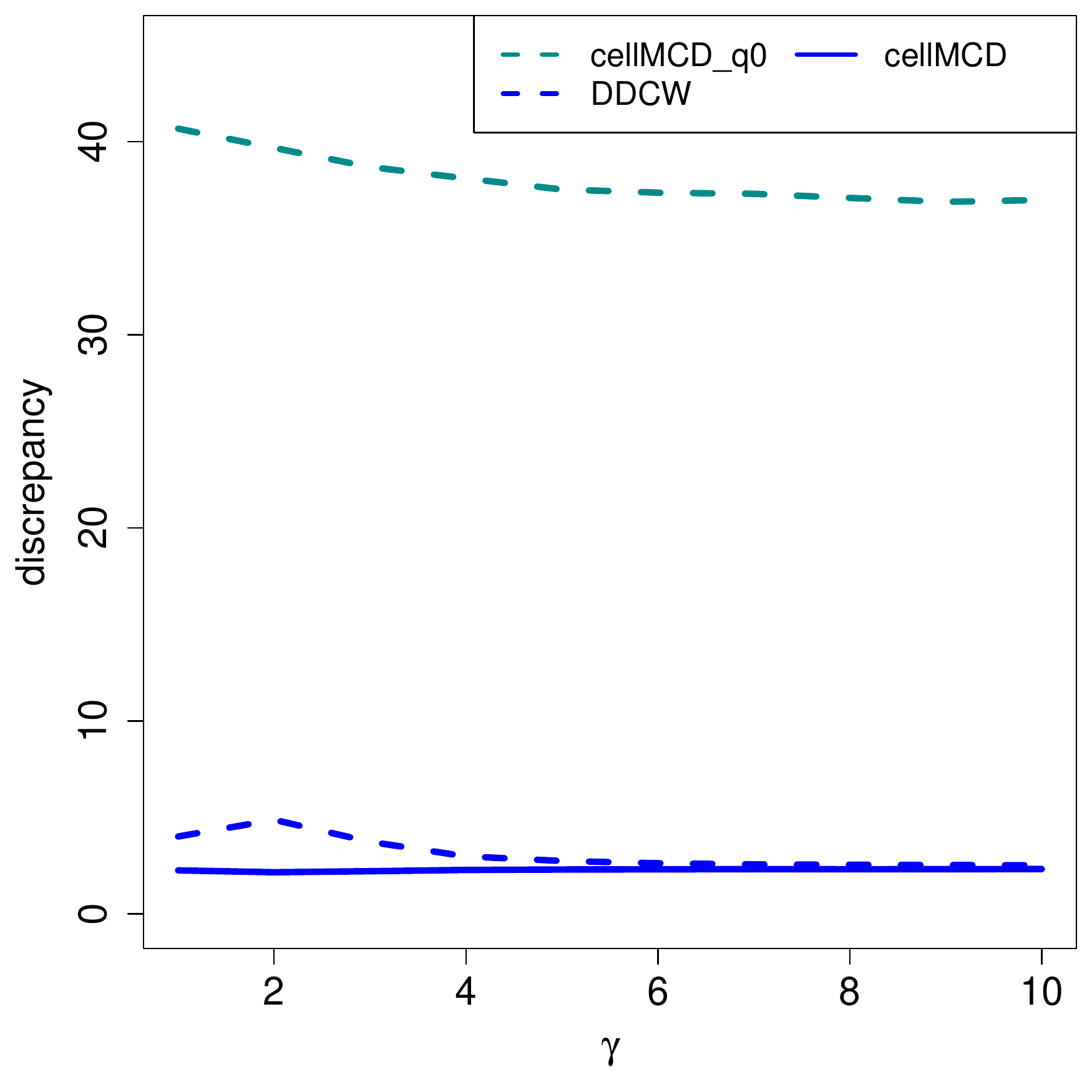}
\vskip-0.4cm
\caption{Kullback-Leibler discrepancy of estimated 
  covariance
  matrices for $n=100$ (top), $n=400$ (middle) and
  $n = 800$ (bottom), for ALYZ (left) and A09 
  (right), with $\varepsilon = 0.1$.}
\label{fig:A_KLdiv_0p1}
\end{figure}

\renewcommand{\refname}{Additional References}

\end{document}